\documentclass[12pt, draftclsnofoot, onecolumn]{IEEEtran}
\ifCLASSINFOpdf
\else
\fi
\usepackage{lipsum}
\usepackage{amsmath,amssymb,amsthm,mathrsfs,amsfonts,dsfont}
\usepackage{subfigure}
\usepackage{graphicx,dblfloatfix}
\usepackage{algorithm,algorithmic}
\usepackage{tabularx}
\usepackage{adjustbox}
\usepackage{color}
\usepackage{colortbl}
\usepackage[noadjust]{cite}

\newcommand\norm[1]{\left\lVert#1\right\rVert}
\usepackage{xcolor}
\newtheorem{theorem}{Theorem}
\newtheorem{proposition}{Proposition}
\newtheorem{corollary}{Corollary}
\newtheorem{remark}{Remark}
\newtheorem{fact}{Fact}

\begin{document}
\title{Optimal SIC Ordering and Power Allocation in Downlink Multi-Cell NOMA Systems}
\author{Sepehr Rezvani, \IEEEmembership{Student~Member,~IEEE}, Eduard A. Jorswieck, \IEEEmembership{Fellow,~IEEE}, Nader Mokari, \IEEEmembership{Senior Member,~IEEE}, and Mohammad R. Javan, \IEEEmembership{Senior Member,~IEEE} 
	\thanks{S. Rezvani and E. A. Jorswieck are with the Department of Information Theory and Communication Systems, Technische Universität Braunschweig, Braunschweig, Germany (e-mails: \{rezvani, jorswieck\}@ifn.ing.tu-bs.de).}
	\thanks{N. Mokari is with the Department of Electrical and
		Computer Engineering, Tarbiat Modares University, Tehran, Iran (e-mail: nader.mokari@modares.ac.ir).}
	\thanks{M. R. Javan is with the Department of Electrical and Robotics Engineering, Shahrood University of Technology, Shahrood, Iran (e-mail: javan@shahroodut.ac.ir).}
}
	
	
\maketitle

\begin{abstract}
	In this work, we propose a globally optimal joint successive interference cancellation (SIC) ordering and power allocation (JSPA) algorithm for the sum-rate maximization problem in downlink multi-cell non-orthogonal multiple access (NOMA) systems. The proposed algorithm is based on the exploration of base stations (BSs) power consumption, and closed-form of optimal powers obtained for each cell. Although the optimal JSPA algorithm scales well with larger number of users, it is still exponential in the number of cells. For any suboptimal decoding order, we propose a low-complexity near-optimal joint rate and power allocation (JRPA) strategy in which the complete rate region of users is exploited. Furthermore, we design a near-optimal semi-centralized JSPA framework for a two-tier heterogeneous network such that it scales well with larger number of small-BSs and users. Numerical results show that JRPA highly outperforms the case that the users are enforced to achieve their channel capacity by imposing the well-known SIC necessary condition on power allocation. Moreover, the proposed semi-centralized JSPA framework significantly outperforms the fully distributed framework, where all the BSs operate in their maximum power budget. Therefore, the centralized JRPA and semi-centralized JSPA algorithms with near-optimal performances are good choices for larger number of cells and users.
\end{abstract}

\IEEEpeerreviewmaketitle

\section{Introduction}
\allowdisplaybreaks
\subsection{Concept of NOMA}
\IEEEPARstart{I}{t} is shown that the channel capacity of degraded broadcast channels (BCs) can be achieved by performing linear superposition coding (SC) in power domain at the transmitter side combined with coherent multiuser detection algorithms, such as successive interference cancellation (SIC), at the receivers side \cite{10.5555/1146355,NIFbook,1683918}. The SC-SIC technique, also called power-domain non-orthogonal multiple access (NOMA)\footnote{In this work, the term 'NOMA' is referred to power-domain NOMA.}, is considered as a candidate radio access technique for the fifth generation (5G) wireless networks and beyond \cite{6692652,7973146,7676258}. In information theory, the main purpose of SC-SIC is reducing the prohibitive complexity of dirty paper coding (DPC) to attain the capacity region of degraded BCs.

\subsection{Single-Cell NOMA}
It is well-known that the downlink single-input single-output (SISO) Gaussian BCs are degraded \cite{10.5555/1146355,NIFbook,1683918}. Hence, NOMA with channel-to-noise ratio (CNR)-based decoding order is capacity-achieving in SISO Gaussian BCs meaning that any rate region is a subset of the rate region of NOMA with CNR-based decoding order \cite{10.5555/1146355,NIFbook,1683918}. The superiority of single-cell NOMA over single-cell OMA is also well-known in information theory \cite{8010756,7676258,7272042}.

In single-cell NOMA, it is verified that the power allocation optimization is necessary to achieve the maximum users sum-rate \cite{7676258,7973146,7272042,9154358,8823873,8010756}. From the optimization perspective, the optimal (CNR-based) decoding order is independent from the power allocation, so is robust and straightforward. Moreover, it is shown that the Hessian of sum-rate function under the CNR-based decoding order is negative definite, so the sum-rate function is strictly concave in powers \cite{7557079,8352643}. In this way, the sum-rate maximization problem in downlink single-cell NOMA is convex\footnote{The feasible region of the general power allocation problem in single-cell NOMA under the minimum rate constraints is affine, so is convex \cite{8352643,7557079}.}. There are some research studies on finding the closed-form of optimal powers for sum-rate maximization problem of $M$-user single-cell NOMA \cite{7557079,8352643}, or its special $2$-user case \cite{7982784}. However, the analysis in \cite{7557079,8352643,7982784} is based on some additional constraints on power allocation among multiplexed users to guarantee successful SIC. In information theory, it is proved that in SISO Gaussian BCs which are degraded, each user can achieve its channel capacity after SIC independent from power allocation among multiplexed users \cite{NIFbook}. Hence, the SIC at users in SISO Gaussian BCs can be successfully proceed independent from power allocation among multiplexed users \cite{8823873}. To this end, the Karush-Kuhn-Tucker (KKT) optimality conditions analysis as well as finding the closed-form of optimal powers for the general $M$-user single-cell NOMA are still open problems.

\subsection{Multi-Cell NOMA with Single-Cell Processing}
Unfortunately, the capacity-achieving schemes are still unknown in downlink multi-cell networks, since the capacity region of the two-user downlink interference channel is still unknown in general \cite{10.5555/1146355,NIFbook,8010756}. For the case that the information of each user is available at only one transmitter, i.e., multi-cell networks with single-cell processing, the signals from neighboring cells are fully treated as additive white Gaussian noise (AWGN) at each user, also called inter-cell interference (ICI). Inspired by the degradation of SISO Gaussian BCs, SC-SIC achieves the channel capacity of each cell in multi-cell systems with single-cell processing \cite{10.5555/1146355,NIFbook}. In this system, 'ICI+AWGN' can be viewed as equivalent noise power at the users. Therefore, NOMA with channel-to-interference-plus-noise ratio (CINR)-based decoding order is capacity-achieving in each cell \cite{10.5555/1146355,NIFbook}. In contrast to single-cell NOMA, finding optimal decoding order in multi-cell NOMA is challenging, because of the impact of ICI on the CINR of multiplexed users. The ICI at users in each cell is affected by the total power consumption of each neighboring (interfering) base station (BS). Therefore, the optimal SIC decoding order in each cell depends on the optimal power control among interfering cells. In this way, the optimal joint SIC ordering and power allocation (JSPA) problem in downlink multi-cell NOMA consists of two components as 1) Power control among interfering cells; 2) optimal joint SIC ordering and power allocation among multiplexed users within each cell. Note that the optimal SIC ordering in case 2 follows the CINR of users known from information theory \cite{10.5555/1146355,NIFbook}.

It is verified that under the CINR-based decoding order, the ICI in the centralized total power minimization problem verifies the basic properties of the standard interference function \cite{7964738,8848606,8114362,8478339}. Hence, the optimal JSPA can be obtained by using the well-known Yates power control framework \cite{414651}. In other words, the globally optimal JSPA for the total power minimization problem in multi-cell NOMA can be found in an iterative distributed manner with a fast convergence speed. However, the ICI in the centralized sum-rate maximization problem does not verify the basic properties of the standard interference function. As a result, Yates power control framework does not guarantee any global optimality for the sum-rate maximization problem \cite{8114362}. 
It is shown that the sum-rate function in multi-cell NOMA is nonconcave in powers, due to existing ICI, which makes the centralized sum-rate maximization problem nonconvex and strongly NP-hard \cite{7812683,7954630,8114362}. The best existing solution for solving the power allocation problem is monotonic optimization which is still approximately exponential in the number of users \cite{7812683,7862919}. The JSPA needs to examine the monotonic-based power allocation ${(M!)}^B$ times, where ${(M!)}^B$ is the total number of possible decoding orders in $B$ cells each having $M$ users. Therefore, the joint optimization via the monotonic-based power allocation \cite{7862919} is basically impractical even at lower number of BSs and users. 
Generally speaking, the JSPA sum-rate maximization problem in multi-cell NOMA falls into the following two schemes:
\begin{enumerate}
	\item \textbf{Optimal SIC ordering (JSPA scheme):} In this scheme, we dynamically update the decoding order of users during power control among the interfering cells. This scheme achieves the globally optimal JSPA, however, it is not yet addressed in the literature.
	\item \textbf{Fixed SIC ordering:} In this scheme, we consider a fixed decoding order (for instance, CNR-based decoding order) during power control among the cells, and find optimal power allocation. This scheme is suboptimal if and only if the fixed decoding order is not guaranteed to be optimal at any feasible region of power allocation. For any fixed, and possibly suboptimal, decoding order, there are two schemes in which we can guarantee successful SIC at users:
	\\
	1) \textit{\textbf{Fixed-rate-region power allocation (FRPA)}}: In this scheme, we limit the power consumption of interfering cells such that the CINR of user with higher decoding order remains larger than that of user with lower decoding order. In other words, we impose an additional constraint on the power consumption of interfering cells, known as SIC necessary condition \cite{7812683,7954630}, such that the fixed decoding order remains optimal. In this scheme, each user achieves its channel capacity for decoding its own signal after SIC. The monotonic optimization is used in \cite{7812683} to find the optimal power allocation in the FRPA scheme. Due to the exponential complexity of monotonic optimization in the number of BSs and users, a sequential programming is proposed in \cite{7812683,7954630} to achieve a near-optimal solution with polynomial time complexity.
	\\
	2) \textit{\textbf{Joint rate and power allocation (JRPA)}}: The so-called SIC necessary condition \cite{7812683,7954630} is indeed not a necessary condition for successful SIC when the decoding order is not guaranteed to be optimal.
	The SIC necessary condition in FRPA imposes additional limitations on the power consumption of neighboring cells to only guarantee that each user achieves its channel capacity for decoding its desired signal. Therefore, FRPA may degrade the rate region of users depending on how much the so-called SIC necessary condition limits the feasible region of BSs power consumption. In JRPA, the SIC necessary condition is removed, while users are allowed to operate in lower rates than their channel capacity to ensure successful SIC. Therefore, we consider the complete rate region of users, known from information theory \cite{10.5555/1146355,NIFbook}. It is obvious that the rate region of FRPA is a subset of the rate region of JRPA.
	The power allocation problem in the JRPA scheme is not yet addressed in the literature. Therefore, finding an efficient algorithm for the JRPA scheme is still an open problem in multi-cell NOMA.
\end{enumerate}
The work in \cite{8114362} considers the FRPA scheme with CNR-based decoding order, while SIC necessary condition is not applied, which may lead to an infeasible solution. Moreover, the ICI management with efficient power control among the cells is not addressed in \cite{8809359}, since it is assumed that all the BSs operate in their maximum power budget. Subsequently, the challenges of finding optimal decoding order in multi-cell NOMA are not addressed in \cite{8809359}. We show that the power allocation problem in \cite{8809359} falls into our distributed resource allocation framework, where all the BSs operate in their maximum power budgets. In addition to the above mentioned open problems in multi-cell NOMA, it is important to know the performance gap between optimal and CNR-based decoding orders, since adopting the CINR-based decoding order is more challenging. In other words, the performance gap between JSPA and FRPA/JRPA is still unknown. The latter performance gap depends on the suboptimality of CNR-based decoding order. Therefore, another interesting topic is \textit{how to ensure that the CNR-based decoding order is optimal independent from ICI?} Moreover, for the CNR-based decoding order, it is still unknown \textit{how much is the performance of considering the complete rate region instead of imposing SIC necessary condition on power allocation?} Actually, \textit{how much is the performance gap between FRPA and JRPA?}

\subsection{Our Contributions}
In this work, we address the above mentioned open problems for the sum-rate maximization problem in the downlink of general single-carrier multi-cell NOMA system with single-cell processing. Our main contributions are presented as follows:
\begin{itemize}
	\item We prove that at the optimal point, only the NOMA cluster-head user, which has the highest CINR (thus highest decoding order), deserves additional power, while the users with lower decoding order get power to only maintain their individual minimal rate demands. Subsequently, we obtain the closed-form expression of optimal powers among multiplexed users within each cell for any given feasible power consumption of BSs. The optimal value is also formulated in closed form. The analysis shows that the optimal power coefficients among multiplexed users are highly insensitive to the users channel gain, specifically for the high signal-to-interference-plus-noise ratio (SINR) regions.
	\item We propose a globally optimal JSPA algorithm for maximizing users sum-rate under the individual minimum rate demand of users. This algorithm utilizes both the exploration of BSs power consumption and closed-form of optimal powers for single-cell NOMA, so is a good benchmark for larger number of users while small number of BSs.
	\item The optimal JSPA algorithm is modified to find the globally optimal solution of the FRPA problem, resulting in dramatically reducing the complexity of monotonic optimization which is still exponential in the number of users. 
	\item We propose a suboptimal sequential programming-based JRPA algorithm which considers the complete achievable rate region of users when the (potentially suboptimal) decoding order is prefixed. The convergence speed and performance of this algorithm for different initialization methods are investigated.
	\item We propose a sufficient condition which determines the optimality of CNR-based decoding order for a user pair independent from ICI. The applications of this theorem are listed in the paper.
	\item We propose a semi-centralized framework for a two-tier heterogeneous network (HetNet) consisting of multiple femto BSs (FBSs) underlying a single macro BS (MBS), where the ICI from MBS to the FBS users is managed efficiently.
\end{itemize}
The complete source code of the simulations including a user guide is available in \cite{sourcecode}.

\subsection{Paper Organization}
The rest of this paper is presented as follows. Section \ref{sec system model} describes the general multi-cell NOMA system, and formulates the main optimization problem for maximizing users sum-rate. The solution algorithms are presented in Section \ref{sec solution}. Numerical results are provided in Section \ref{sec simulation}. Our conclusions and future research directions are presented in Section \ref{sec conclusion}.

\section{General Downlink Multi-Cell NOMA System}\label{sec system model}
Consider the downlink transmission of a multiuser single-carrier multi-cell NOMA system. The set of single-antenna BSs, and users served by BS $b$ are indicated by $\mathcal{B}$ and $\mathcal{U}_b$, respectively. According to the NOMA protocol, the users associated to the same transmitter form a NOMA cluster. The signal of users associated to other transmitters (known as ICI) is fully treated as AWGN at users within the NOMA cluster. Hence, we consider a single NOMA cluster at each cell $b$ \cite{7964738} including $|\mathcal{U}_b|$ users, where $|.|$ is the cardinality of a finite set. The term $k \to i$ indicates that user $k$ has a higher decoding order than user $i$ such that user $k$ is scheduled (and enforced) to decode and cancel the whole signal of user $i$, while the whole signal of user $k$ is treated as noise at user $i$. For instance, assume that each cell $b$ serves $M$ users, i.e., $|\mathcal{U}_b|=M$. Generally, there are $M!$ possible decoding orders for users within the $M$-order NOMA cluster. Without loss of generality, let $k \to i$ if $k > i$, i.e., the SIC of NOMA in each cell $b$ follows $M \to M-1 \to \dots \to 1$. As shown in Fig. 1 in \cite{7964738}, in this SIC decoding order, the signal of each user $i$ will be decoded prior to user $k>i$. In general, each user $i$ first decodes and cancels the signal of users $1, \dots, i-1$. Then, it decodes its desired signal such that the signal of users $i+1, \dots, M$ is treated as noise at user $i$ \cite{7676258}. In this regard, in each cell, the NOMA cluster-head user $M$ does not experience any intra-NOMA interference (INI).

Let $s_{b,i} \sim \mathcal{CN}(0,1)$ be the desired signal of user $i \in \mathcal{U}_b$.
Denoted by $\lambda_{b,i,k} \in \{0,1\}$, the binary decoding decision indicator, where $\lambda_{b,i,k}=1$ if user $k \in \mathcal{U}_b$ is scheduled to decode (and cancel when $k\neq i$) $s_{b,i}$, and otherwise, $\lambda_{b,i,k}=0$.
Since for each user pair within a NOMA cluster, only one user can decode and cancel the signal of other user, we have $\lambda_{b,i,k}+\lambda_{b,k,i}=1,~\forall b \in \mathcal{B},~ i,k \in \mathcal{U}_b,~k \neq i$. Moreover, the signal of each user should be decoded at that user, meaning that $\lambda_{b,i,i}=1,~\forall b \in \mathcal{B},~ i \in \mathcal{U}_b$. Due to the transitive nature of SIC ordering, if $\lambda_{b,i,k}=1$ and $\lambda_{b,k,h}=1$, then we should have $\lambda_{b,i,h}=1$. In other words, we have $\lambda_{b,i,k} \lambda_{b,k,h} \leq \lambda_{b,i,h},~\forall b \in \mathcal{B}, i,k,h \in \mathcal{U}_b$.
According to the SIC protocol, $s_{b,i}$ should be decoded at user $i \in \mathcal{U}_b$ as well as all the users in\footnote{The term $\Phi_{b,i}$ is the set of users in cell $b$ with higher decoding orders than user $i \in \mathcal{U}_b$.} $\Phi_{b,i}=\{k \in \mathcal{U}_b \setminus \{i\} \mid \lambda_{b,i,k}=1\}$. Therefore, in the SIC of NOMA, each user $i$ first decodes and subtracts each signal $s_{b,j},~\forall j \in \mathcal{U}_b \setminus \{ \{i\} \cup \Phi_{b,i}\}$, then it decodes its desired signal $s_{b,i}$ such that the signal of users in $\Phi_{b,i}$ is treated as noise (called INI). According to the SIC protocol, the signal of user $i$ will be decoded prior to user $k$ if $|\Phi_{b,i}| > |\Phi_{b,k}|$. According to the above, the SIC decoding order among users can be determined by finding $\lambda_{b,i,k}$. Actually, $\lambda_{b,i,k}=1,~\forall b \in \mathcal{B},~ i,k \in \mathcal{U}_b,~k \neq i$ is equivalent to $k \to i$ in cell $b$.
Similar to the related works, we assume that the perfect channel state information (CSI) of all the users is available at the scheduler. The channel gain from BS $j \in \mathcal{B}$ to user $i \in \mathcal{U}_b$ is denoted by $g_{j,b,i}$. The allocated power from BS $b$ to user $i \in \mathcal{U}_b$ is denoted by $p_{b,i}$.
After performing (perfect) SIC at each user $l \in \mathcal{U}_b \setminus \Phi_{b,i}$, the received signal of user $i \in \mathcal{U}_b$ at user $k \in \{i\} \cup \Phi_{b,i}$ is given by
\begin{equation}\label{rec signal}
	y_{b,i,k} = \underbrace{\sqrt{p_{b,i}} g_{b,b,k} s_{b,i}}_\text{intended signal} + \underbrace{\sum\limits_{j \in \Phi_{b,i}} \sqrt{p_{b,j}} g_{b,b,k} s_{b,j}}_\text{INI} + \underbrace{\sum\limits_{\hfill j \in \mathcal{B} \hfill\atop  j \neq b} \sum\limits_{l \in \mathcal{U}_j} \sqrt{p_{j,l}} g_{j,b,k} s_{j,l}}_\text{ICI} + N_{b,k},
\end{equation}
where the first and second terms are the received desired signal and INI of user $i \in \mathcal{U}_b$ at user $k \in \mathcal{U}_b$, respectively. The third term represents the ICI at user $k \in \mathcal{U}_b$. Moreover, $N_{b,k}$ is the AWGN at user $k \in \mathcal{U}_b$ with zero mean and variance $\sigma^2_{b,k}$. Without loss of generality, assume that $|s_{b,i}|=1,~\forall b \in \mathcal{B},~i \in \mathcal{U}_b$, and $h_{j,b,k}=|g_{j,b,k}|^2$ \cite{7964738,8861078,8114362}. According to \eqref{rec signal}, the SINR of user $k \in \Phi_{b,i}$ for decoding and canceling the signal of user $i \in \mathcal{U}_b$ is $\gamma_{b,i,k} = \frac{p_{b,i} h_{b,b,k}}{\sum\limits_{j \in \Phi_{b,i}} p_{b,j} h_{b,b,k} + (I_{b,k} + \sigma^2_{b,k})}$, where $\sum\limits_{j \in \Phi_{b,i}} p_{b,j} h_{b,b,k}$ is the INI power of user $i \in \mathcal{U}_b$ received at user $k \in \mathcal{U}_b$, and $I_{b,k}=\sum\limits_{\hfill j \in \mathcal{B} \hfill\atop  j \neq b} \sum\limits_{l \in \mathcal{U}_j} p_{j,l} h_{j,b,k}$ is the received ICI power at user $k \in \mathcal{U}_b$. For the case that $k=i$, $\gamma_{b,i,i}$ denotes the SINR of user $i \in \mathcal{U}_b$ for decoding its desired signal $s_{b,i}$ after SIC. For convenience, let $h_{b,b,i} \equiv h_{b,i}$, and $\gamma_{b,i,i} \equiv \gamma_{b,i}$. Furthermore, let us denote the matrix of all the decoding indicators by $\boldsymbol{\lambda}=[\lambda_{b,i,k}],\forall b \in \mathcal{B}, i,k \in \mathcal{U}_b$, in which $\boldsymbol{\lambda}_{b}=[\lambda_{b,i,k}],\forall i,k \in \mathcal{U}_b$, represents the decoding indicator matrix of users in $\mathcal{U}_b$. Moreover, $\boldsymbol{p}=[p_{b,i}],\forall b \in \mathcal{B}, i \in \mathcal{U}_b$, is the power allocation matrix of all the users, in which $\boldsymbol{p}_{b}$ is the $b$-th row of this matrix indicating the power allocation vector of users in cell $b$.
According to the Shannon's capacity formula, the achievable spectral efficiency of user $i \in \mathcal{U}_b$ after successful SIC is obtained by \cite{10.5555/1146355,NIFbook}
\begin{equation}\label{useri Mcell}
	R_{b,i} (\boldsymbol{p},\boldsymbol{\lambda}_b) = \min\limits_{k \in \{i\} \cup \Phi_{b,i}}\left\{\log_2\left( 1+ \frac{p_{b,i} h_{b,k}}{\sum\limits_{j \in \Phi_{b,i}} p_{b,j} h_{b,k} + (I_{b,k} (\boldsymbol{p}_{-b}) + \sigma^2_{b,k})} \right)\right\}.
\end{equation}
Note that the set $\Phi_{b,i}$ depends on $\boldsymbol{\lambda}_b$, although it is not explicitly shown in \eqref{useri Mcell}.
The centralized total spectral efficiency maximization problem is formulated by
\begin{subequations}\label{centg problem}
	\begin{align}\label{obf centg problem}
		\text{JSPA}:~\max_{ \boldsymbol{p} \geq 0,~\boldsymbol{\lambda} \in \{0,1\} }\hspace{.0 cm}	
		~~ & \sum\limits_{b \in \mathcal{B}} \sum\limits_{i \in \mathcal{U}_b} R_{b,i}(\boldsymbol{p},\boldsymbol{\lambda}_b)
		\\
		\text{s.t.}~~~~~~\label{Constraint max power}
		& \sum\limits_{i \in \mathcal{U}_b} p_{b,i} \leq P^{\text{max}}_b,~\forall b \in \mathcal{B},
		\\
		\label{Constraint QoS}
		& R_{b,i} (\boldsymbol{p},\boldsymbol{\lambda}_b) \geq R^{\text{min}}_{b,i},~\forall b \in \mathcal{B},~i \in \mathcal{U}_b,
		\\
		\label{Constraint uniqe order}
		& \lambda_{b,i,k} + \lambda_{b,k,i}=1,~\forall b \in \mathcal{B},~i,k \in \mathcal{U}_b,~k \neq i,
		\\
		\label{Constraint transitive order}
		& \lambda_{b,i,k} \lambda_{b,k,h} \leq \lambda_{b,i,h},~\forall b \in \mathcal{B}, i,k,h \in \mathcal{U}_b,
		\\
		\label{Constraint own decoding}
		& \lambda_{b,i,i} =1,~\forall b \in \mathcal{B}, i,\in \mathcal{U}_b,
	\end{align}
\end{subequations}
where \eqref{Constraint max power} and \eqref{Constraint QoS} are the per-BS maximum power and per-user minimum rate constraints, respectively.
$P^\text{max}_b$ denotes the maximum power of BS $b$, and $R^{\text{min}}_{b,i}$ is the minimum spectral efficiency demand of user $i \in \mathcal{U}_b$. The rest of the constraints are described above.

\section{Solution Algorithms for the Sum-Rate Maximization Problem}\label{sec solution}
In this section, we propose globally optimal and suboptimal solutions for the main problem \eqref{centg problem} under the centralized/decentralized resource management frameworks. Finally, we compare the computational complexity of the proposed resource allocation algorithms.

\subsection{Centralized Resource Management Framework}\label{subsection centralized}
In this subsection, we first propose a globally optimal JSPA algorithm for problem \eqref{centg problem}. Then, by considering a fixed SIC decoding order in \eqref{centg problem}, we propose two suboptimal rate adoption and power allocation algorithms.

\subsubsection{Globally Optimal JSPA Algorithm}\label{subsection optimal}
Problem \eqref{centg problem} can be classified as a mixed-integer non-linear
programming (MINLP) problem. 
\begin{corollary}\label{corol NPhard}
	Problem \eqref{centg problem} is strongly NP-hard for $|\mathcal{U}_b|=1,~\forall b \in \mathcal{B}$.
\end{corollary}
\begin{proof}
	Consider a single user within each cell, i.e., $|\mathcal{U}_b|=1,~\forall b \in \mathcal{B}$. In this way, \eqref{useri Mcell} can be rewritten as
	$R_{b,1} (\boldsymbol{p}) = \log_2\left( 1+ \frac{p_{b,1} h_{b,1}}{I_{b,1} (\boldsymbol{p}_{-b}) + \sigma^2_{b,1}} \right),~\forall b \in \mathcal{B}$.
	Moreover, problem \eqref{centg problem} is rewritten as the following power allocation problem
	\begin{equation*}\label{centg NP1}
		\max_{\boldsymbol{p} \geq 0}~\sum\limits_{b \in \mathcal{B}} R_{b,1}(\boldsymbol{p})~~~~~~\text{s.t.}~\eqref{Constraint max power},~\eqref{Constraint QoS}.
	\end{equation*}
	The latter problem is equivalent to problem $(P_1)'$ in \cite{4453890}. In Theorem 1 in \cite{4453890}, it is proved that the sum-rate maximization problem $(P_1)'$ is strongly NP-hard, due to the nonconcavity of sum-rate function in $\boldsymbol{p}$. In other words, for $|\mathcal{U}_b|=1,~\forall b \in \mathcal{B}$, multi-cell NOMA with $B$ cells is identical to the point-to-point interference-limited network with $B$ transmitter-receiver pairs.
\end{proof}
Corollary \ref{corol NPhard} can be generalized to $|\mathcal{U}_b|>1$ for some $b \in \mathcal{B}$. However, the complete proof is more difficult, and we consider it as future work.
Let us define the power consumption coefficient of BS $b$ as $\alpha_b \in [0,1]$ such that $\sum\limits_{i \in \mathcal{U}_b} p_{b,i} = \alpha_b P^{\text{max}}_b$. 
The received ICI power at user $i \in \mathcal{U}_b$ can be reformulated by $I_{b,i} (\boldsymbol{\alpha}_{-b})=\sum\limits_{\hfill j \in \mathcal{B} \hfill\atop  j \neq b} \alpha_j P^{\text{max}}_j h_{j,b,i}$. Let $\boldsymbol{\alpha}=[\alpha_b]_{1 \times B}$ be the BSs power consumption coefficient vector. Although the joint optimization of $\boldsymbol{\alpha}$, $\boldsymbol{p}$, and $\boldsymbol{\lambda}$ is very challenging, for any given $\boldsymbol{\alpha}$, the optimal $(\boldsymbol{\lambda}^*_b,\boldsymbol{p}^*_b)$ can be obtained in closed form as follows:
\begin{corollary}\label{corollary optorder}
	In multi-cell NOMA, the optimal decoding order for the user pair $i,k \in \mathcal{U}_b$ is $k \to i$ if and only if
	$\tilde{h}_{b,i}(\boldsymbol{\alpha}_{-b})\leq \tilde{h}_{b,k}(\boldsymbol{\alpha}_{-b})$, where $\tilde{h}_{b,l}(\boldsymbol{\alpha}_{-b})=\frac{h_{b,l}}{I_{b,l}(\boldsymbol{\alpha}_{-b}) + \sigma^2_{b,l}},~l=i,k$ \cite{10.5555/1146355,NIFbook}. Therefore, $\lambda^*_{b,i,k}=1$ if and only if $\tilde{h}_{b,i}(\boldsymbol{\alpha}_{-b})\leq \tilde{h}_{b,k}(\boldsymbol{\alpha}_{-b})$.
\end{corollary}
The additional notes on the optimal decoding order can be found in Appendix \ref{appendix corol optorder}.
\begin{corollary}\label{corol optorder 3}
	For any given $\boldsymbol{\alpha}_{-b}$ and subsequently $\tilde{h}_{b,i}(\boldsymbol{\alpha}_{-b})\leq \tilde{h}_{b,k}(\boldsymbol{\alpha}_{-b})$, at the optimal $\boldsymbol{\lambda}^*_b$, we have
	\begin{equation}\label{SIC nec cond}
		\log_2\left( 1+ \frac{p_{b,i} h_{b,i}}{\sum\limits_{j \in \Phi^*_{b,i}} p_{b,j} h_{b,i} + (I_{b,i} + \sigma^2_{b,i})} \right) \leq \log_2\left( 1+ \frac{p_{b,i} h_{b,k}}{\sum\limits_{j \in \Phi^*_{b,i}} p_{b,j} h_{b,k} + (I_{b,k} + \sigma^2_{b,k})} \right).
	\end{equation}
	According to \eqref{useri Mcell} and \eqref{SIC nec cond}, at the optimal $\boldsymbol{\lambda}^*_b$, we have $R_{b,i}(\boldsymbol{p},\boldsymbol{\lambda}^*_b) = \log_2\left( 1+ \frac{p_{b,i} h_{b,i}}{\sum\limits_{j \in \Phi^*_{b,i}} p_{b,j} h_{b,i} + (I_{b,i} + \sigma^2_{b,i})} \right)$.
\end{corollary}
According to Corollary \ref{corollary optorder}, the optimal SIC ordering ($\boldsymbol{\lambda}^*_b$) and power control of interfering cells ($\boldsymbol{\alpha}^*_{-b}$) cannot be decoupled in general. However, for any given $\boldsymbol{\alpha}$, and subsequently $\boldsymbol{\lambda}^*_b$, the optimal power $\boldsymbol{p}^*_b$ can be obtained in closed form as follows:
\begin{proposition}\label{Propos optpower}
	Assume that $\boldsymbol{\alpha}$ is fixed. For convenience, let $|\mathcal{U}_b|=M$, and $k>i$ if $\tilde{h}_{b,k}(\boldsymbol{\alpha}_{-b}) > \tilde{h}_{b,i}(\boldsymbol{\alpha}_{-b})$. According to Corollary \ref{corollary optorder}, the decoding order $M \to M-1 \to \dots \to 1$ is optimal. The optimal powers in $\boldsymbol{p}^*_b$ can be obtained in closed form as follows:
	\begin{equation}\label{optpow i}
		p^*_{b,i}=\left[\beta_{b,i} \left( \prod\limits_{j=1}^{i-1} \left(1-\beta_{b,j}\right) \alpha_b P^{\text{max}}_b
		+\frac{1}{\tilde{h}_{b,i}}-
		\sum\limits_{j=1}^{i-1} \frac{\prod\limits_{k=j+1}^{i-1} \left(1-\beta_{b,k}\right) \beta_{b,j}} {\tilde{h}_{b,j}} \right)\right]^+
		,\forall i=1,\dots,M-1,
	\end{equation}
	and
	\begin{equation}\label{optpow M}
		p^*_{b,M}=\left[\alpha_b P^{\text{max}}_b - \sum\limits_{i=1}^{M-1} \beta_{b,i} \left( \prod\limits_{j=1}^{i-1} \left(1-\beta_{b,j}\right) \alpha_b P^{\text{max}}_b +\frac{1}{\tilde{h}_{b,i}}-
		\sum\limits_{j=1}^{i-1} \frac{\prod\limits_{k=j+1}^{i-1} \left(1-\beta_{b,k}\right) \beta_{b,j}} {\tilde{h}_{b,j}} \right)\right]^+,
	\end{equation}
	where $\beta_{b,i} = \frac{2^{R^{\text{min}}_{b,i}}-1} {2^{R^{\text{min}}_{b,i}}}, ~\forall i=1,\dots,M-1$, and $\left[.\right]^+=\max \{.,0\}$.
\end{proposition}
\begin{proof}
	Please see Appendix \ref{appendix optpower}.
\end{proof}
\begin{remark}\label{remark approx power}
	For sufficiently large normalized channel gains $\tilde{h}_{b,i},\forall i \in \mathcal{U}_b\setminus\{M\}$, \eqref{optpow i} and \eqref{optpow M} can be approximated to
	\begin{equation}\label{optpow i approx}
		p^*_{b,i} \approx \left[\alpha_b P^{\text{max}}_b \beta_{b,i} \left( \prod\limits_{j=1}^{i-1} \left(1-\beta_{b,j}\right) \right)\right]^+
		,\forall i=1,\dots,M-1,
	\end{equation}
	and
	\begin{equation}\label{optpow M approx}
		p^*_{b,M} \approx \left[\alpha_b P^{\text{max}}_b \left(1 - \sum\limits_{i=1}^{M-1} \beta_{b,i} \left( \prod\limits_{j=1}^{i-1} \left(1-\beta_{b,j}\right) \right) \right)\right]^+,
	\end{equation}
	respectively.
\end{remark}
\begin{remark}\label{remark approx power sameminrate}
	For sufficiently large normalized channel gains $\tilde{h}_{b,i},\forall i \in \mathcal{U}_b\setminus\{M\}$, if users have the same minimum rate demands denoted by $R_\text{min}$ in cell $b$, the optimal power coefficient of each user $i \in \mathcal{U}_b$, denoted by $q^*_{b,i}=\frac{p^*_{b,i}}{\alpha_b P^{\text{max}}_b}$, based on \eqref{optpow i approx} and \eqref{optpow M approx} can be obtained by
	\begin{equation}\label{power sameminrate}
		q^*_{b,i} \approx \frac{2^{R_\text{min}}-1}{ \left(2^{R_\text{min}}\right)^i },~\forall i=1,\dots,M-1,~~~~~~~~~~q^*_{b,M} \approx \frac{1}{\left(2^{R_\text{min}}\right)^{M-1}}.
	\end{equation}
\end{remark}
According to the above,
\begin{itemize}
	\item Remark \ref{remark approx power} shows that in the high SINR regions, the optimal powers are approximately insensitive to the channel gains. Hence, \eqref{optpow i approx} and \eqref{optpow M approx} are valid for the fast fading and/or imperfect CSI scenarios, where the CSI variations are small such that the optimal decoding order remains constant.
	\item Remark \ref{remark approx power} shows that the weaker user may get less power than the stronger user, depending on its minimum rate demand. For $M=2$, Remark \ref{remark approx power sameminrate} shows that the weaker user gets less power when its minimum rate demand is less than $1$ bps/Hz.
\end{itemize}
The performance of these approximations is numerically evaluated in Subsection \ref{subsection approx opt power}.

According to Proposition \ref{Propos optpower}, for any given feasible $\boldsymbol{\alpha}$, the optimal powers can be obtained in closed form. Therefore, the globally optimal $\boldsymbol{\alpha}^*$ should satisfy the necessary and sufficient condition $\sum\limits_{b \in \mathcal{B}} R^\text{tot}_b(\boldsymbol{\alpha}^*) \geq \sum\limits_{b \in \mathcal{B}} R^\text{tot}_b(\boldsymbol{\alpha})$, for any feasible $\boldsymbol{\alpha}$. Based on Proposition \ref{Propos optpower}, for any given feasible $\boldsymbol{\alpha}$, the optimal value in cell $b$ can be calculated in closed form as $R^\text{tot}_b(\boldsymbol{\alpha})=\sum\limits_{i=1}^{M-1} R^{\text{min}}_{b,i} + \log_2\left( 1+p^*_{b,M}(\boldsymbol{\alpha}) \tilde{h}_{b,M} \right)$, where $p^*_{b,M}(\boldsymbol{\alpha})$ is obtained by \eqref{optpow M}. Hence, the globally optimal solution can be obtained by performing a greedy search on $\boldsymbol{\alpha}$, and comparing the optimal value $\sum\limits_{b \in \mathcal{B}} R^\text{tot}_b(\boldsymbol{\alpha})$ over all the possible values of $\boldsymbol{\alpha}$. Since the number of samples $S_\alpha$ for each $\alpha_b \in \boldsymbol{\alpha}_{1 \times B}$ is finite, the greedy search on $\boldsymbol{\alpha}$ falls into an $\epsilon$-suboptimal solution such that when $S_\alpha \to \infty$, then $\epsilon \to 0$. Our proposed globally optimal, or more precisely, $\epsilon$-suboptimal JSPA algorithm utilizes both the exploration of different values of $\alpha_b \in \boldsymbol{\alpha}_{1 \times B}$ and the distributed power allocation optimization (calculating sum-rate $\sum\limits_{b \in \mathcal{B}} R^\text{tot}_b(\boldsymbol{\alpha})$). The pseudo code of the proposed JSPA method is presented in Alg. \ref{Alg global}.
\begin{algorithm}[tp]
	\caption{Optimal JSPA for Sum-Rate Maximization Problem.} \label{Alg global}
	\begin{algorithmic}[1]
		\STATE Initialize the step size $\epsilon_\alpha \ll 1$, and $R_\text{tot}=0$.
		\FOR {each sample $\boldsymbol{\hat{\alpha}}$}
		\STATE Update $\tilde{h}_{b,i}=\frac{h_{b,i}}{\hat{I}_{b,i} + \sigma_{b,i}},~\forall b \in \mathcal{B},~i \in \mathcal{U}_b$, where 
		$\hat{I}_{b,i}=\sum\limits_{\hfill j \in \mathcal{B} \hfill\atop  j \neq b} \hat{\alpha}_j P^{\text{max}}_j h_{j,b,i}$.
		\STATE Update $\boldsymbol{\lambda}$ according to $\lambda_{b,i,k}=1$ if $\tilde{h}_{b,k} > \tilde{h}_{b,i}$, or equivalently update users index according to $k > i$ if $\tilde{h}_{b,k} > \tilde{h}_{b,i}$.
		\STATE Find $\boldsymbol{p}$ according to \eqref{optpow i} and \eqref{optpow M}.
		\\
		\STATE\textbf{if}~~$\left(\sum\limits_{b \in \mathcal{B}} \sum\limits_{i \in \mathcal{U}_b} R_{b,i}(\boldsymbol{p},\boldsymbol{\lambda}_b) > R_\text{tot}\right)$~\textbf{then}
		\\~~~~~Update $R^*_\text{tot}=\left(\sum\limits_{b \in \mathcal{B}} \sum\limits_{i \in \mathcal{U}_b} R_{b,i}(\boldsymbol{p},\boldsymbol{\lambda}_b)\right)$, $\boldsymbol{p}^*=\boldsymbol{p}$, and $\boldsymbol{\lambda}^*=\boldsymbol{\lambda}$.
		\\\textbf{end if}
		\ENDFOR
		\STATE The outputs $\boldsymbol{\lambda}^{*}$ and $\boldsymbol{p}^{*}$ are the optimal solutions.
	\end{algorithmic}
\end{algorithm}
This algorithm needs to explore all the possible values in $\boldsymbol{\alpha}_{1 \times B}$. For the total number of samples $S_\alpha$ for each $\alpha_b$, the complexity of Alg. \ref{Alg global} is $S^B_\alpha$. For the case that each cell has $M$ users, the complexity of exhaustive search is $S^{BM}_p \times (M!)^B$, where $S_p$ is the total number of samples for each $p_{b,i} \in \boldsymbol{p}$. Hence, Alg. \ref{Alg global} reduces the complexity of exhaustive search by a factor of $S^{M}_p \times (M!)^B$ when $S_\alpha=S_p$. In fact, Alg. \ref{Alg global} has two main advantages: 1) The complexity is independent from the number of users, since for any given $\boldsymbol{\alpha}$, the optimal value $\sum\limits_{b \in \mathcal{B}} R^\text{tot}_b(\boldsymbol{\alpha})$ is formulated in closed form; 2) The complexity of finding optimal SIC ordering is negligible since for a fixed $\boldsymbol{\alpha}$, the optimal decoding order is obtained based on Corollary \ref{corollary optorder}.

Since Alg. \ref{Alg global} is still exponential in the number of BSs, it is important to check the feasibility of problem \eqref{centg problem} with a low-complexity algorithm before performing Alg. \ref{Alg global}. The feasibility problem of \eqref{centg problem} can be formulated by
\begin{equation}\label{feasible problem}
\min_{ \boldsymbol{p} \geq 0,~\boldsymbol{\lambda} \in \{0,1\} }\hspace{.0 cm}	
f(\boldsymbol{p},\boldsymbol{\lambda})~~~~~~~\text{s.t.}~\eqref{Constraint max power}\text{-}\eqref{Constraint own decoding},
\end{equation}
where $f(\boldsymbol{p},\boldsymbol{\lambda})$ can be any objective function such that the intersection of the feasible domain of \eqref{Constraint max power}\text{-}\eqref{Constraint own decoding} is a subset of the feasible domain of $f(\boldsymbol{p},\boldsymbol{\lambda})$. Finding a feasible solution for \eqref{feasible problem} is challenging, due to the binary variables in $\boldsymbol{\lambda}$. 
For the case that $f(\boldsymbol{p})=\sum\limits_{b \in \mathcal{B}} \sum\limits_{i \in \mathcal{U}_b} p_{b,i}$, the well-known iterative distributed power minimization framework can globally solve \eqref{feasible problem} with a fast convergence speed \cite{7964738}.
Here, we briefly present the structure of the iterative distributed power minimization algorithm. 
Let $\boldsymbol{p}^{(t-1)}$ be the output of iteration $t-1$ which is the initial power matrix for iteration $t$ denoted by $\boldsymbol{p}^{(t)}$.
At iteration $t$, we first update the optimal decoding order in cell $1$ based on the updated ICI (according to Corollary \ref{corollary optorder}).
Then, we find $\boldsymbol{p^*}_{1}^{(t)}$ for the fixed $\boldsymbol{p}^{(t)}_{-1}$. After that, we substitute $\boldsymbol{p}^{(t)}_{1}$ with $\boldsymbol{p^*_{1}}^{(t)}$. The updated $\boldsymbol{p}^{(t)}$ is an initial power matrix for the next cell $2$. We repeat these steps at all the remaining cells. The updated power at cell $B$ is the output of iteration $t$. We continue the iterations until the convergence is achieved. For the fixed $I_{b,i}$ and subsequently given $\boldsymbol{\lambda}^*_{b}$ in cell $b$, the optimal powers can be obtained in closed form as
\begin{equation}\label{opt powmin}
p^*_{b,i}=\beta_{b,i} \left(\prod\limits_{j=i+1}^{M} \left(1+\beta_{b,j}\right) +\frac{1}{\tilde{h}_{b,i}}+
\sum\limits_{j=i+1}^{M} \frac{ \beta_{b,j} \prod\limits_{k=i+1}^{j-1} \left(1+\beta_{b,k}\right)}{\tilde{h}_{b,j}}\right),~\forall i=1,\dots,M,
\end{equation}
where $\beta_{b,i}=2^{R^{\text{min}}_{b,i}}-1, ~\forall i=1,\dots,M$, and $\tilde{h}_{b,l}=\frac{h_{b,l}}{I_{b,l} + \sigma_{b,l}}$. For convenience, we assumed that $|\mathcal{U}_b|=M$ and also updated the users index based on the ascending order of $\tilde{h}_{b,i}$ in \eqref{opt powmin}. For more details, please see Appendix \ref{appendix optpower powmin}.
The pseudo code of the proposed power minimization algorithm is presented in Alg. \ref{Alg opt Mcell powmin}.
\begin{algorithm}[tp]
	\caption{Optimal Joint SIC Ordering and Power Allocation for Total Power Minimization Problem.} \label{Alg opt Mcell powmin}
	\begin{algorithmic}[1]
		\STATE Initialize feasible $\boldsymbol{p}^{(0)}$, and tolerance $\epsilon_\text{tol}$ (sufficiently small).
		\WHILE {$\norm{\boldsymbol{p}^{(t-1)}}- \norm{\boldsymbol{p}^{(t)}}> \epsilon_\text{tol}$}
		\STATE Set $t=:t+1$, and then update $\boldsymbol{p}^{(t)}=:\boldsymbol{p}^{(t-1)}$.
		\FOR{b=1:B}
		\STATE Update the ICI term $I_{b,i}=\sum\limits_{\hfill j \in \mathcal{B} \hfill\atop  j \neq b} \left(\sum\limits_{l\in \mathcal{U}_j} p_{j,l}\right) h_{j,b,i}$ at cell $b$.
		\STATE Update $\boldsymbol{\lambda}_b$ in cell $b$ according to $k \to i$ if $\tilde{h}_{b,k} > \tilde{h}_{b,i}$, where $\tilde{h}_{b,l}=\frac{h_{b,l}}{I_{b,l} + \sigma_{b,l}},~\forall b \in \mathcal{B},~l \in \mathcal{U}_b$.
		\STATE Find $\boldsymbol{p^*}^{(t)}_{b}$ using \eqref{opt powmin}. Then, substitute $\boldsymbol{p}^{(t)}_{b}$ with $\boldsymbol{p^*}^{(t)}_{b}$.
		\ENDFOR
		\ENDWHILE
		\STATE The outputs $\boldsymbol{\lambda}^{*}$ and $\boldsymbol{p}^{*}$ are the optimal solutions.
	\end{algorithmic}
\end{algorithm}
Numerical assessments show that Alg. \ref{Alg opt Mcell powmin} has a very fast convergence speed. The only concern is finding a feasible initial point $\boldsymbol{p}^{(0)}$. In Subsection \ref{subsection conya num}, we discussed about the impact of feasible/infeasible initial points on the convergence of Alg. \ref{Alg opt Mcell powmin}.

It is proved that the optimal $\boldsymbol{\alpha}^*$ in the total power minimization problem is indeed a component-wise minimum \cite{7964738}. It means that for any feasible $\boldsymbol{\hat{\alpha}}=[\hat{\alpha}_b],\forall b \in \mathcal{B}$ in \eqref{feasible problem}, it can be guaranteed that $\alpha^*_b\leq \hat{\alpha}_b$. As a result, the exploration area of each $\alpha_b$ in Alg. \ref{Alg global} can be reduced from $[0,1]$ to $[\alpha^\text{min}_b,1]$, where $\alpha^\text{min}_b$ denotes the optimal power consumption coefficient of BS $b$ in the total power minimization problem. The lower-bounds $\alpha^\text{min}_b,\forall b \in \mathcal{B}$ can significantly reduce the complexity of Alg. \ref{Alg global} for the case that $\alpha^\text{min}_b$ grows, e.g., when the minimum rate demand of users increases.

\subsubsection{Suboptimal SIC Ordering: Joint Rate and Power Allocation}\label{subsection subopt order}
According to Corollary \ref{corollary optorder}, since $\boldsymbol{\lambda}^*$ in \eqref{centg problem} depends on $\boldsymbol{\alpha}^*_{-b}$, any fixed decoding order before power allocation is potentially suboptimal. For a fixed decoding order, Corollary \ref{corol optorder 3} may not hold in the complete feasible region of $\boldsymbol{\alpha}$. In other words, a user with higher decoding order may have lower CINR for some power consumption level of interfering cells.
The main problem \eqref{centg problem} for the given $\boldsymbol{\lambda}$, or equivalently fixed $\Phi_{b,i}$, can be rewritten as
\begin{equation}\label{centg fixed problem}
	\max_{ \boldsymbol{p} \geq 0 }\hspace{.0 cm}	
	~\sum\limits_{b \in \mathcal{B}} \sum\limits_{i \in \mathcal{U}_b} R_{b,i}(\boldsymbol{p})~~~~~~~~~~
	\text{s.t.}~~\eqref{Constraint max power},~\eqref{Constraint QoS}.
\end{equation}
Constraint \eqref{Constraint QoS} can be equivalently transformed to a linear form as follows:
\begin{multline}\label{linear form QoS const}
	2^{R^{\text{min}}_{b,i}} \left(\sum\limits_{j \in \Phi_{b,i}} p_{b,j} h_{b,k} + I_{b,k} + \sigma^2_{b,k}\right) \leq \left(p_{b,i} + \sum\limits_{j \in \Phi_{b,i}} p_{b,j}\right) h_{b,k} + I_{b,k} + \sigma^2_{b,k},~\forall b \in \mathcal{B},
	\\
	i,k \in \mathcal{U}_b,~k \in \{i\} \cup \Phi_{b,i}.
\end{multline}
Hence, the feasible region of \eqref{centg fixed problem} is affine, so is convex in $\boldsymbol{p}$. However, \eqref{centg fixed problem} is still strongly NP-hard, due to the nonconcavity of the sum-rate function with respect to $\boldsymbol{p}$ (see Corollary \ref{corol NPhard}). 
The optimal powers in \eqref{optpow i} and \eqref{optpow M} are derived based on the optimal decoding order. Therefore, Alg. \ref{Alg global} is not applicable for solving \eqref{centg fixed problem}. Besides, the complexity of exhaustive search for solving \eqref{centg fixed problem} is $S^{BM}_p$, which is exponential in the number of users. In the following, we propose a sequential programming method to find a suboptimal power allocation for \eqref{centg fixed problem} while considering the complete rate region of users. To tackle the non-differentiability of $R_{b,i}(\boldsymbol{p})$, we first apply the epigraph technique \cite{Boydconvex}, and define $r_{b,i}$ as the adopted spectral efficiency of user $i \in \mathcal{U}_b$ such that $ r_{b,i} \leq \log_2\left( 1+ \frac{p_{b,i} h_{b,k}}{\sum\limits_{j \in \Phi_{b,i}} p_{b,j} h_{b,k} + I_{b,k} + \sigma^2_{b,k}} \right),~\forall b \in \mathcal{B},~i,k \in \mathcal{U}_b,~k \in \{i\} \cup \Phi_{b,i}$. Hence, \eqref{centg fixed problem} can be equivalently transformed to the following problem as
\begin{subequations}\label{cent subopt problem}
	\begin{align}\label{obf cent subopt problem}
		\text{JRPA}:~&\max_{ \boldsymbol{p} \geq 0,~\boldsymbol{r} \geq 0 }\hspace{.0 cm}	
		~~  \sum\limits_{b \in \mathcal{B}} \sum\limits_{i \in \mathcal{U}_b} r_{b,i}
		\\
		&\text{s.t.}~\eqref{Constraint max power}, \nonumber
		\\
		\label{Constraint QoS RA}
		& r_{b,i} \geq R^{\text{min}}_{b,i},~\forall b \in \mathcal{B},~i \in \mathcal{U}_b,
		\\
		\label{Constraint decoding}
		& r_{b,i} \leq \log_2\left( 1+ \frac{p_{b,i} h_{b,k}}{\sum\limits_{j \in \Phi_{b,i}} p_{b,j} h_{b,k} + I_{b,k} + \sigma^2_{b,k}} \right),~\forall b \in \mathcal{B},~i,k \in \mathcal{U}_b,~k \in \{i\} \cup \Phi_{b,i},
	\end{align}
\end{subequations}
where $\boldsymbol{r}=[r_{b,i}],\forall b \in \mathcal{B}, i \in \mathcal{U}_b$.
Although the objective function \eqref{obf cent subopt problem} and constraints \eqref{Constraint max power} and \eqref{Constraint QoS RA} are affine, problem \eqref{cent subopt problem} is still nonconvex due to the nonconvexity of \eqref{Constraint decoding}. The derivations of the proposed sequential programming method for solving \eqref{cent subopt problem} is presented in Appendix \ref{appendix jointpowerrate}. The pseudo code of our proposed JRPA algorithm is shown in Alg. \ref{Alg JRPA sequential}.
\begin{algorithm}[tp]
	\caption{Suboptimal JRPA Algorithm Based on the Sequential Programming.} \label{Alg JRPA sequential}
	\begin{algorithmic}[1]
		\STATE Initialize $\boldsymbol{r}^{(0)}$, iteration index $t=0$, and tolerance $\epsilon_s \ll 1$.
		\WHILE {$\norm {\boldsymbol{r}^{(t)}-\boldsymbol{r}^{(t-1)}} > \epsilon_s$}
		\STATE Set $t=:t+1$.
		\STATE Update the approximation parameter $\hat{g}(r^{(t)}_{b,i})$ based on ${\boldsymbol{r^*}^{(t-1)}}$.
		\STATE Find $\left({\boldsymbol{r^*}^{(t)}}, {\boldsymbol{\tilde{p}^*}^{(t)}}\right)$ by solving the convex approximated problem \eqref{cent subopt problem 2}.
		\ENDWHILE
		\STATE Set $p^*_{b,i}=\text{e}^{\tilde{p}^*_{b,i}},~\forall b \in \mathcal{B},~i \in \mathcal{U}_b$. The output $\boldsymbol{p^*}^{(t)}$ is adopted for the network.
	\end{algorithmic}
\end{algorithm}

The decoding order for some user pairs is independent from the ICI, so it can be determined prior to power allocation optimization. For a specific channel gain condition, we prove that the CNR-based decoding order is optimal for a user pair independent from power allocation.
\begin{theorem}\label{Theorem SIC M-cell positiveterm}
	\textbf{(SIC sufficient condition)} For each user pair $i,k \in \mathcal{U}_b$ with $\frac{h_{b,k}}{\sigma_{b,k}} \geq \frac{h_{b,i}}{\sigma_{b,i}}$, if
	\begin{equation}\label{suff cond}
		\frac{h_{b,k}}{\sigma_{b,k}} - \frac{h_{b,i}}{\sigma_{b,i}} \geq \sum\limits_{j \in \mathcal{Q}_{b,i,k}} \frac{P^\text{max}_j} {\sigma_{b,i} \sigma_{b,k}}  \left( h_{j,b,k} h_{b,i} - h_{j,b,i} h_{b,k} \right),
	\end{equation}
	where $\mathcal{Q}_{b,i,k}=\left\{j \in \mathcal{B}\setminus\{b\} \middle\vert \frac{h_{b,k}}{h_{b,i}} < \frac{h_{j,b,k}}{h_{j,b,i}}\right\}$, the decoding order $k \to i$ is optimal.
\end{theorem}
\begin{proof}
	Please see Appendix \ref{appendix theorem positiveterm}.
\end{proof}
The applications of Theorem \ref{Theorem SIC M-cell positiveterm} are listed below:
\begin{enumerate}
	\item It reduces the exploration area of finding optimal decoding order based on greedy search algorithm since the optimal decoding order for some user pairs can be determined independent from power allocation. 
	\item It reduces the complexity of Alg. \ref{Alg JRPA sequential}: Assume that the suboptimal CNR-based decoding order is applied in \eqref{cent subopt problem}, i.e., $k \to i$ or $\lambda_{b,i,k}=1$ if $\frac{h_{b,k}}{\sigma_{b,k}} \geq \frac{h_{b,i}}{\sigma_{b,i}}$. Let us define the set of users with higher decoding orders than user $i \in \mathcal{U}_b$ that satisfy the SIC sufficient condition by $\Phi^\text{CNR}_{b,i}=\{k \in \mathcal{U}_b \setminus \{i\} \mid \lambda_{b,i,k}=1,~\frac{h_{b,k}}{\sigma_{b,k}} - \frac{h_{b,i}}{\sigma_{b,i}} \geq \sum\limits_{j \in \mathcal{Q}_{b,i,k}} \frac{P^\text{max}_j} {\sigma_{b,i} \sigma_{b,k}}  \left( h_{j,b,k} h_{b,i} - h_{j,b,i} h_{b,k} \right)\}$. Obviously, we have $\Phi^\text{CNR}_{b,i} \subseteq \Phi_{b,i}$. Based on Theorem \ref{Theorem SIC M-cell positiveterm}, it is required to check the feasibility of constraint \eqref{Constraint decoding} for only users in $\{i\} \cup \Phi_{b,i} \setminus \Phi^\text{CNR}_{b,i}$, which significantly reduces the complexity of our proposed JRPA algorithm. For the case that $\Phi^\text{CNR}_{b,i} = \Phi_{b,i}$, we can guarantee that Corollary \eqref{corol optorder 3} holds for user $i \in \mathcal{U}_b$, meaning that $r^{*}_{b,i} = \tilde{R}_{b,i} (\boldsymbol{p}^*,\boldsymbol{\lambda}^*_b)$. Finally, for the case that the SIC sufficient condition holds for every user pair within cell $b$, i.e., $\Phi^\text{CNR}_{b,i} = \Phi_{b,i},~\forall i \in \mathcal{U}_b$, the CNR-based decoding order is optimal in cell $b$, meaning that the number of constraints in \eqref{Constraint decoding} will be reduced to one for each user pair within cell $b$.
	\item Theorem \ref{Theorem SIC M-cell positiveterm} can be considered as a necessary condition in user grouping of the multicarrier NOMA systems. However, it is still unknown how much the system performance decreases if we limit user grouping based on Theorem \ref{Theorem SIC M-cell positiveterm} which can be considered as future work. 
\end{enumerate}
Similar to problem \eqref{feasible problem}, it can be shown that the feasible region of \eqref{centg fixed problem} is the same as the feasible region of the following total power minimization problem as
\begin{equation}\label{subopt powermin problem}
	\min_{ \boldsymbol{p} \geq 0,~\boldsymbol{r} \geq 0 }\hspace{.0 cm}	
	~\sum\limits_{b \in \mathcal{B}} \sum\limits_{i \in \mathcal{U}_b} p_{b,i}~~~~~~~~~~
	\text{s.t.}~~\eqref{Constraint max power},~\eqref{linear form QoS const}.
\end{equation}
Since the objective function and all the constraints are affine, problem \eqref{subopt powermin problem} is a linear program, so is convex. 
Hence, \eqref{subopt powermin problem} can be solved by using the Dantzig’s simplex method or interior-point methods (IPMs) \cite{Boydconvex}. The solution of \eqref{subopt powermin problem} can be utilized as initial feasible point for Alg. \ref{Alg JRPA sequential}.

\subsubsection{Suboptimal SIC Ordering: Power Allocation for a Fixed Rate Region}\label{subsection SRR}
In Subsection \ref{subsection subopt order}, we showed that for any fixed, thus potentially suboptimal, decoding order, Corollary \ref{corol optorder 3} may not hold at any feasible power allocation, so the joint power allocation and rate adoption is necessary to achieve the maximum sum-rate which is still lower than the users sum-capacity. 
There are a number of research studies in multi-cell NOMA with suboptimal CNR-based decoding order while enforcing users to achieve their channel capacity, i.e., 
$\tilde{R}_{b,i} = \log_2\left( 1+ \frac{p_{b,i} h_{b,i}}{\sum\limits_{j \in \Phi_{b,i}} p_{b,j} h_{b,i} + (I_{b,i} (\boldsymbol{p}_{-b}) + \sigma^2_{b,i})} \right),~\forall b \in \mathcal{B},~i \in \mathcal{U}_b$ by imposing the SIC necessary condition on power allocation. To guarantee that user $i \in \mathcal{U}_b$ achieves its Shannon's capacity for decoding its desired signal after SIC, the condition in \eqref{SIC nec cond} should be satisfied for each user $k \in \Phi_{b,i}$. For any fixed $\boldsymbol{\lambda}$, so fixed $\Phi_{b,i}$, \eqref{SIC nec cond} can be rewritten as
\begin{equation}\label{SIC nec ind pb}
	\frac{h_{b,i}}{ I_{b,i}(\boldsymbol{\alpha}_{-b}) + \sigma^2_{b,i}} \leq \frac{h_{b,k}}{I_{b,k}(\boldsymbol{\alpha}_{-b}) + \sigma^2_{b,k}},~\forall b \in \mathcal{B},~i,k \in \mathcal{U}_b,~k \in \Phi_{b,i}.
\end{equation}
Actually, constraint \eqref{SIC nec ind pb}, known as SIC necessary condition \cite{7812683,8114362,8353846}, states that in the FRPA scheme, the feasible region of 
$\boldsymbol{\alpha}_{-b}$ is limited such that $\tilde{h}_{b,i}(\boldsymbol{\alpha}_{-b})\leq \tilde{h}_{b,k}(\boldsymbol{\alpha}_{-b})$ always holds, in which $\tilde{h}_{b,l}(\boldsymbol{\alpha}_{-b})=\frac{h_{b,l}}{I_{b,l}(\boldsymbol{\alpha}_{-b}) + \sigma^2_{b,l}},~l=i,k$. According to Corollary \ref{corollary optorder}, the feasible region of $\boldsymbol{\alpha}_{-b}$ is limited such that the decoding order $\boldsymbol{\lambda}$ remains optimal. The FRPA problem is formulated by
\begin{subequations}\label{SRR problem}
	\begin{align}\label{obf SRR problem}
		\text{FRPA}:~\max_{ \boldsymbol{p} \geq 0 }\hspace{.0 cm}	
		~~ & \sum\limits_{b \in \mathcal{B}} \sum\limits_{i \in \mathcal{U}_b} \tilde{R}_{b,i}(\boldsymbol{p})
		\\
		\text{s.t.}~~&\eqref{Constraint max power},~\eqref{SIC nec ind pb},
		\\
		\label{Constraint QoS Rtilde}
		& \tilde{R}_{b,i} (\boldsymbol{p}) \geq R^{\text{min}}_{b,i},~\forall b \in \mathcal{B},~i \in \mathcal{U}_b.
	\end{align}
\end{subequations}
\begin{remark}\label{remark equivalent problem}
	For any fixed $\boldsymbol{\lambda}$, problem \eqref{SRR problem} is identical to the main problem \eqref{centg problem} with SIC necessary condition \eqref{SIC nec ind pb}.
\end{remark}
\begin{proof}
	As is mentioned above, \eqref{SIC nec ind pb} is to guarantee that the given decoding order $\boldsymbol{\lambda}$ remains optimal by limiting the feasible region of $\boldsymbol{\alpha}_{-b}$. Therefore, when \eqref{SIC nec ind pb} is imposed to the main problem \eqref{centg problem} with the given $\boldsymbol{\lambda}$, Corollary $2$ holds. It means that, $R_{b,i} (\boldsymbol{p}) = \tilde{R}_{b,i} (\boldsymbol{p}),~\forall b \in \mathcal{B},~i \in \mathcal{U}_b$. Subsequently, the objective function \eqref{obf centg problem} can be transformed to $\sum\limits_{b \in \mathcal{B}} \sum\limits_{i \in \mathcal{U}_b} \tilde{R}_{b,i}(\boldsymbol{p})$ formulated in \eqref{obf SRR problem}. For any fixed $\boldsymbol{\lambda}$, the constraints \eqref{Constraint uniqe order}-\eqref{Constraint own decoding} will be removed from \eqref{centg problem}. Therefore, the resulting problem of \eqref{centg problem} has the same objective function and also the same constraints as in \eqref{SRR problem}, so these problems are identical.
\end{proof}
\begin{remark}\label{remark FRPA optimal}
	Assume that $\boldsymbol{\lambda}$ is fixed. For any given feasible $\boldsymbol{\alpha}$ satisfying \eqref{SIC nec ind pb}, the optimal powers of \eqref{SRR problem} can be obtained by using Proposition \ref{Propos optpower}.
\end{remark}
\begin{proof}
	According to Remark \ref{remark equivalent problem}, we proved that \eqref{SRR problem} is identical to \eqref{centg problem} with SIC necessary condition \eqref{SIC nec ind pb}. Moreover, it is stated that $\boldsymbol{\lambda}$ is optimal in the feasible region of $\boldsymbol{\alpha}$ satisfying \eqref{SIC nec ind pb}. It means that with \eqref{SIC nec ind pb}, we guarantee that $\tilde{h}_{b,i}(\boldsymbol{\alpha}_{-b})\leq \tilde{h}_{b,k}(\boldsymbol{\alpha}_{-b})$ independent from $\boldsymbol{p}_b$. According to Appendix \ref{appendix optpower}, for any given $\boldsymbol{\alpha}$ satisfying \eqref{SIC nec ind pb}, problem \eqref{SRR problem} can be equivalently divided into $B$ single-cell NOMA sub-problems. In each cell $b$, we find $\boldsymbol{p}^*_b$ by solving \eqref{cell problem 2}, and the proof is completed.
\end{proof}
The problem \eqref{SRR problem} is strongly NP-hard, since for $|\mathcal{U}_b|=1,~\forall b \in \mathcal{B}$, the multi-cell NOMA network is identical to the point-to-point interference-limited network with $B$ transmitter-receiver pairs (see Corollary \ref{corol NPhard}). Besides, the optimal $\boldsymbol{\alpha}^*$ in \eqref{SRR problem} should satisfy the following necessary and sufficient conditions: 1) It results maximum sum-rate: $\sum\limits_{b \in \mathcal{B}} R^\text{tot}_b(\boldsymbol{\alpha}^*) \geq \sum\limits_{b \in \mathcal{B}} R^\text{tot}_b(\boldsymbol{\alpha})$, for any feasible $\boldsymbol{\alpha}$, where $R^\text{tot}_b(\boldsymbol{\alpha})=\sum\limits_{i=1}^{M-1} R^{\text{min}}_{b,i} + \log_2\left( 1+p^*_{b,M}(\boldsymbol{\alpha}) \tilde{h}_{b,M} \right)$; 2) SIC necessary conditions: $\frac{h_{b,i}}{ I_{b,i}(\boldsymbol{\alpha}^*_{-b}) + \sigma^2_{b,i}} \leq \frac{h_{b,k}}{I_{b,k}(\boldsymbol{\alpha}^*_{-b}) + \sigma^2_{b,k}},~\forall b \in \mathcal{B},~i,k \in \mathcal{U}_b,~k \in \Phi_{b,i}$. Therefore, similar to Alg. \ref{Alg global}, we perform a greedy search on the possible values of $\boldsymbol{\alpha}$ to find optimal, or more precisely, $\epsilon$-suboptimal $\boldsymbol{\alpha}^*$ such that the latter two conditions are satisfied. The pseudo code of the proposed algorithm for solving the FRPA problem is presented in Alg. \ref{Alg power FRPA}.
\begin{algorithm}[tp]
	\caption{Optimal Power Allocation for Fixed SIC Ordering and Rate Region.} \label{Alg power FRPA}
	\begin{algorithmic}[1]
		\STATE Initialize the step size $\epsilon_\alpha=\frac{1}{S_\alpha} \ll 1$, where $S_\alpha$ is the number of samples for each $\alpha_b$, and $R_\text{tot}=0$.
		\FOR {each sample $\boldsymbol{\hat{\alpha}}$}
		\STATE Update $\tilde{h}_{b,i}=\frac{h_{b,i}}{\hat{I}_{b,i} + \sigma_{b,i}},~\forall b \in \mathcal{B},~i \in \mathcal{U}_b$, where 
		$\hat{I}_{b,i}=\sum\limits_{\hfill j \in \mathcal{B} \hfill\atop  j \neq b} \hat{\alpha}_j P^{\text{max}}_j h_{j,b,i}$.
		\\
		\STATE\textbf{if}~~$\left(\frac{h_{b,i}}{ I_{b,i} + \sigma^2_{b,i}} \leq \frac{h_{b,k}}{I_{b,k} + \sigma^2_{b,k}},~\forall b \in \mathcal{B},~i,k \in \mathcal{U}_b,~k \in \Phi_{b,i}\right)$~\textbf{then}
		\\
		\STATE~~~~~Find $\boldsymbol{p}$ according to \eqref{optpow i} and \eqref{optpow M}.
		\\
		\STATE~~~~~\textbf{if}~~$\left(\sum\limits_{b \in \mathcal{B}} \sum\limits_{i \in \mathcal{U}_b} R_{b,i}(\boldsymbol{p}) > R_\text{tot}\right)$~\textbf{then}
		\\~~~~~~~~~~Update $R^*_\text{tot}=\left(\sum\limits_{b \in \mathcal{B}} \sum\limits_{i \in \mathcal{U}_b} R_{b,i}(\boldsymbol{p})\right)$, and $\boldsymbol{p}^*=\boldsymbol{p}$.
		\\~~~~~\textbf{end if}
		\\
		\STATE\textbf{end if}
		\ENDFOR
		\STATE The output $\boldsymbol{p}^{*}$ is the optimal solution of \eqref{SRR problem}.
	\end{algorithmic}
\end{algorithm}
Although our proposed FRPA algorithm scales well with any number of multiplexed users, it is still exponential in the number of BSs (see Subsection \ref{subsection optimal}). One solution is to apply the well-known suboptimal power allocation based on the sequential programming method proposed in \cite{7862919,7812683,8114362}.
\begin{fact}\label{Propos negative SIC}
	In multi-cell NOMA, the SIC necessary condition \eqref{SIC nec ind pb} for each user pair $i,k \in \mathcal{U}_b,~k \in \Phi_{b,i}$ implies additional maximum power constraints on the BSs in $\mathcal{B}\setminus \{b\}$.
\end{fact}
\begin{proof}
	Please see Appendix \ref{appendix negative SIC}.
\end{proof}
\begin{corollary}\label{corol negative SIC 1}
	Restricting the rate region of users under the suboptimal decoding order results in additional limitations on the BSs power consumption. For the user pairs satisfying the SIC sufficient condition, the decoding order is optimal, and subsequently the SIC constraint \eqref{SIC nec ind pb} will be completely independent from the power allocation (see Corollary \ref{corol optorder 3}). As a result, the negative side impact of the SIC necessary condition on power allocation will be eliminated.
\end{corollary}

The feasible solution for \eqref{SRR problem} can be obtained by solving the following total power minimization problem as
\begin{equation}\label{SRR powermin problem}
	\min_{ \boldsymbol{p} \geq 0}\hspace{.0 cm}	
	~\sum\limits_{b \in \mathcal{B}} \sum\limits_{i \in \mathcal{U}_b} p_{b,i}~~~~~~~~
	\text{s.t.}~\eqref{Constraint max power},~\eqref{SIC nec ind pb},~\eqref{Constraint QoS Rtilde}.
\end{equation}
Problem \eqref{SRR powermin problem} is a linear program which can be solved by using the Dantzig’s simplex method.

\subsection{Decentralized Resource Management Frameworks}\label{subsection Decent}
\subsubsection{Fully Distributed JSPA Framework}\label{subsubsection full dist}
Although the globally optimal solution in Subsection \ref{subsection optimal} achieves the channel capacity of users, the complexity of Alg. \ref{Alg global} is still exponential in the number of BSs. Moreover, the centralized framework would cause a large signaling overhead. Here, we propose a fully distributed resource allocation framework in which each BS independently allocates power to its associated users. Actually, we divide the main problem \eqref{centg problem} into $B$ single-cell NOMA problems. 
\begin{fact}\label{Propos singlecell maxpow}
	At the optimal point of sum-rate maximization problem of single-cell NOMA, the BS operates in its maximum available power. It means that the power constraint \eqref{Constraint max power} is active for each BS $b$, i.e., $\sum\limits_{i \in \mathcal{U}_b} p^*_{b,i} = P^{\text{max}}_b,~\forall b \in \mathcal{B}$ in the fully distributed framework.
\end{fact}
\begin{proof}
	Please see Appendix \ref{appendix Propos singlecell maxpow}.
\end{proof}
According to Fact \ref{Propos singlecell maxpow}, we set $\alpha^*_b=1,~\forall b \in \mathcal{B}$. Based on the given $\boldsymbol{\alpha}^*$, the optimal decoding order of users under the fully distributed framework can be easily obtained by Corollary \ref{corollary optorder}. According to Subsection \ref{subsection optimal}, the optimal power $\boldsymbol{p}^*$ under the fully distributed framework can be obtained by using Proposition \ref{Propos optpower}.

\subsubsection{Semi-Centralized JSPA Framework}\label{Subsubsection semicent sum-rate}
The fully distributed framework works well for the case that the ICI levels are significantly low, so $\alpha^*_b \to 1,~\forall b \in \mathcal{B}$ in problem \eqref{centg problem}. For the case that the ICI levels are high, e.g., at femto-cells underlying a single MBS \cite{8325426,8809359,7954630}, the fully distributed framework may seriously degrade the spectral efficiency of femto-cell users. In the following, we propose a semi-centralized JSPA framework in which we assume that the low-power FBSs operate in their maximum power, while the MBS power consumption is obtained by the joint power allocation and SIC ordering of all the users. Let $b=1$ be the MBS's index, and $b=2,\dots,B$ be the index of FBSs. In this framework, we assume that $\alpha_{b}=1,~\forall b=2,\dots,B$. Then, we utilize Alg. \ref{Alg global} to find the globally optimal $\boldsymbol{p}^*$ and $\boldsymbol{\lambda}^*$ for problem \eqref{centg problem}. 
This algorithm performs a greedy search on $0\leq\alpha_1\leq1$. Hence, the complexity of finding optimal JSPA is on the order of total number of samples for $\alpha_1$, i.e., $S_\alpha$. Actually, the computational complexity of this algorithm is independent from the number of FBSs and users which is a good solution for the large-scale systems. The performance of the semi-centralized framework depends on the FBSs optimal power consumption in \eqref{centg problem}. For the case that $\alpha^*_{b} \to 1,~\forall b=2,\dots,B$ in \eqref{centg problem}, the performance gap between the semi-centralized and centralized frameworks tends to zero.

\subsection{Computational Complexity Comparison Between Resource Allocation Algorithms}\label{Subsection complexityanalysis}
In this section, we compare the computational complexity of the proposed different resource allocation algorithms for solving the sum-rate maximization problem \eqref{centg problem}. To simplify the analysis, we assume that each cell has $M$ users. In this comparison, we apply the barrier method with inner Newton's method to achieve an $\epsilon$-suboptimal solution for a convex problem. The number of barrier iterations required to achieve $\frac{m}{t}=\epsilon$-suboptimal solution is exactly $\Upsilon=\lceil \frac{\log\left(m/\left(\epsilon t^{(0)}\right)\right)}{\log \mu} \rceil$, where $m$ is the total number of inequality constraints, $t^{(0)}$ is the initial accuracy parameter for approximating the functions in inequality constraints in standard form, and $\mu$ is the step size for updating the accuracy parameter $t$ \cite{Boydconvex}. The number of inner Newton's iterations at each barrier iteration $i$ is denoted by $N_i$. In general, $N_i$ depends on $\mu$ and how good is the initial points at the barrier iteration $i$. The computational complexity of solving a convex problem is thus on the order of total number of Newton's iterations obtained by $C_\text{cnvx}=\sum\limits_{i=1}^{\Upsilon} N_i$. For the case that the sequential programming converges in $Q$ iterations, the complexity of this method is on the order of $C_\text{SP}=\sum\limits_{q=1}^{Q} \sum\limits_{i=1}^{\Upsilon} N_{q,i}$, where $N_{q,i}$ denotes the number of inner Newton's iterations at the $i$-th barrier iteration of the $q$-th sequential iteration.
For convenience, assume that $N_{q,i}=N,~\forall q=1,\dots,Q,~i=1,\dots,\Upsilon$. Hence, we have $C_\text{cnvx}=N\Upsilon= N \lceil \frac{\log\left(m/\left(\epsilon t^{(0)}\right)\right)}{\log \mu} \rceil$, and $C_\text{SP}=QN\Upsilon=QN\lceil \frac{\log\left(m/\left(\epsilon t^{(0)}\right)\right)}{\log \mu} \rceil$. The complexity order of different solution algorithms for \eqref{centg problem} is presented in Table \ref{table complexity}.
\begin{table}[tp]
	\caption{Computational Complexity of Resource Allocation Algorithms for Solving the Sum-Rate Maximization Problem.}
	\begin{center} \label{table complexity}
		\begin{adjustbox}{width=\columnwidth,center}
			\scalebox{1}{\begin{tabular}{|c|c|c|c|c|c|c|c|}
					\hline \rowcolor[gray]{0.890}
					\textbf{Algorithm} & \textbf{Complexity} & \textbf{Framework} & \textbf{Optimal} & $\boldsymbol{\lambda}$ & $\boldsymbol{r}$ & $\boldsymbol{\alpha}$ & $\boldsymbol{p}$  \\
					\hline
					\rowcolor[gray]{0.940}
					JSPA (Alg. \ref{Alg global}) & $S^B$ & Centralized & \checkmark & \checkmark & \checkmark & \checkmark & \checkmark \\
					\hline
					\rowcolor[gray]{0.945}
					Exhaustive Search & $(M!)^B \times S^{BM}$ & Centralized & \checkmark & \checkmark & \checkmark & \checkmark & \checkmark \\
					\hline
					\rowcolor[gray]{0.950}
					JRPA (Alg. \ref{Alg JRPA sequential}) & $QN\left\lceil \frac{\log\left(\left(B+BM+B{M-1\choose 2}\right)/\left(\epsilon t^{(0)}\right)\right)}{\log \mu} \right\rceil$ & Centralized & $\times$  & $\times$ & \checkmark & \checkmark & \checkmark \\
					\hline
					\rowcolor[gray]{0.955}
					FRPA & $S^B$ & Centralized & \checkmark  & $\times$ & $\times$ & \checkmark & \checkmark \\
					\hline
					\rowcolor[gray]{0.960}
					Monotonic Optimization \cite{7862919,7812683} & $\approx S^{BM}$ & Centralized & \checkmark & $\times$ & $\times$ & \checkmark & \checkmark \\
					\hline
					\rowcolor[gray]{0.965}
					Sequential Program \cite{7862919,7812683} & $QN\left\lceil \frac{\log\left(\left(B+BM+B{M-1\choose 2}\right)/\left(\epsilon t^{(0)}\right)\right)}{\log \mu} \right\rceil$ & Centralized & $\times$  & $\times$ & $\times$ & \checkmark & \checkmark \\
					\hline
					\rowcolor[gray]{0.970}
					Subsection \ref{Subsubsection semicent sum-rate} & $S$ & Semi-Centralized & \checkmark  & \checkmark & \checkmark & $\alpha_1$ & \checkmark \\
					\hline
					\rowcolor[gray]{0.975}
					Subsection \ref{subsubsection full dist} & $1$ & Fully Distributed & \checkmark  & \checkmark & \checkmark & $\times$ & \checkmark \\
					\hline
			\end{tabular}}
		\end{adjustbox}
	\end{center}
\end{table}
In this table, $S$ denotes the number of samples for each $p_{b,i}$ or $\alpha_{b}$. Moreover, the optimality status is for only the simplified problem. For example, FRPA finds the globally optimal powers in $\boldsymbol{p}$, and subsequently $\boldsymbol{\alpha}$, when $\boldsymbol{\lambda}$ and $\boldsymbol{r}$ are fixed. It does not mean that the output is globally optimal solution for the main problem \eqref{centg problem}.
Actually, only the first and second rows in Table \ref{table complexity} can find the globally optimal solution of \eqref{centg problem}, and the rest of the algorithms are indeed suboptimal. It is noteworthy that in Table \ref{table complexity}, we considered the highest possible computational complexity (worst case) of each algorithm.

\section{Numerical Results}\label{sec simulation}
In this section, we evaluate the performance of our proposed resource allocation algorithms via MATLAB Monte Carlo simulations over $10000$ channel realizations. This comparison is divided into two subsections: 1) performance comparison among our proposed JSPA, JRPA, and FRPA algorithms to demonstrate the benefits of optimal SIC ordering, and rate adoption for any suboptimal decoding order (see Table \ref{table complexity}); 2) performance comparison among the centralized and decentralized resource allocation frameworks to demonstrate the effect of optimal $\boldsymbol{\alpha}^*$ (ICI management) in the main problem \eqref{centg problem}. In our simulations, we adopt the commonly-used (suboptimal) CNR-based decoding order \cite{8114362,7812683} for the JRPA and FRPA algorithms.
Finally, we evaluate the convergence of the iterative algorithms for solving \eqref{feasible problem} and \eqref{cent subopt problem}, and the performance of approximated optimal powers in Remark \ref{remark approx power}. The source code is available on GitLab \cite{sourcecode}.

\subsection{Simulation Settings}
Here, we consider a two-tier HetNet consisting of one FBS underlying a MBS\footnote{Although in practice there are larger number of FBSs, in this experiment, we aim to fundamentally investigate the impact of ICI from MBS to FBS and vice versa in optimal decoding order of users.}. Within each cell, there is one BS at the center of a circular area and $\mathcal{U}_b$ users inside it \cite{7964738}. The system parameters and their corresponding notations are shown in Table \ref{Table parameters}.
\begin{table}[tp]
	\centering
	\caption{SYSTEM PARAMETERS}
	\begin{adjustbox}{width=\columnwidth,center}
		\begin{tabular}{c c c || c c c}
			\hline \rowcolor[gray]{0.880}
			\textbf{Parameter} & \textbf{Notation} & \textbf{Value} & \textbf{Parameter} & \textbf{Notation} & \textbf{Value} 
			\\
			\hline \rowcolor[gray]{0.930}
			Coverage of MBS & $\times$ & 500 m  & Lognormal shadowing standard deviation & $\times$ & $8$ dB 
			\\ \rowcolor[gray]{0.935}
			Coverage of FBS & $\times$ & 40 m  & Small-scale fading & $\times$ & Rayleigh fading
			\\ \rowcolor[gray]{0.940}
			Distance between MBS and FBS & $\times$ & 200 m  & AWGN power density & $N_{b,i}$ & -174 dBm/Hz
			\\ \rowcolor[gray]{0.945}
			Number of macro-cell users & $|\mathcal{U}_m|$ & \{2;3;4\}  & MBS transmit power & $P^\text{max}_m$ & 46 dBm
			\\ \rowcolor[gray]{0.950}
			Number of femto-cell users & $|\mathcal{U}_f|$ & \{2;3;4;5;6\}  & FBS transmit power & $P^\text{max}_f$ & 30 dBm
			\\ \rowcolor[gray]{0.955}
			User distribution model & $\times$ & Uniform  & Minimum rate of macro-cell users  & $R^\text{min}_m$ & $\{0.5;1;2\}$ bps/Hz
			\\ \rowcolor[gray]{0.960}
			Minimum distance of users to MBS & $\times$ & 20 m  & Minimum rate of femto-cell users  & $R^\text{min}_f$ &  $\{0.25;0.50;0.75;1;2;3\}$ bps/Hz
			\\ \rowcolor[gray]{0.965}
			Minimum distance of users to FBS & $\times$ & 2 m  & Tolerance of Alg. \ref{Alg opt Mcell powmin} & $\epsilon_\text{tol}$ & $10^{-6}$
			\\ \rowcolor[gray]{0.970}
			Wireless bandwidth & $\times$ & $5$ MHz & Step size of each $\alpha_b \in [0,1]$  & $\epsilon_\alpha$ & $10^{-2}$
			\\ \rowcolor[gray]{0.975}
			MBS path loss & $\times$ & $128.1 + 37.6 \log_{10} (d/\text{Km})$ dB  & Tolerance of Alg. \ref{Alg JRPA sequential} & $\epsilon_s$ & $10^{-1}$
			\\ \rowcolor[gray]{0.975}
			FBS path loss & $\times$ & $140.7 + 36.7 \log_{10} (d/\text{Km})$ dB  & -  & - & -
			\\
			\hline
		\end{tabular}
		\label{Table parameters}
	\end{adjustbox}
\end{table}
The network topology and exemplary users placement are shown in Fig. \ref{FigTopology}.
\begin{figure}
	\centering
	\includegraphics[scale=0.6]{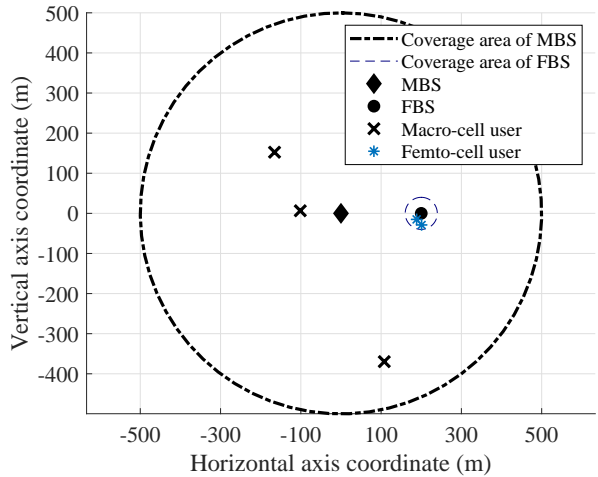}
	\caption{Network topology and exemplary user placement for $|\mathcal{U}_m|=3$, and $|\mathcal{U}_f|=2$.}
	\label{FigTopology}
\end{figure}

\subsection{Performance of Centralized Resource Allocation Algorithms}
In this subsection, we compare the performance in terms of outage probability, and users total spectral efficiency of our proposed JSPA, JRPA, and FRPA algorithms. Note that FRPA finds the globally optimal solution of the sum-rate maximization problem in \cite{8114362} and the single-carrier-based downlink power allocation problem in \cite{7812683} with significantly reduced computational complexity (see Table \ref{table complexity}).
\subsubsection{Outage Probability Performance}
The outage probability of the JSPA, JRPA, and FRPA schemes is obtained based on optimally solving their corresponding total power minimization problems. For each scheme, the outage probability is calculated by dividing the number of infeasible problem
instances by total number of samples \cite{7964738}. Fig. \ref{Fig_outage_cent} shows the impact of order of NOMA clusters and minimum rate demands on the outage probability of the JSPA, JRPA, and FRPA problems. 
\begin{figure}[tp]
	\centering
	\subfigure[Average number of user pairs which cannot satisfy the SIC sufficient condition vs. order of NOMA cluster for the CNR-based decoding order.]{
		\includegraphics[scale=0.52]{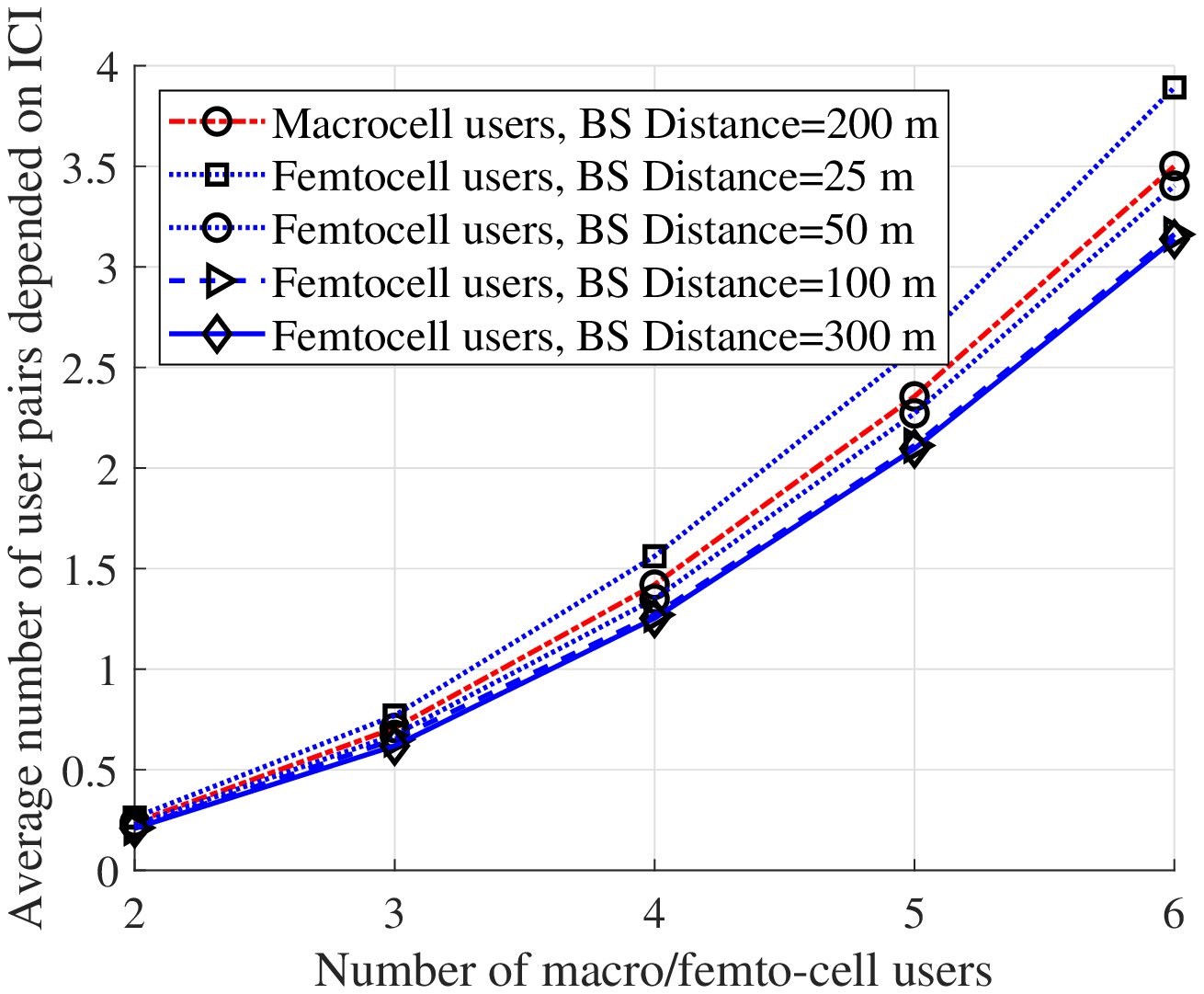}
		\label{Fig_suff_femtoBS}
	}
	\subfigure[Outage probability vs. order of NOMA cluster for $R^\text{min}_m=R^\text{min}_f=1$ bps/Hz.]{
		\includegraphics[scale=0.52]{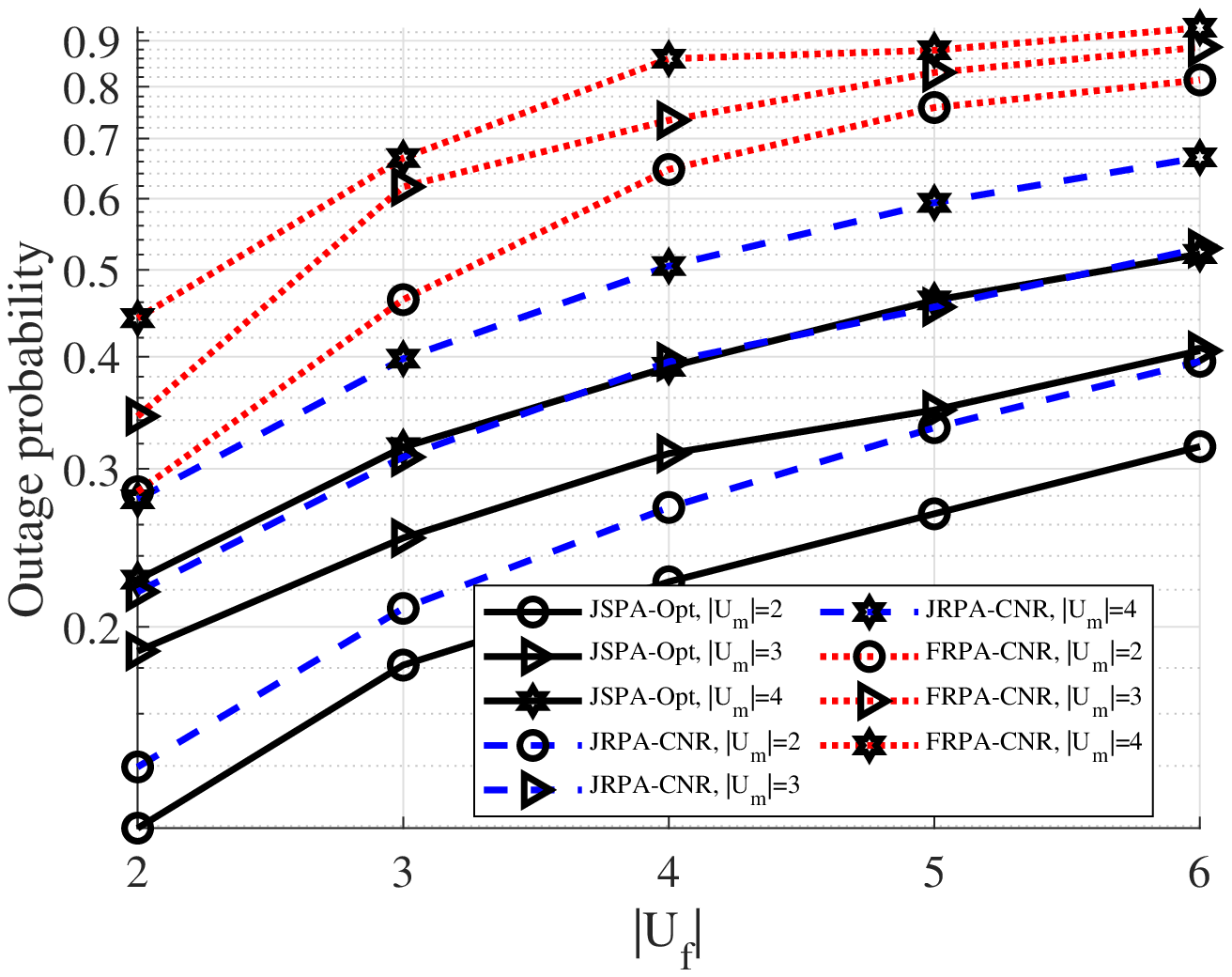}
		\label{Fig_outage_Algs_usernum}
	}
	\subfigure[Outage probability vs. users minimum rate demand for $2$-order NOMA clusters.]{
		\includegraphics[scale=0.52]{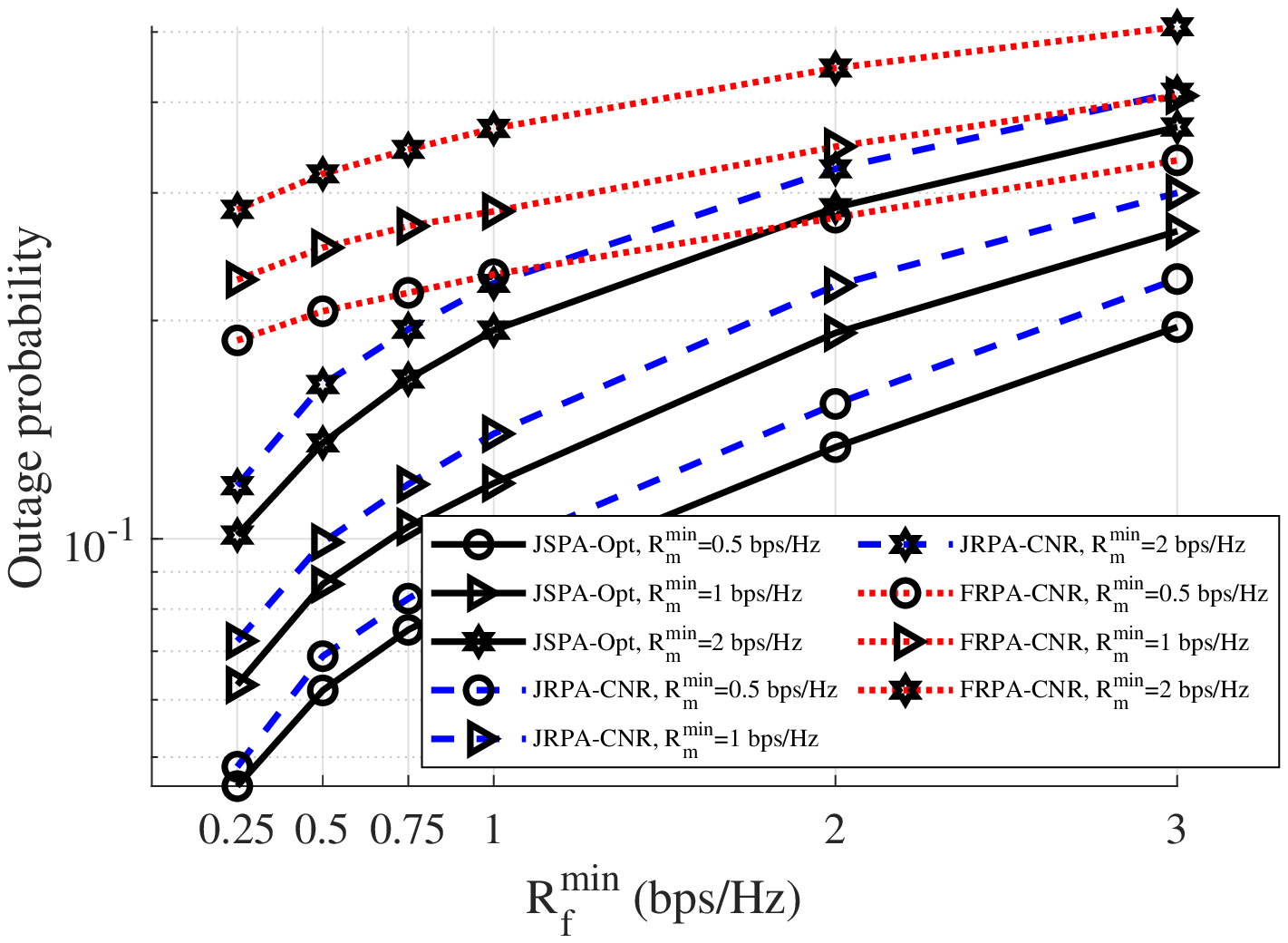}
		\label{Fig_outage_Algs_minrate_2user}
	}
	\subfigure[Outage probability vs. users minimum rate demand for $3$-order NOMA clusters.]{
		\includegraphics[scale=0.52]{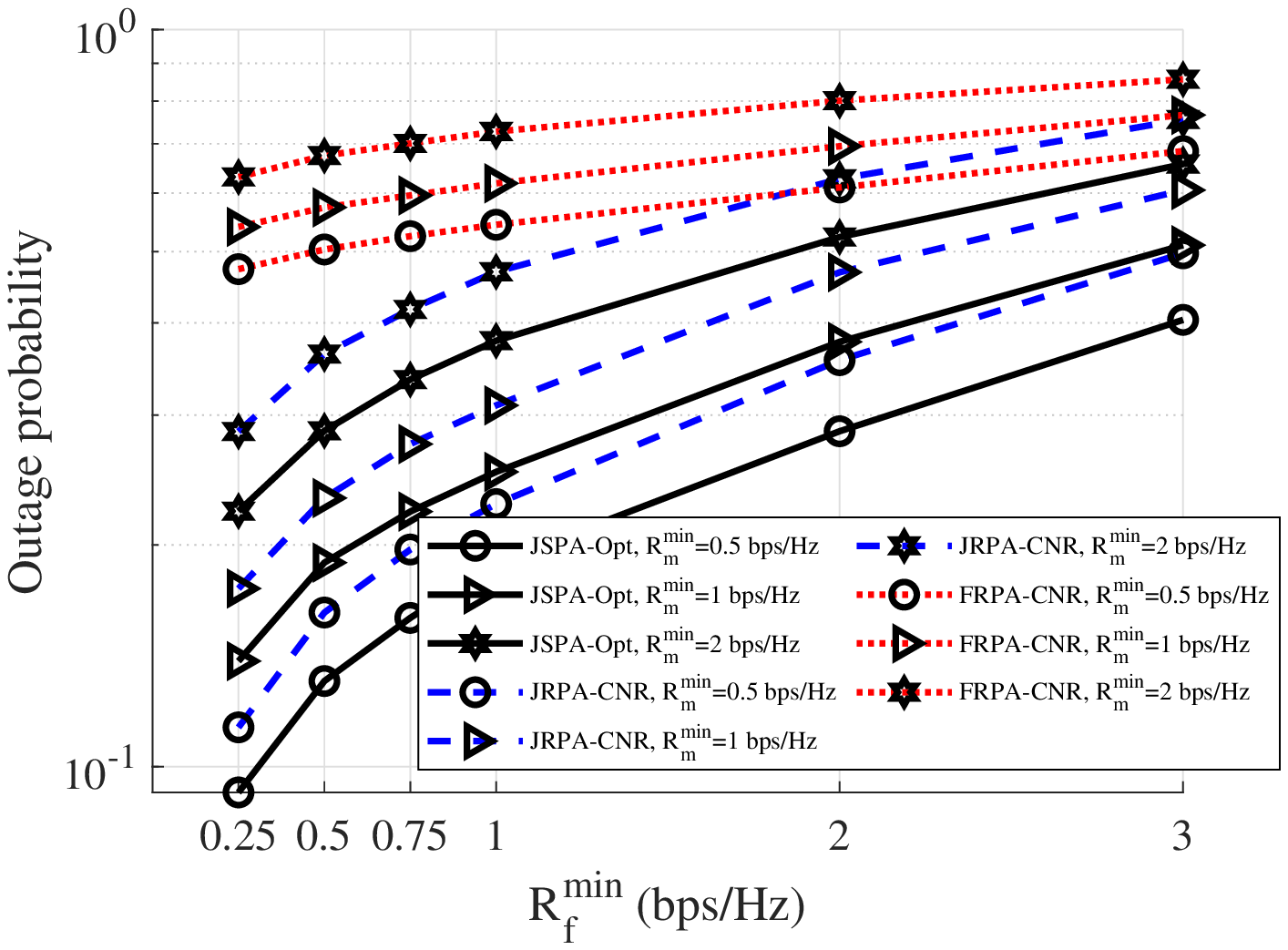}
		\label{Fig_outage_Algs_minrate_3user}
	}
	\caption
	{Outage probability of the centralized JSPA, JRPA, and FRPA algorithms for different number of users and minimum rate demands.}
	\label{Fig_outage_cent}
\end{figure}
According to Theorem \ref{Theorem SIC M-cell positiveterm}, the ICI may affect the optimal decoding order of user pairs which cannot satisfy the SIC sufficient condition. We call these user pairs as the pairs depending on ICI when the CNR-based decoding order is applied. In Fig. \ref{Fig_suff_femtoBS}, we calculate the average number of user pairs depending on ICI (denoted by $\Psi$) for different distances among BSs, and different order of NOMA clusters in the CNR-based decoding order. The parameter $\Psi$ is increased by 1) increasing $|\mathcal{U}_b|$, which inherently decreases the LHS of \eqref{suff cond}; 2) increasing the inter-cell channel gain $h_{j,b,i}$, due to increasing the RHS of \eqref{suff cond}. The second case for the femto-cell user pairs is inversely proportional to the BSs distance, due to the existing path loss. As shown, the impact of $|\mathcal{U}_b|$ is higher than the impact of BSs distance. The wide coverage of macro-cell results in large differences between the users channel gains, however the CNR of macro-cell users is typically low reducing the LHS of \eqref{suff cond}. As a result, $\Psi$ is also affected by $|\mathcal{U}_m|$. It is noteworthy that the position of FBS (BSs distance) in the coverage area of MBS does not have significant impact on $\Psi$ for the macro-cell users. 
For the case that at least one user pair within a cell does not satisfy the SIC sufficient condition, the CNR-based decoding order in that cell may not be optimal, resulting in reduced spectral efficiency and increased outage.

In Figs. \ref{Fig_outage_Algs_usernum}-\ref{Fig_outage_Algs_minrate_3user}, we observe that there exist significant performance gaps between the JSPA and JRPA schemes, which shows the superiority of finding the optimal decoding order in multi-cell NOMA. Besides, JRPA significantly reduces the outage probability of the FRPA scheme which shows the importance of rate adoption when a suboptimal decoding order is applied. The large performance gap between JSPA and FRPA shows that the SIC necessary condition in \eqref{SIC nec ind pb} seriously restricts the feasible region of the FRPA problem (see Fact \ref{Propos negative SIC}) resulting in high outage, specifically for the larger order of NOMA clusters (see Fig. \ref{Fig_outage_Algs_usernum}). Last but not least, in Fig. \ref{Fig_outage_Algs_usernum}, we observe that a larger order of NOMA clusters results in higher outage probability even for the optimal JSPA algorithm.

\subsubsection{Average Total Spectral Efficiency Performance}
Fig. \ref{Fig_sumrate} investigates the impact of order of NOMA clusters and minimum rate remands on the average total spectral efficiency of users. 
\begin{figure}[tp]
	\centering
	\subfigure[Average total spectral efficiency vs. order of NOMA cluster for $R^\text{min}_m=R^\text{min}_f=1$ bps/Hz.]{
		\includegraphics[scale=0.52]{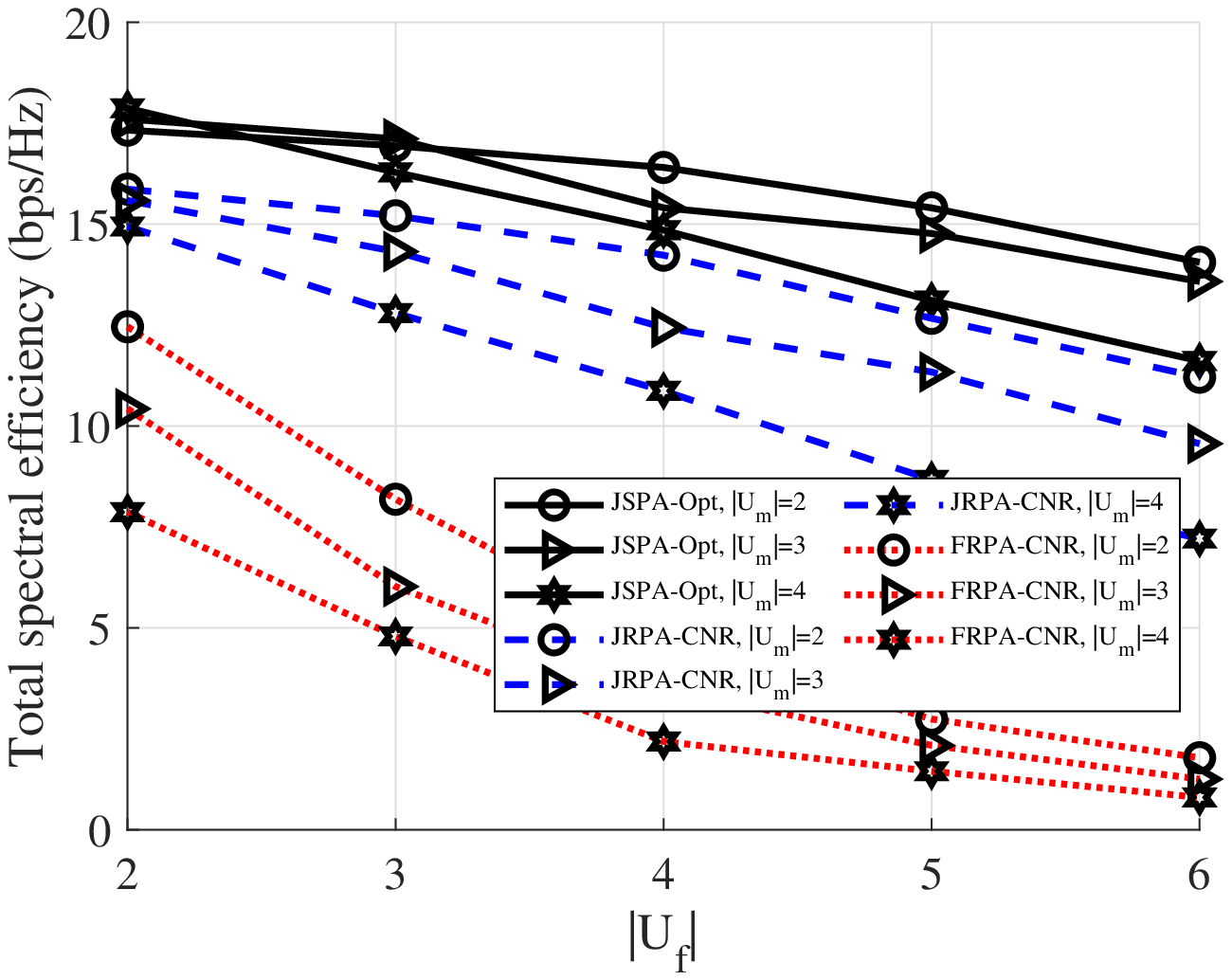}
		\label{Fig_Rtot_usernum}
	}
	\subfigure[Average total spectral efficiency vs. users minimum rate demand for $2$-order NOMA clusters.]{
		\includegraphics[scale=0.52]{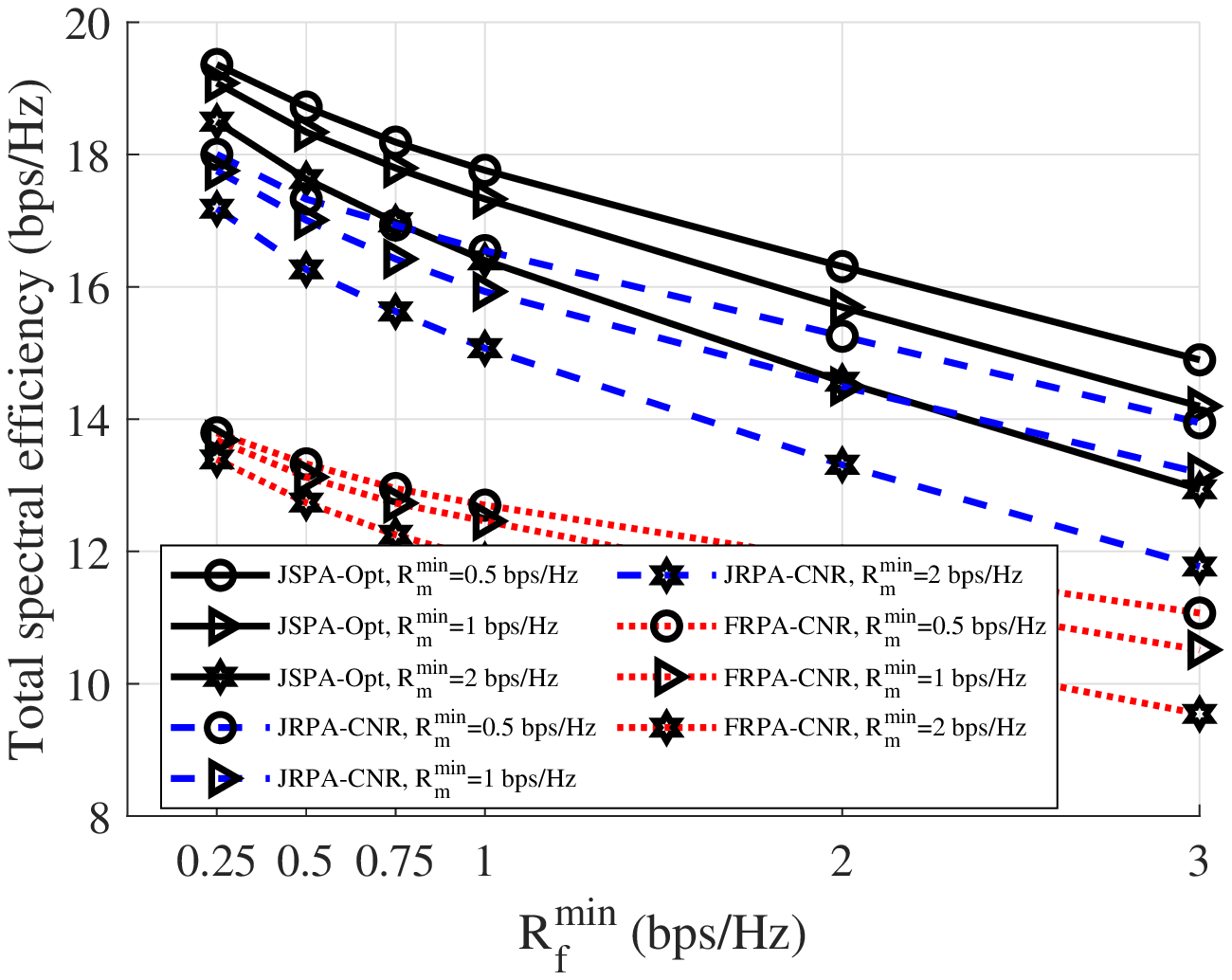}
		\label{Fig_Rtot_minrate_2user}
	}
	\subfigure[Average total spectral efficiency vs. users minimum rate demand for $3$-order NOMA clusters.]{
		\includegraphics[scale=0.52]{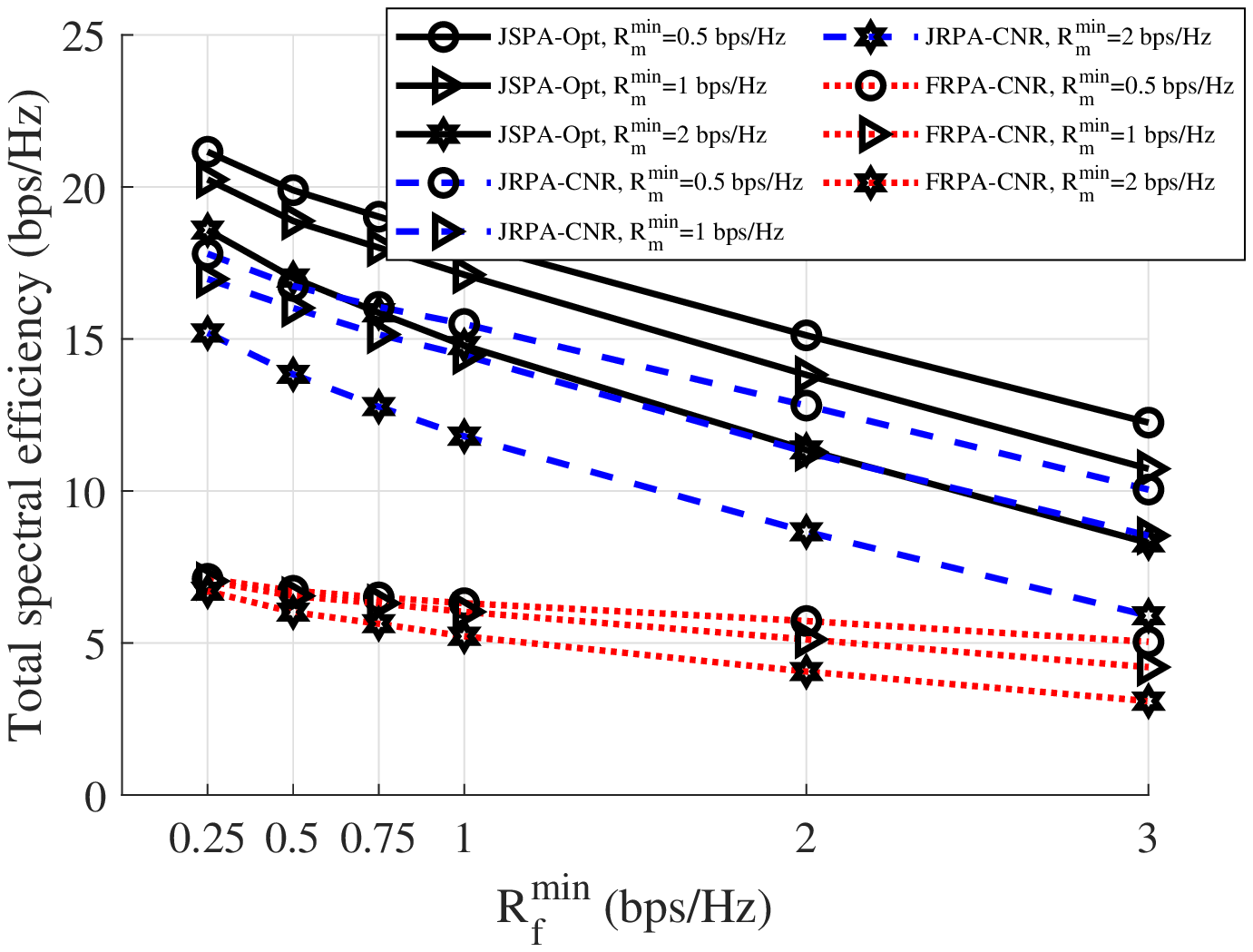}
		\label{Fig_Rtot_minrate_3user}
	}
	\caption
	{Average total spectral efficiency of the centralized JSPA, JRPA, and FRPA algorithms for different number of users and minimum rate demands.}
	\label{Fig_sumrate}
\end{figure}
Here, we set the sum-rate to zero when the problem is infeasible. As shown, JSPA always outperforms the JRPA and FRPA algorithms. The resulting performance gap between JSPA and JRPA is indeed an upper-bound of the exact performance gap between the optimal and CNR-based decoding orders, since Alg. \ref{Alg JRPA sequential} provides a lower-bound for the optimal value of \eqref{cent subopt problem} (see Table \ref{table complexity}). 
In Subsection \ref{subsection JRPA converg}, we show that JRPA is a near to optimal algorithm, so this lower-bound is significantly tighten. 
Subsequently, the performance gap between JRPA and FRPA is indeed the lower-bound of the exact performance gain of rate adoption for the CNR-based decoding order. As can be seen, FRPA has a lower performance compared to JRPA. In Fig. \ref{Fig_Rtot_usernum}, we observe that for $R^\text{min}_m=R^\text{min}_f=1$ bps/Hz, the negative side impact of increasing the order of NOMA cluster is higher than the multi-user diversity gain. As a result, increasing the order of NOMA clusters results in lower total spectral efficiency. Another reason is increasing the outage probability (shown in Fig. \ref{Fig_outage_Algs_usernum}) which can significantly affect the average total spectral efficiency. For $|\mathcal{U}_m|=|\mathcal{U}_f|=2$, we observed that the performance gap is low between JSPA and JRPA, due to decreasing $\Psi$ in Fig. \ref{Fig_suff_femtoBS}. This performance gap grows when $|\mathcal{U}_b|$ increases. As is expected, Figs. \ref{Fig_Rtot_minrate_2user} and \ref{Fig_Rtot_minrate_3user} show that the total spectral efficiency of users is decreasing in minimum rate demands. More importantly, there exist huge performance gaps between JSPA and FRPA for any number of users and minimum rate demands.

\subsection{Performance of Centralized and Decentralized Frameworks}
In this subsection, we compare the performance of the centralized, semi-centralized, and fully distributed JSPA frameworks (shown in Table \ref{table complexity}).
\subsubsection{Outage Probability Performance}
Fig. \ref{Fig_outage_FW} evaluates the outage probability of different resource allocation frameworks. 
\begin{figure}[tp]
	\centering
	\subfigure[Outage probability vs. order of NOMA cluster for $R^\text{min}_m=R^\text{min}_f=1$ bps/Hz.]{
		\includegraphics[scale=0.52]{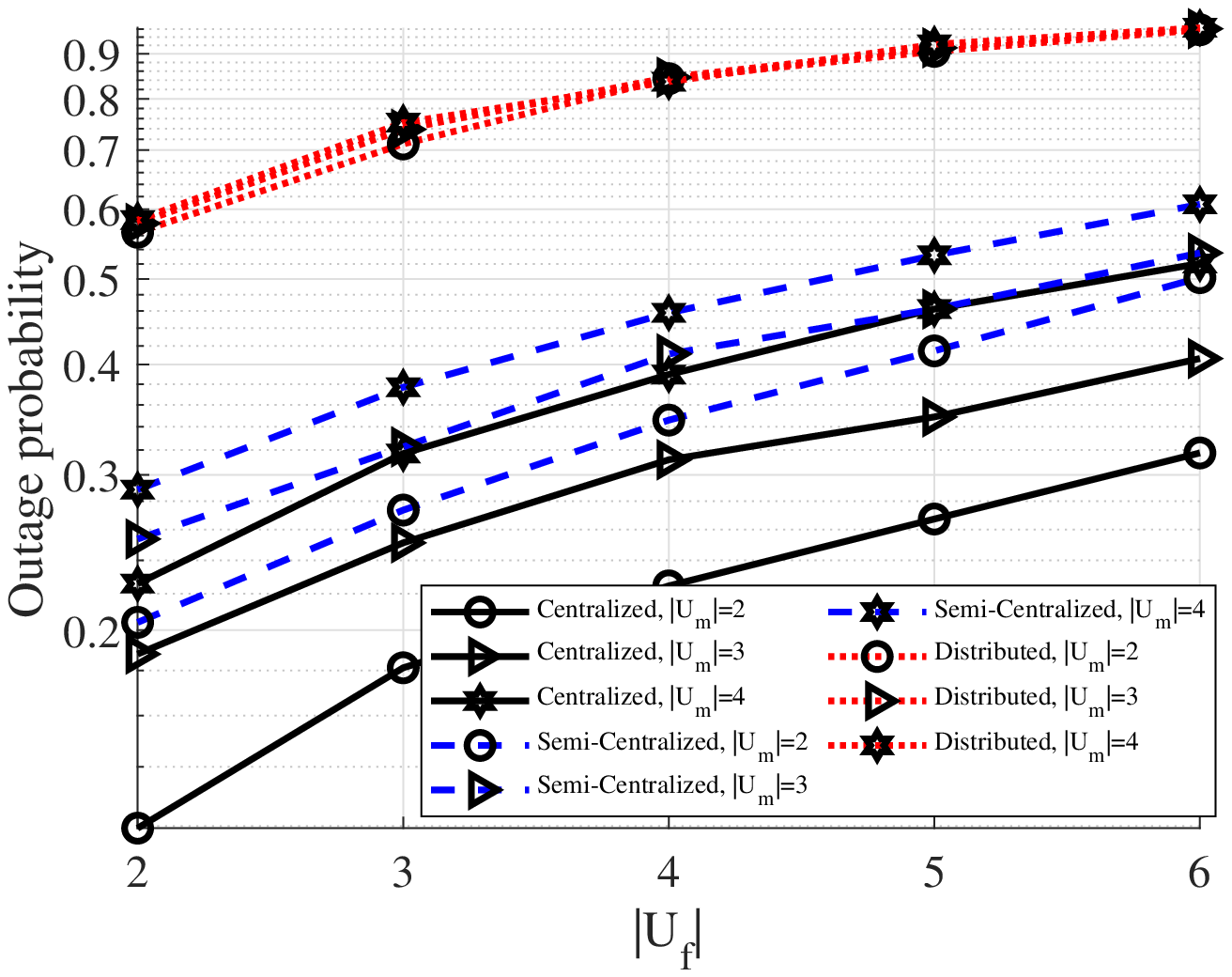}
		\label{Fig_outage_FW_usernum}
	}
	\subfigure[Outage probability vs. users minimum rate demand for $2$-order NOMA clusters.]{
		\includegraphics[scale=0.52]{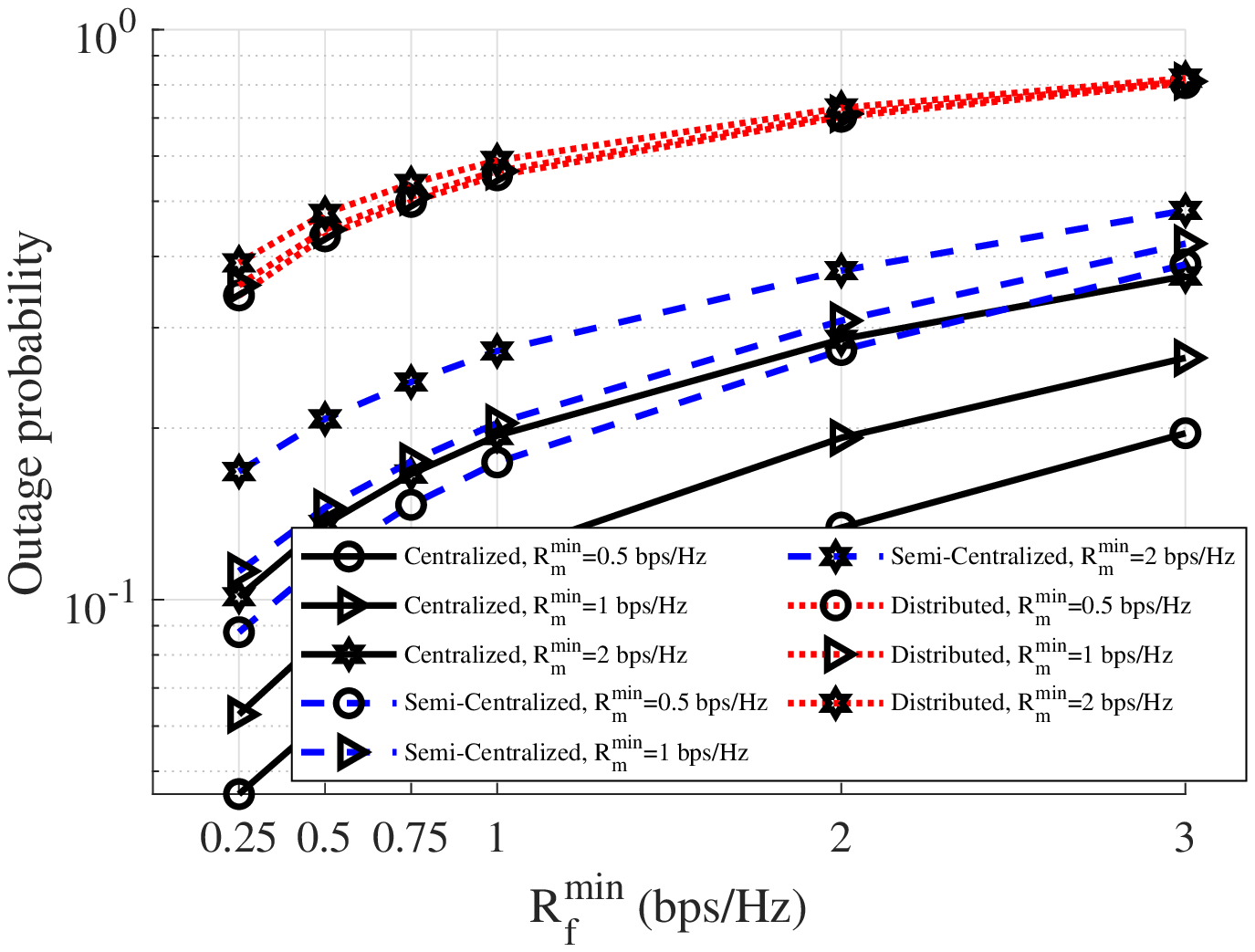}
		\label{Fig_outage_FW_minrate_2user}
	}
	\subfigure[Outage probability vs. users minimum rate demand for $3$-order NOMA clusters.]{
		\includegraphics[scale=0.52]{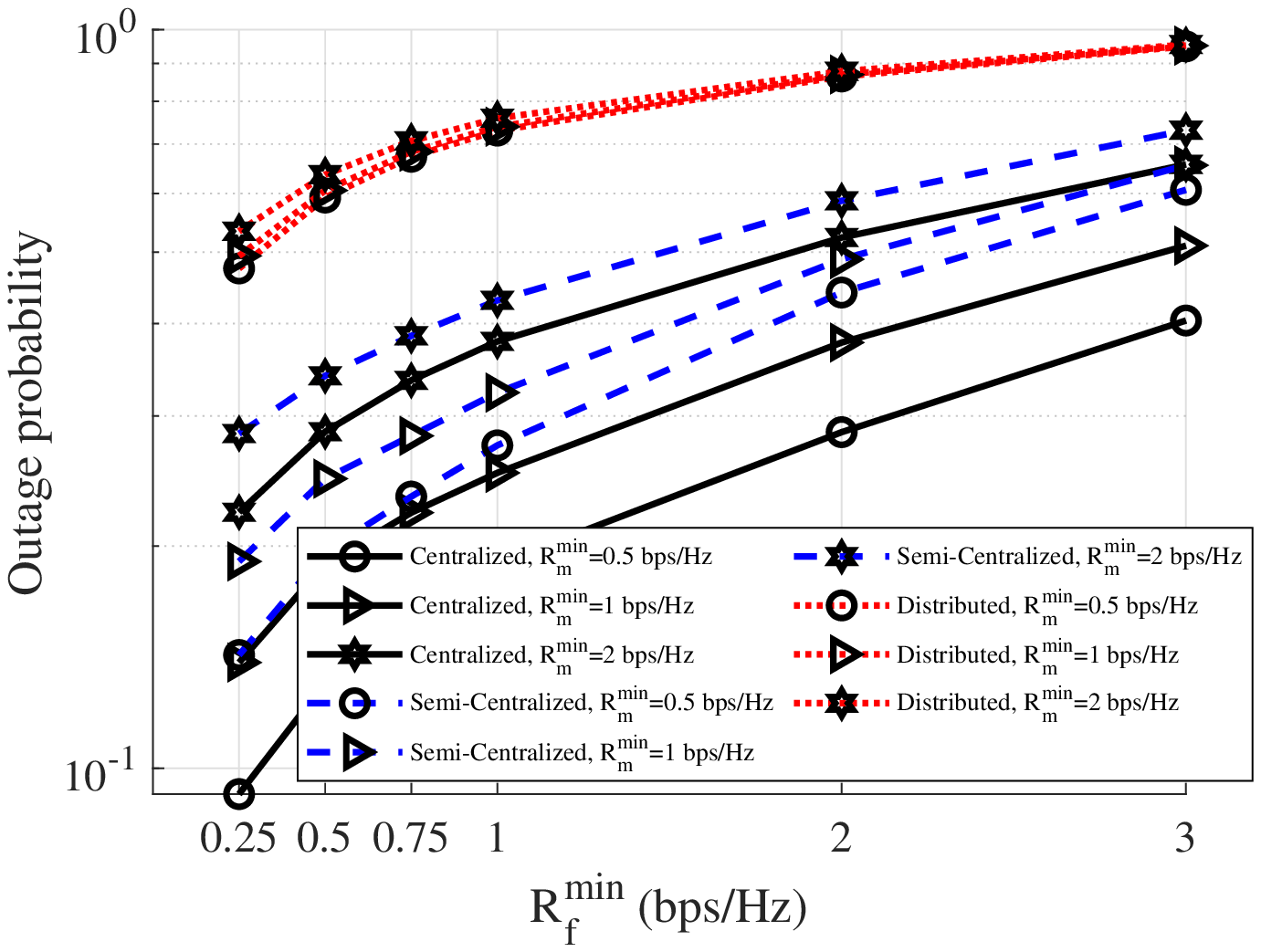}
		\label{Fig_outage_FW_minrate_3user}
	}
	\caption
	{Outage probability of the centralized and decentralized JSPA frameworks for different number of users and minimum rate demands.}
	\label{Fig_outage_FW}
\end{figure}
As can be seen, the fully distributed framework results in huge outage probability. However, the outage probability gap between the semi-centralized and centralized frameworks is decreasing with larger $|\mathcal{U}_m|$ (Fig. \ref{Fig_outage_FW_usernum}), and/or higher $R^\text{min}_m$ (Figs. \ref{Fig_outage_FW_minrate_2user} and \ref{Fig_outage_FW_minrate_3user}). The large performance gap between the semi-centralized and fully distributed frameworks shows the importance of ICI management from MBS to femto-cell users. It is noteworthy that the feasible point for the semi-centralized framework is obtained based on the greedy search on $\alpha_1$ with stepsize $\epsilon_\alpha=10^{-2}$.
Although it can be shown that further reducing $\epsilon_\alpha$ would not cause significant higher total spectral efficiency, $\epsilon_\alpha=10^{-2}$ is not good enough for calculating outage probability. $\epsilon_\alpha$ can be easily reduced to $10^{-4}$ in the semi-centralized framework, however we set $\epsilon_\alpha=10^{-2}$ for both the centralized and semi-centralized frameworks to have a fair comparison. For the same stepsize $\epsilon_\alpha$, the feasible region of the semi-centralized problem is a subset of the feasible region of its corresponding centralized problem.

\subsubsection{Average Total Spectral Efficiency Performance}
The performance of the decentralized frameworks depends on how good is the approximation of the power consumption of the BSs compared to the centralized framework. Figs. \ref{Fig_alphaFBS_FW_usernum}-\ref{Fig_alphaFBS_FW_minrate_3user} show the impact of order of NOMA clusters and minimum rate demands on the FBS power consumption coefficient $\alpha_f$ at the optimal point of the centralized framework. 
\begin{figure}[tp]
	\centering
	\subfigure[FBS power consumption in the centralized framework vs. order of NOMA cluster for $R^\text{min}_m=R^\text{min}_f=1$ bps/Hz.]{
		\includegraphics[scale=0.33]{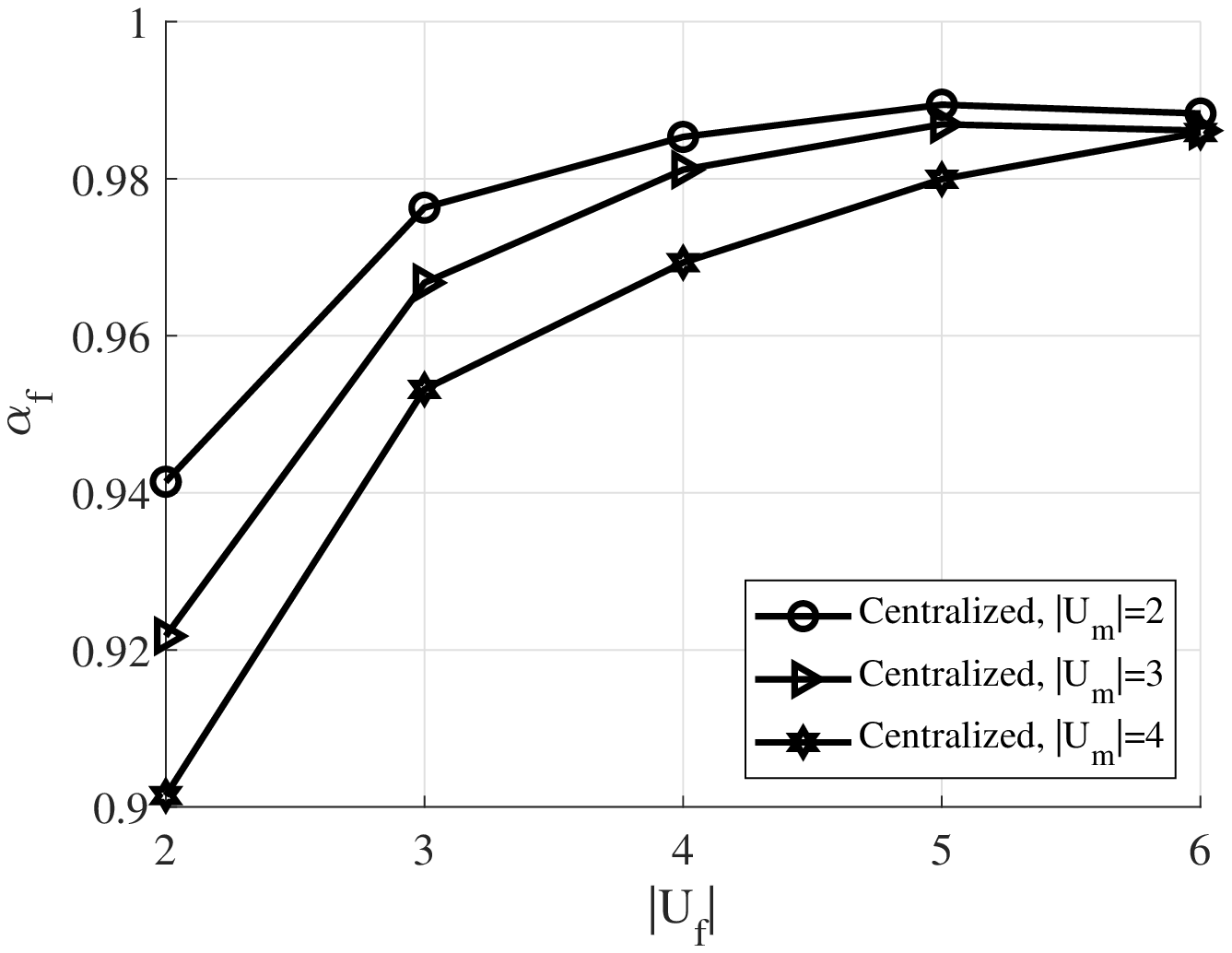}
		\label{Fig_alphaFBS_FW_usernum}
	}
	\subfigure[FBS power consumption in the centralized framework vs. minimum rate demand for $2$-order NOMA clusters.]{
		\includegraphics[scale=0.33]{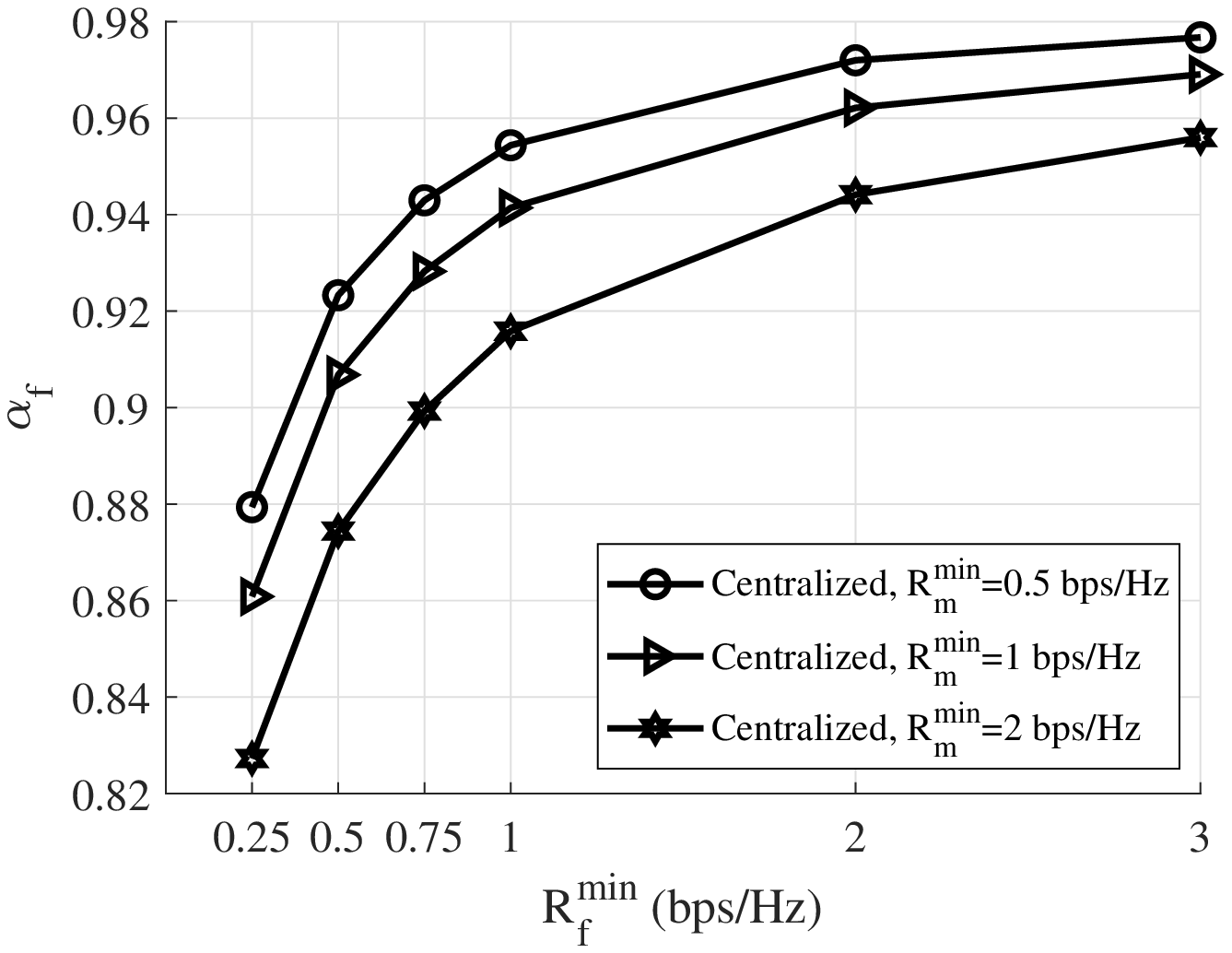}
		\label{Fig_alphaFBS_FW_minrate_2user}
	}
	\subfigure[FBS power consumption in the centralized framework vs. minimum rate demand for $3$-order NOMA clusters.]{
		\includegraphics[scale=0.33]{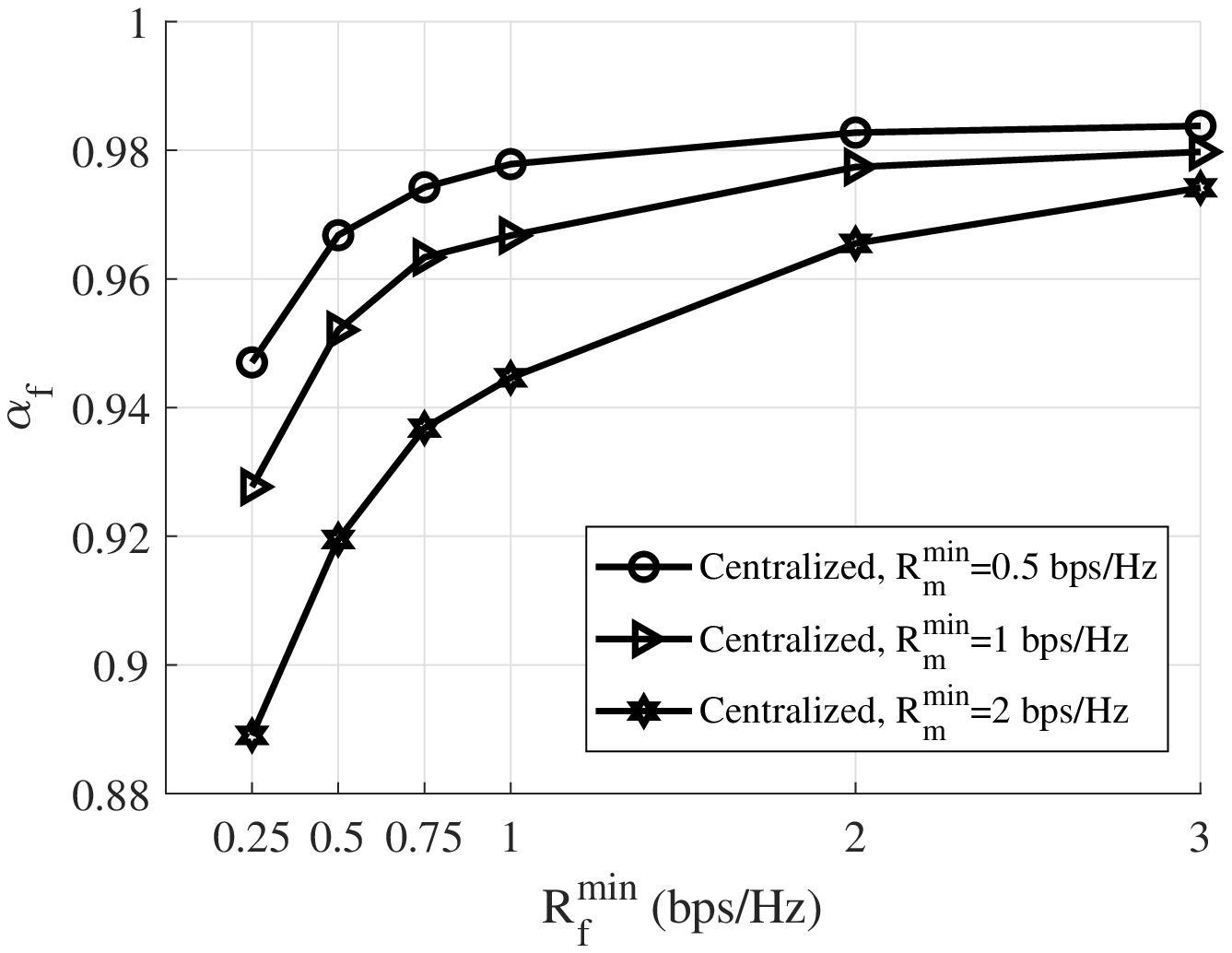}
		\label{Fig_alphaFBS_FW_minrate_3user}
	}
	\subfigure[MBS power consumption in the centralized/semi-centralized frameworks vs. order of NOMA cluster for $R^\text{min}_m=R^\text{min}_f=1$ bps/Hz.]{
		\includegraphics[scale=0.33]{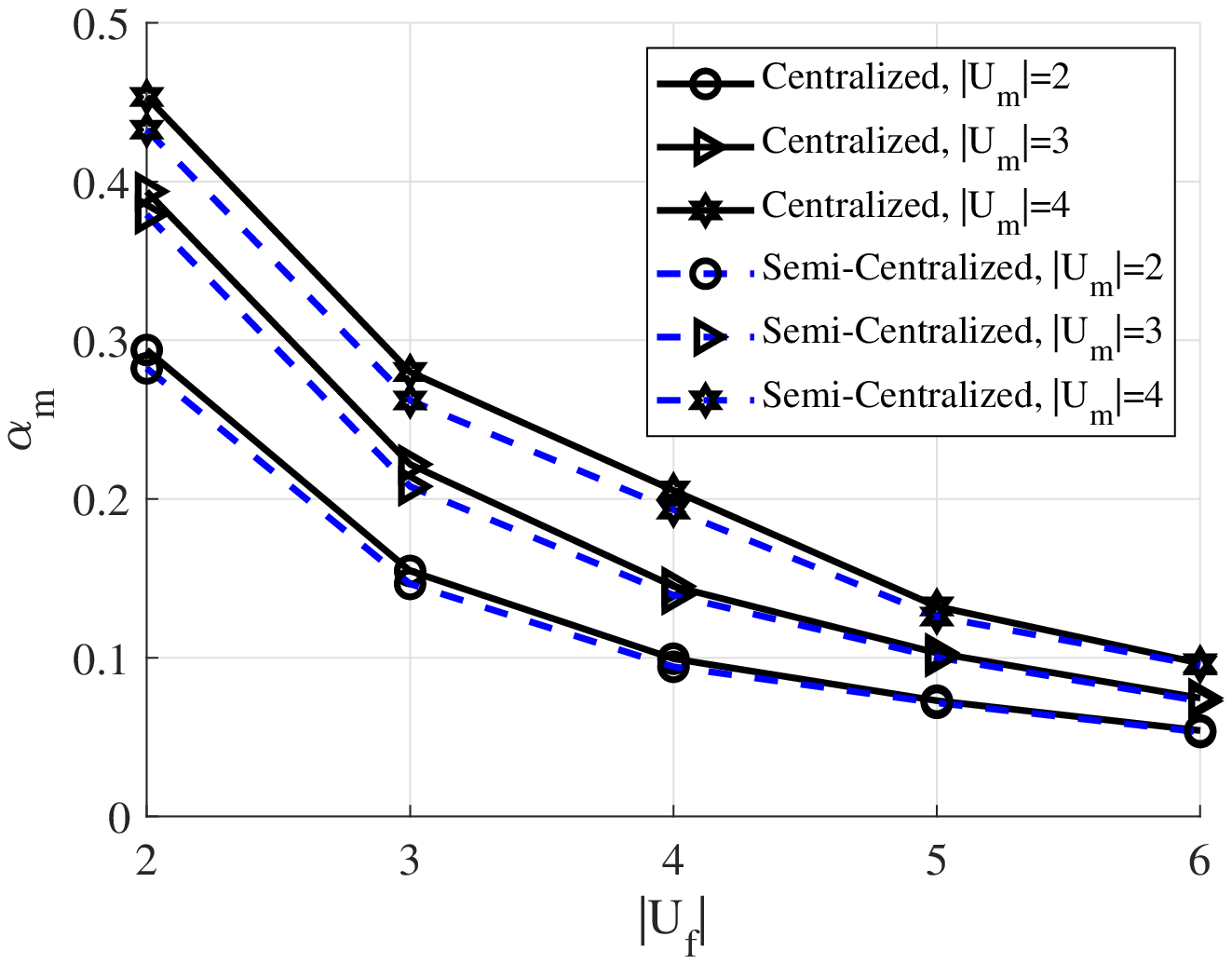}
		\label{Fig_alphaMBS_FW_usernum}
	}
	\subfigure[MBS power consumption in the centralized/semi-centralized frameworks vs. minimum rate demand for $2$-order NOMA clusters.]{
		\includegraphics[scale=0.33]{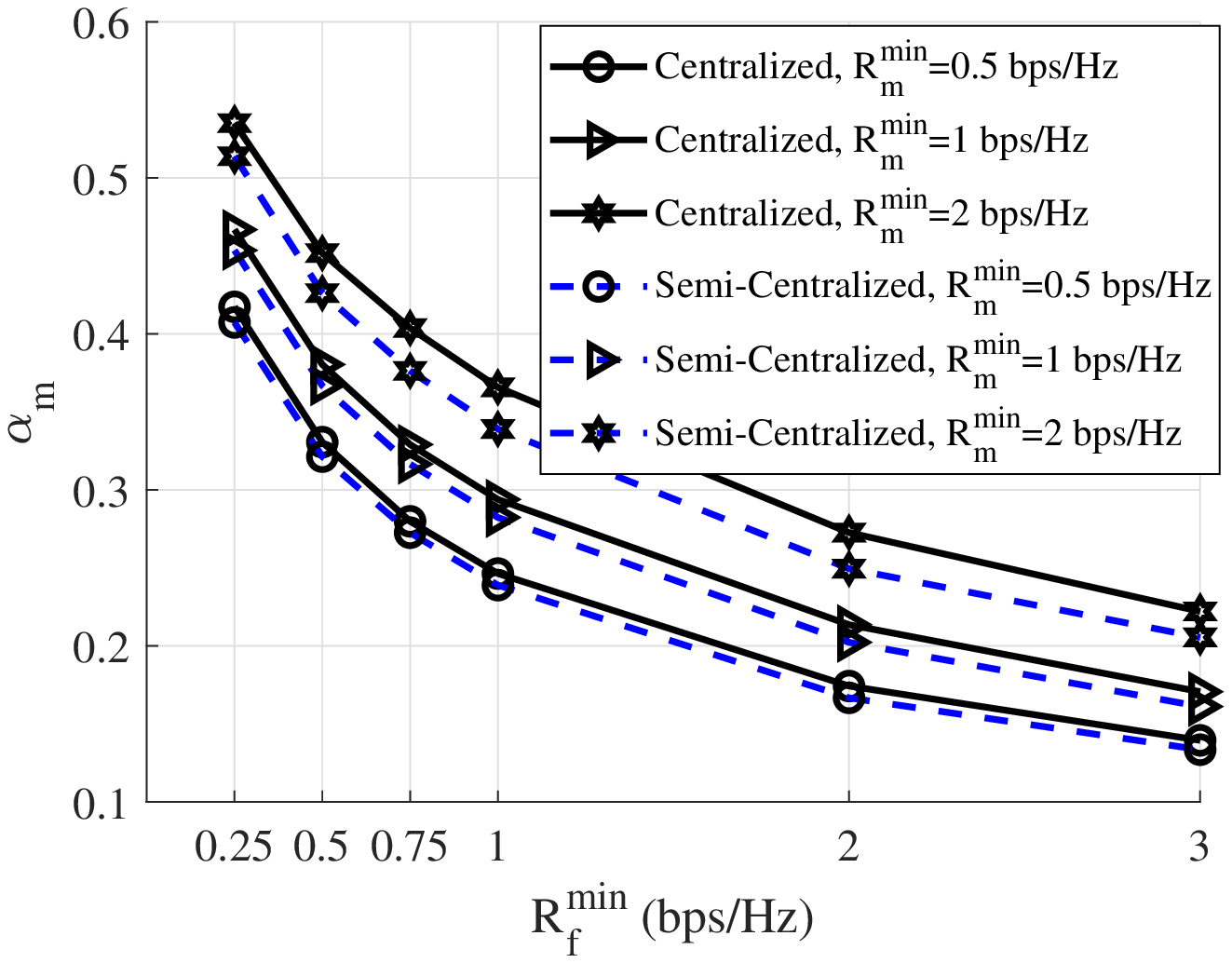}
		\label{Fig_alphaMBS_FW_minrate_2user}
	}
	\subfigure[MBS power consumption in the centralized/semi-centralized frameworks vs. minimum rate demand for $3$-order NOMA clusters.]{
		\includegraphics[scale=0.33]{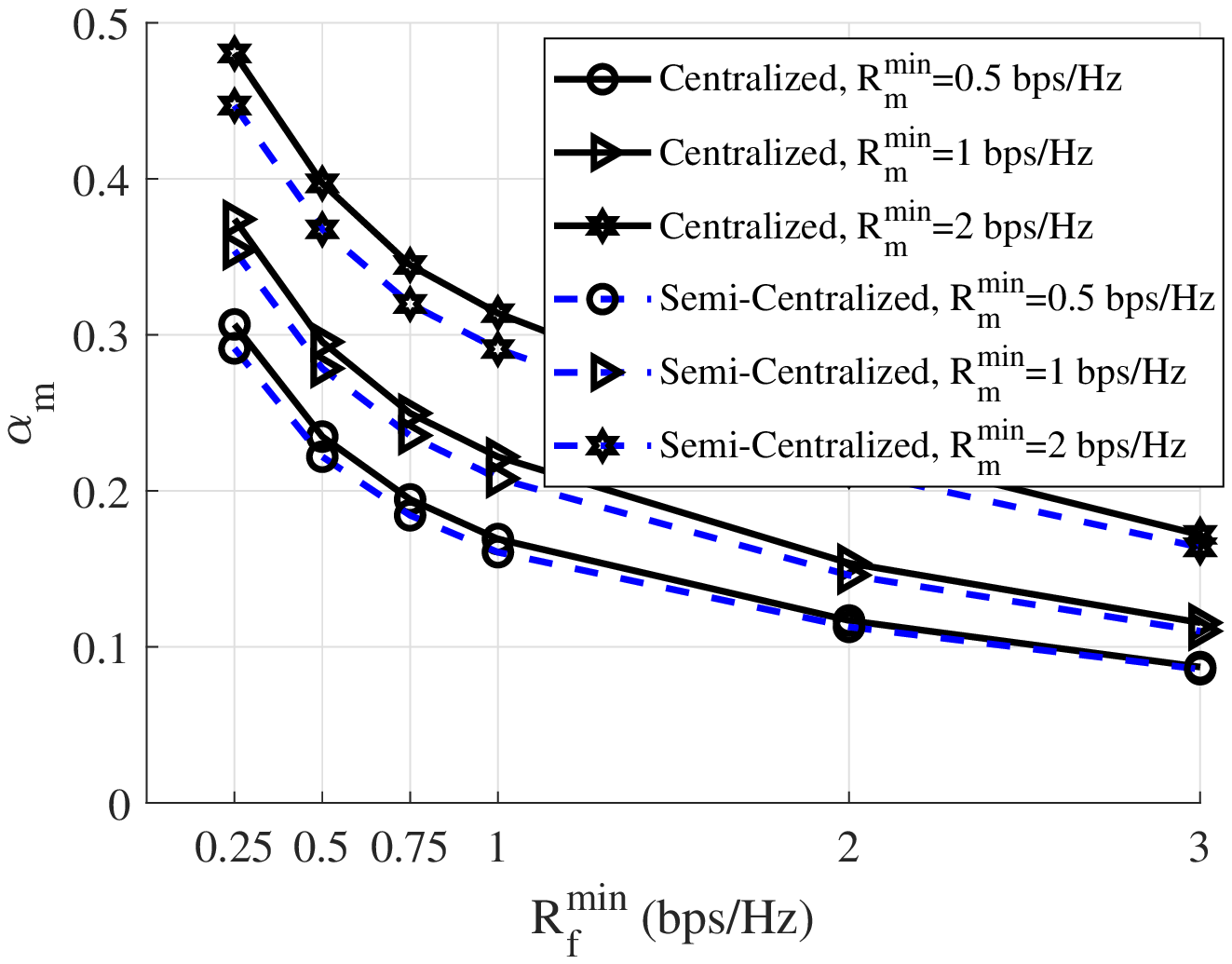}
		\label{Fig_alphaMBS_FW_minrate_3user}
	}
	\caption
	{Average power consumption coefficient of femto/macro BSs in the centralized/semi-centralized frameworks for different number of macro/femto-cell users and minimum rate demands.}
	\label{Fig_FBSpowercons_FW}
\end{figure}
As is expected, larger $|\mathcal{U}_f|$ and/or $R^\text{min}_f$ results in larger FBS power consumption. Moreover, we observe that increasing $|\mathcal{U}_m|$ and/or $R^\text{min}_m$ decreases $\alpha_f$. However, the impact of $|\mathcal{U}_m|$ and/or $R^\text{min}_m$ on $\alpha_f$ is quite low and negligible, due to the low ICI level from low-power FBS to macro-cell users in average. More importantly, we observe that in most of the cases, the FBS operates in up to $90\%$ of its available power. It is noteworthy that in both the decentralized frameworks, we assume that the FBS operates in $100\%$ of its available power.
Figs. \ref{Fig_alphaMBS_FW_usernum}-\ref{Fig_alphaMBS_FW_minrate_3user} evaluate the impact of order of NOMA clusters and minimum rate demands on the MBS power consumption coefficient $\alpha_m$ at the optimal point of the centralized and semi-centralized frameworks. As is expected, $\alpha_m$ is directly proportional to $|\mathcal{U}_m|$ and/or $R^\text{min}_m$, while is inversely proportional to $|\mathcal{U}_f|$ and/or $R^\text{min}_f$. More importantly, we observe that
\begin{enumerate}
	\item $\alpha_m$ in the semi-centralized framework is always upper-bounded by $\alpha_m$ in the centralized framework. This is due to larger $\alpha_f$ in the semi-centralized framework compared to the centralized framework. 
	\item The MBS typically operates in less than $60\%$ of its available power. Hence, the fully distributed framework with $\alpha_m=1$ results in significantly degraded spectral efficiency at femto-cell users, due to the high ICI from the MBS to femto-cell users.
\end{enumerate}
Last but not least, we observe that the MBS power consumption gap between the centralized and semi-centralized frameworks is quite low (see Figs. \ref{Fig_alphaMBS_FW_usernum}-\ref{Fig_alphaMBS_FW_minrate_3user}).

Fig. \ref{Fig_sumrate_FW} investigates the total spectral efficiency of users in the centralized and decentralized frameworks. 
\begin{figure}[tp]
	\centering
	\subfigure[Average total spectral efficiency vs. order of NOMA cluster for $R^\text{min}_m=R^\text{min}_f=1$ bps/Hz.]{
		\includegraphics[scale=0.52]{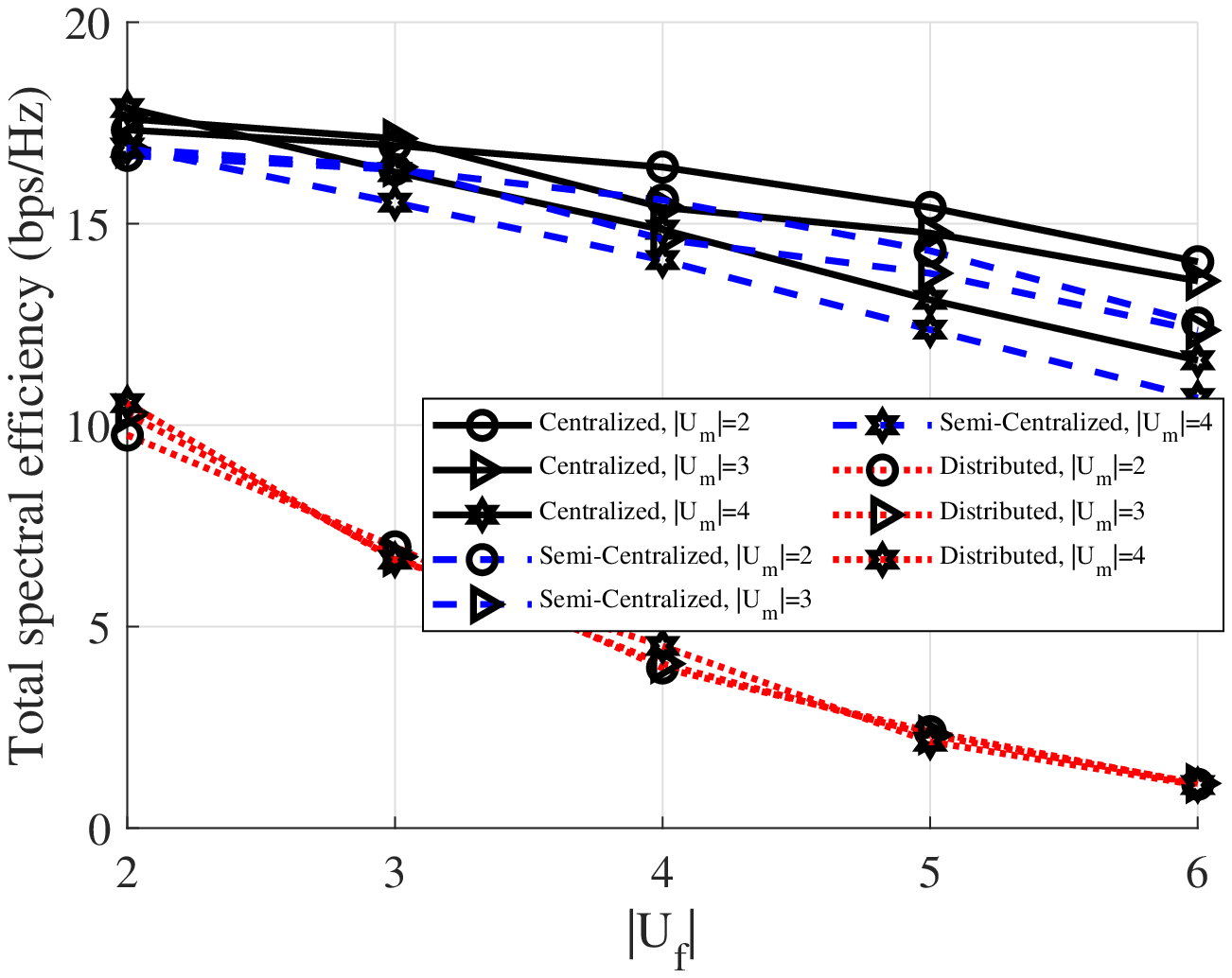}
		\label{Fig_Rtot_FW_usernum}
	}
	\subfigure[Average total spectral efficiency vs. users minimum rate demand for $2$-order NOMA clusters.]{
		\includegraphics[scale=0.52]{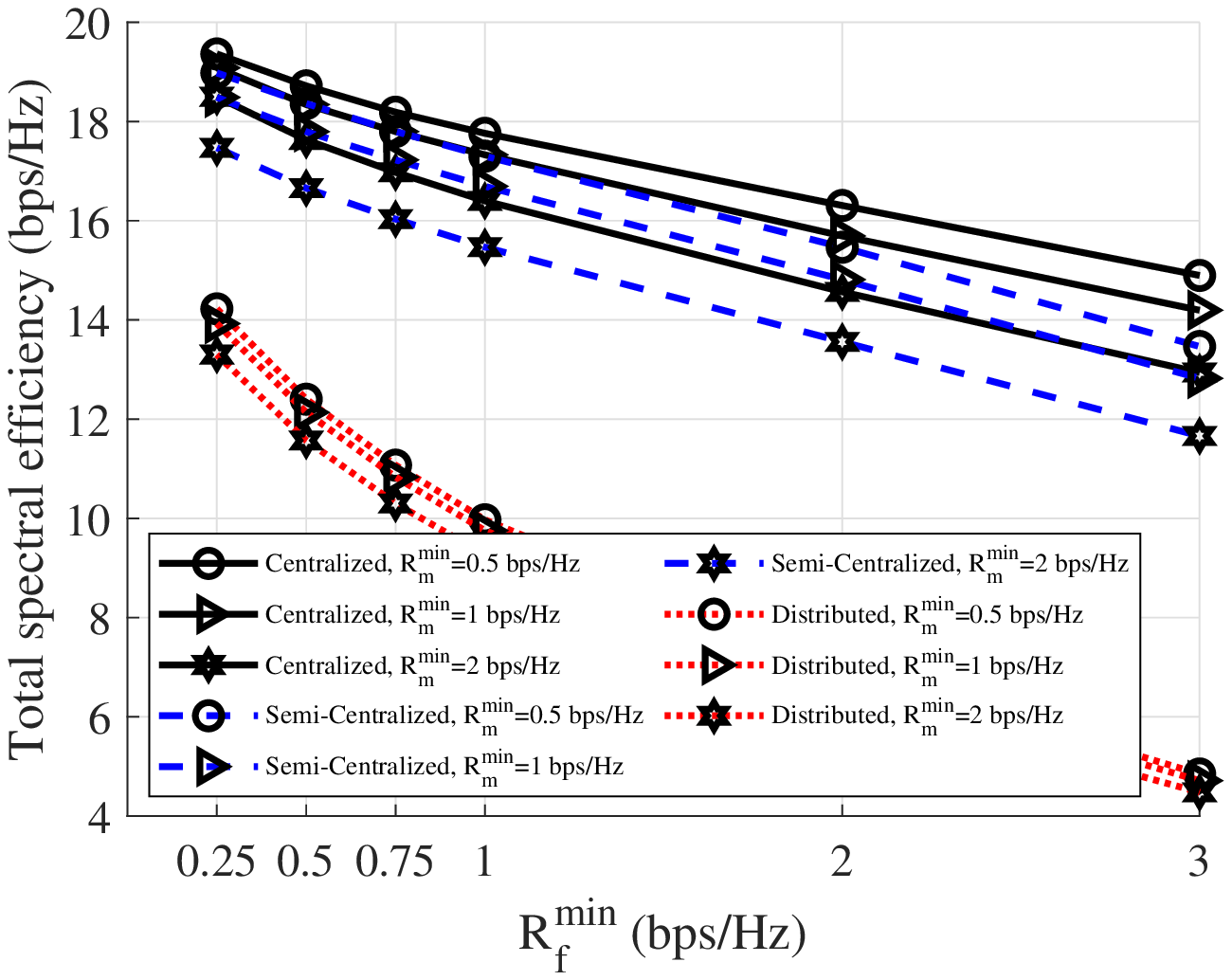}
		\label{Fig_Rtot_FW_minrate_2user}
	}
	\subfigure[Average total spectral efficiency vs. users minimum rate demand for $3$-order NOMA clusters.]{
		\includegraphics[scale=0.52]{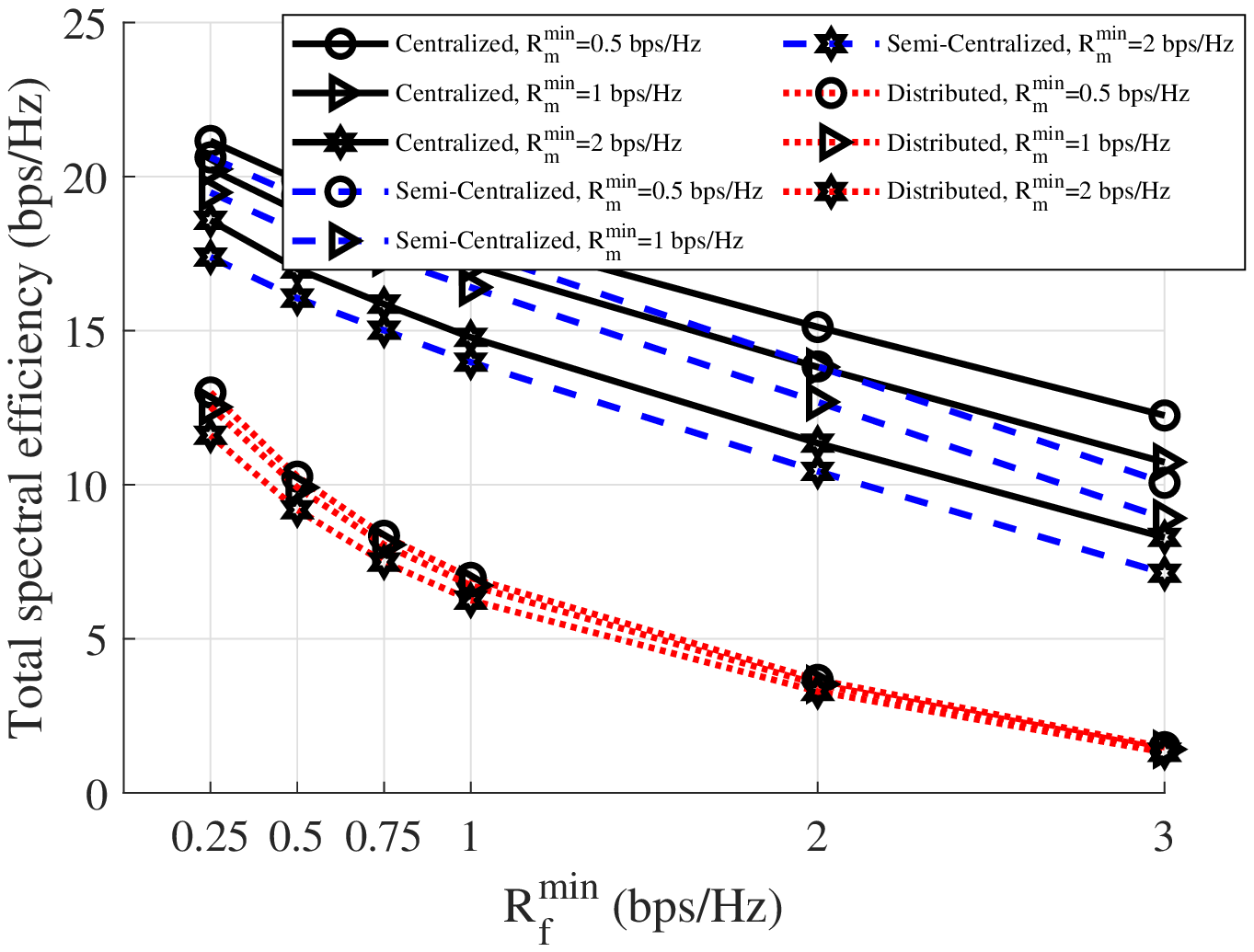}
		\label{Fig_Rtot_FW_minrate_3user}
	}
	\caption
	{Average total spectral efficiency of the centralized and decentralized frameworks for different number of macro/femto-cell users and minimum rate demands.}
	\label{Fig_sumrate_FW}
\end{figure}
According to the discussions for Fig. \ref{Fig_FBSpowercons_FW}, we observed that in most of the cases, the performance gap between the centralized (optimal) and semi-centralized frameworks is quite low, specifically for the lower order of the femto-cell NOMA cluster. Hence, the semi-centralized framework with its low computational complexity (see Table \ref{table complexity}) is a good candidate solution for the larger-scale systems. Besides, the fully distributed framework results in quite low performance, due to the discussions for Fig. \ref{Fig_FBSpowercons_FW}.

\subsection{Convergence of the Iterative Distributed Framework for Solving \eqref{feasible problem}} \label{subsection conya num}
The feasible domain of problem \eqref{feasible problem} is empty if
\begin{enumerate}
	\item Problem \eqref{feasible problem} is infeasible when the maximum power constraint \eqref{Constraint max power} is removed. This corresponds to the feasibility of \eqref{feasible problem} which can be determined by the Perron–Frobenius eigenvalues of the matrices arising from the power control subproblems (see Theorem 8 in \cite{7964738}). In this theorem, it is proved that regardless of the availability of powers, \eqref{feasible problem} can be infeasible, due to the existing ICI and minimum rate demands.
	\item  Problem \eqref{feasible problem} is infeasible while \eqref{feasible problem} without \eqref{Constraint max power} is feasible. As a result, \eqref{feasible problem} is infeasible only because of the lack of power resources to meet the QoS constraints in \eqref{Constraint QoS}.
\end{enumerate}
Since Alg. \ref{Alg opt Mcell powmin} is a component-wise minimum \cite{7964738}, for any feasible $\boldsymbol{p}^{(0)}$, it converges to a unique point which is the globally optimal solution. More importantly, the results show that Alg. \ref{Alg opt Mcell powmin} also converges to the globally optimal solution for any infeasible but finite $\boldsymbol{p}^{(0)}$ if the feasible domain of problem \eqref{feasible problem} is nonempty.
Fig. \ref{Fig_convergence_Yate} shows the convergence of Alg. \ref{Alg opt Mcell powmin} for different initial points.
\begin{figure}[tp]
	\centering
	\subfigure[The network topology and users placement for $|\mathcal{U}_m|=3$, and $|\mathcal{U}_f|=2$.]{
		\includegraphics[scale=0.52]{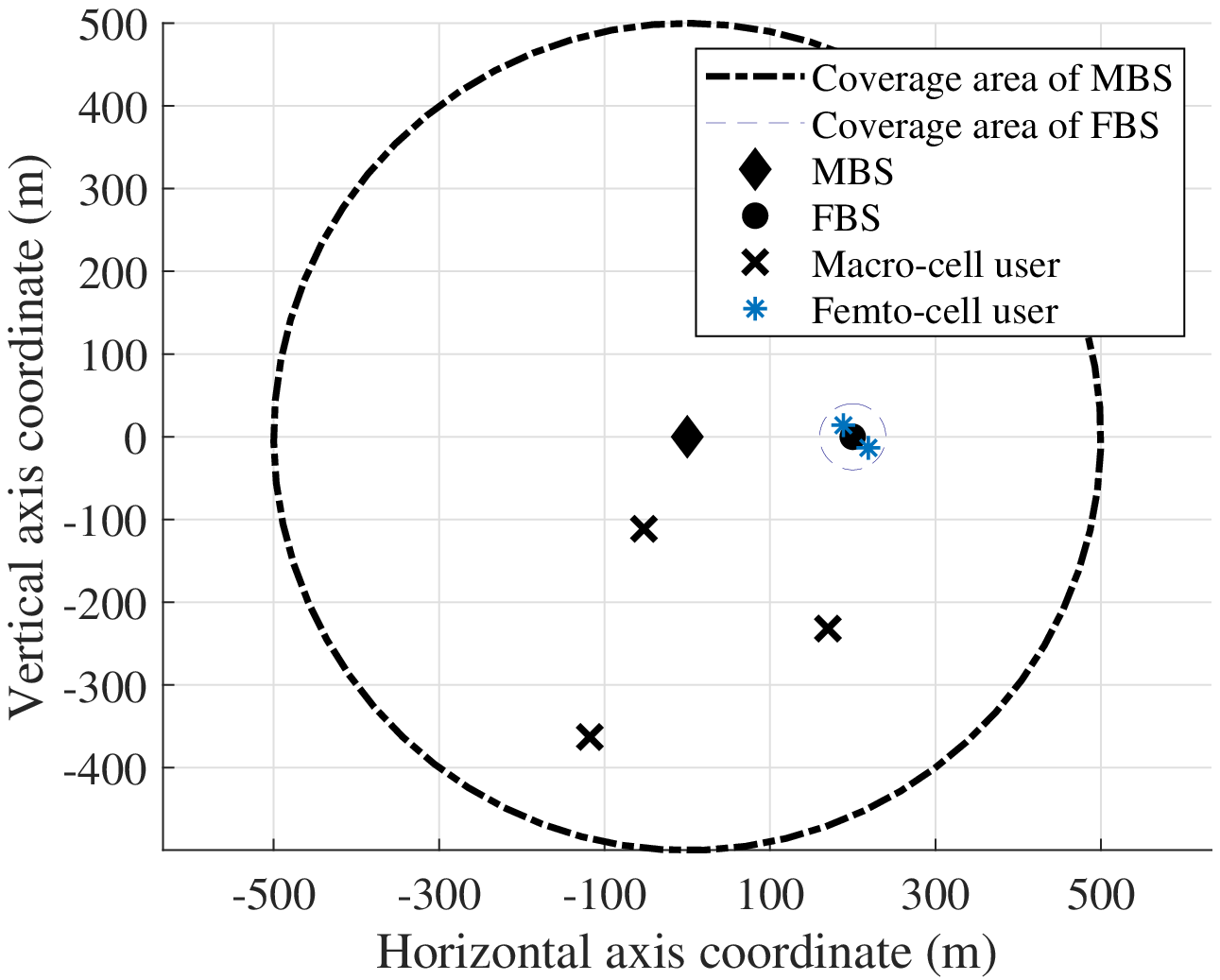}
		\label{Fig_topology_Yate}
	}
	\subfigure[Convergence of the iterative distributed framework for $R^\text{min}_m=1$ bps/Hz, and $R^\text{min}_f=1$.]{
		\includegraphics[scale=0.52]{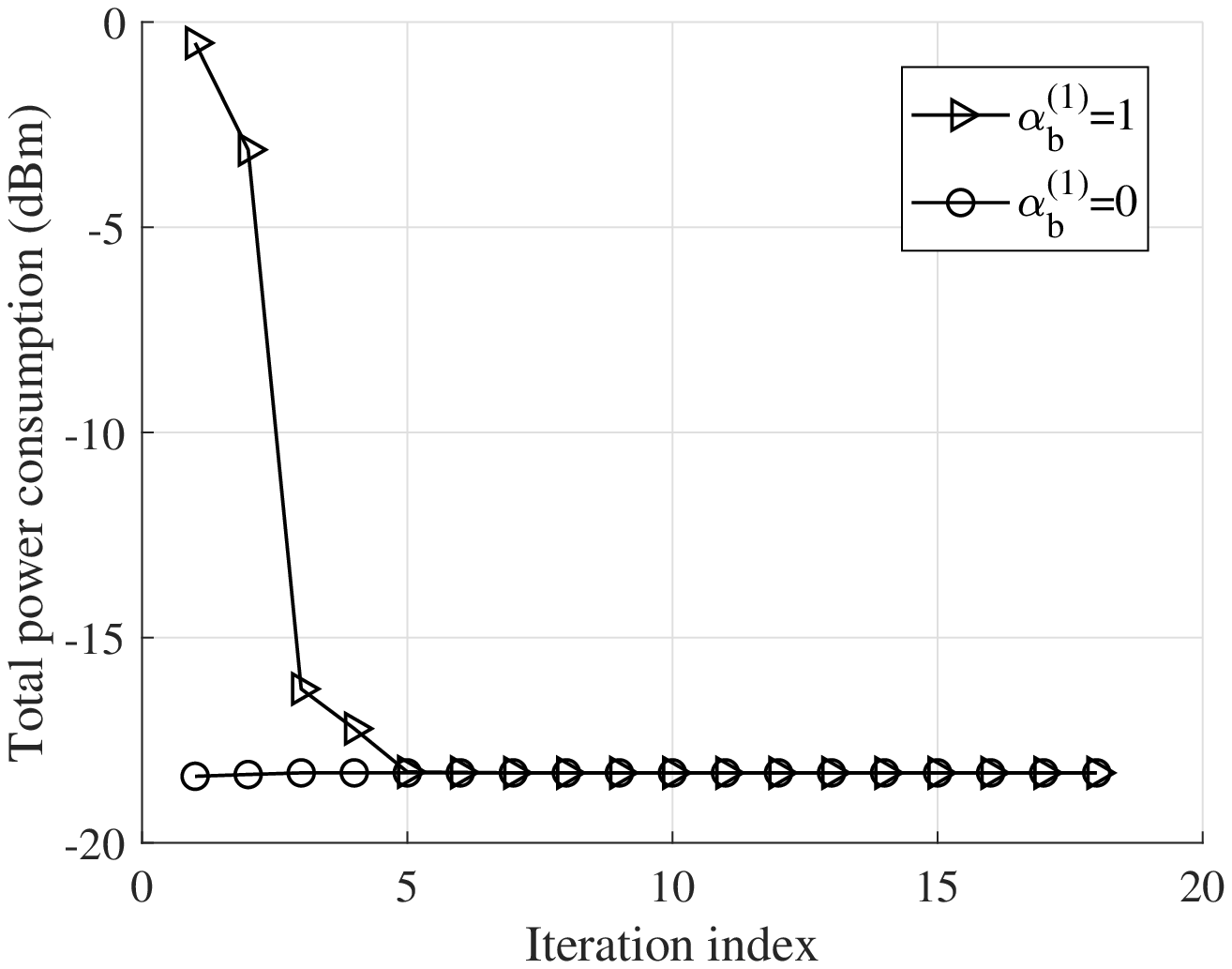}
		\label{Fig_convergence_Yate}
	}
	\subfigure[Divergence of the iterative distributed framework for $R^\text{min}_m=4$ bps/Hz, and $R^\text{min}_f=4$.]{
		\includegraphics[scale=0.52]{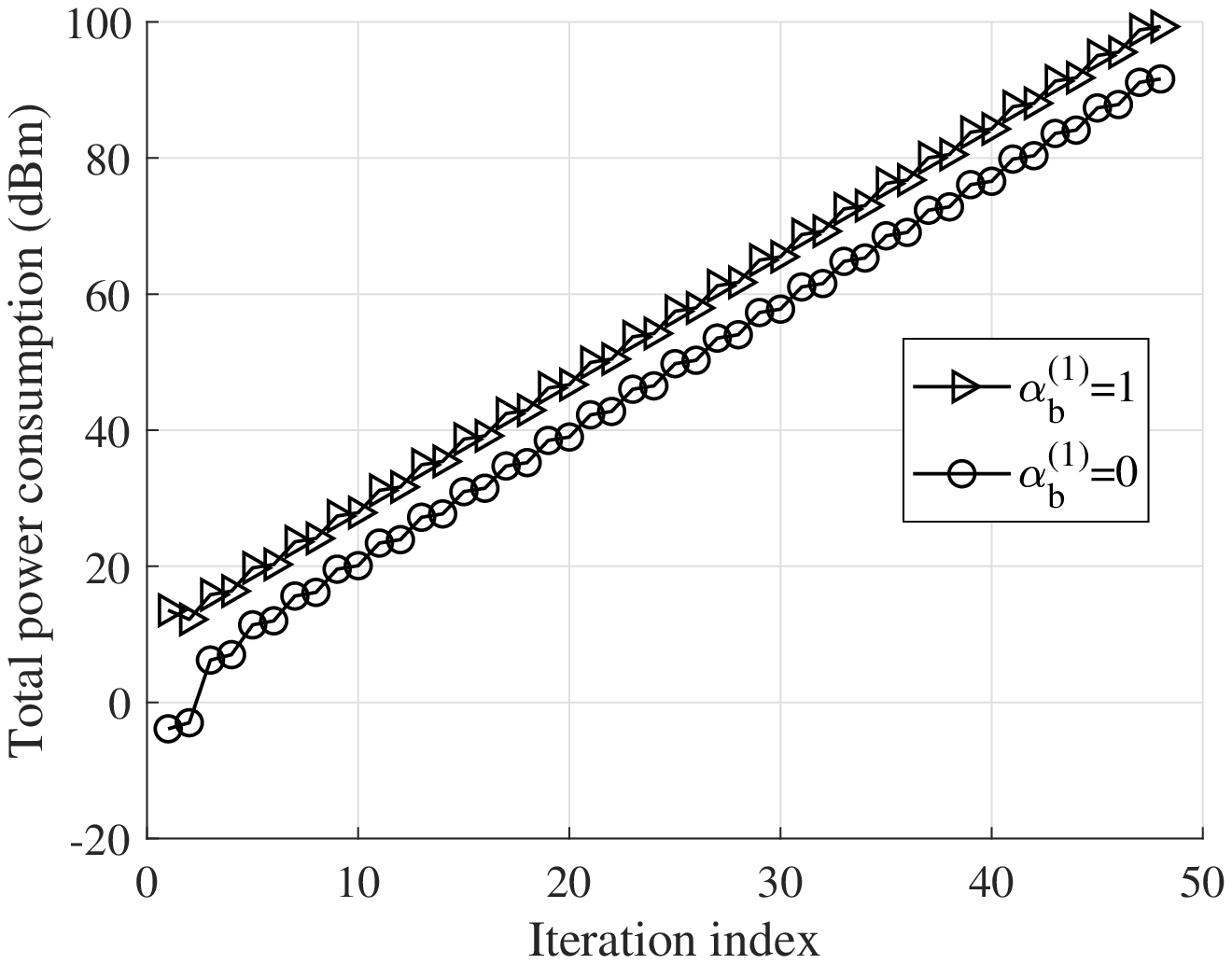}
		\label{Fig_Divergence_Yate}
	}
	\caption
	{Convergence/Divergence of Alg. \ref{Alg opt Mcell powmin} for different initial points, and a channel realization with $|\mathcal{U}_m|=3$, $|\mathcal{U}_f|=2$, and different minimum rate demands.}
	\label{Fig_powmin_yates}
\end{figure}
$\alpha^{(0)}_b=0$ denotes a zero power consumption for the BSs, i.e., $p^{(0)}_{b,i}=0,\forall b\in \mathcal{B},~i \in \mathcal{U}_b$. Besides, $\alpha^{(0)}_b=1$ denotes that the BSs operate in their maximum power at the initial point which may violate \eqref{Constraint QoS}. Fig. \ref{Fig_convergence_Yate} shows that Alg. \ref{Alg opt Mcell powmin} in both the initial points converges to the unique point. For $\alpha^{(0)}_b=0$, the convergence of Alg. \ref{Alg opt Mcell powmin} corresponds tightening the lower-bound of optimal value, since the ICI at each user is always upper-bounded by its ICI at the converged point. Besides, since the ICI at each user reaches to its maximum possible value at $\alpha^{(0)}_b=1$, the convergence of Alg. \ref{Alg opt Mcell powmin} for $\alpha^{(0)}_b=1$ corresponds to tightening the upper-bound of optimal value. Based on the KKT conditions analysis in Appendix \ref{appendix optpower powmin}, we observed that $\boldsymbol{p}^*$ is independent from the maximum power constraint \eqref{Constraint max power}. It can be shown that if problem \eqref{feasible problem} is infeasible while \eqref{feasible problem} without \eqref{Constraint max power} is feasible (Case 2 of the infeasibility reasons of problem \eqref{feasible problem}), Alg. \ref{Alg opt Mcell powmin} will converge to the optimal finite point violating \eqref{Constraint max power}. For larger minimum rate demands (Fig. \ref{Fig_Divergence_Yate}), we observe that Alg. \ref{Alg opt Mcell powmin} diverges and the optimal value tends to infinity, regardless of the maximum power constraint \eqref{Constraint max power}. This corresponds to the first case of the infeasibility reasons of \eqref{feasible problem}. Hence, it is important to check the Perron–Frobenius eigenvalues of the matrices arising from the power control subproblems (see Theorem 8 in \cite{7964738}) before finding a feasible point for \eqref{feasible problem}.
Last but not least, Fig. \ref{Fig_convergence_Yate} verifies a fast convergence speed of Alg. \ref{Alg opt Mcell powmin} for both the initialization methods, however $\alpha^{(0)}_b=0$ converges in less iterations compared to $\alpha^{(0)}_b=1$.

\subsection{Convergence and Performance of the JRPA Algorithm}\label{subsection JRPA converg}
In Fig. \ref{Fig_convergenceJRPA}, we investigate the convergence of our proposed JRPA algorithm which is based on sequential programming. 
\begin{figure}[tp]
	\centering
	\subfigure[Approximation of $y=\ln\left(2^{r}-1\right)$ with the linear function $m \times r$, where $m=y'(r=15)$.]{
		\includegraphics[scale=0.52]{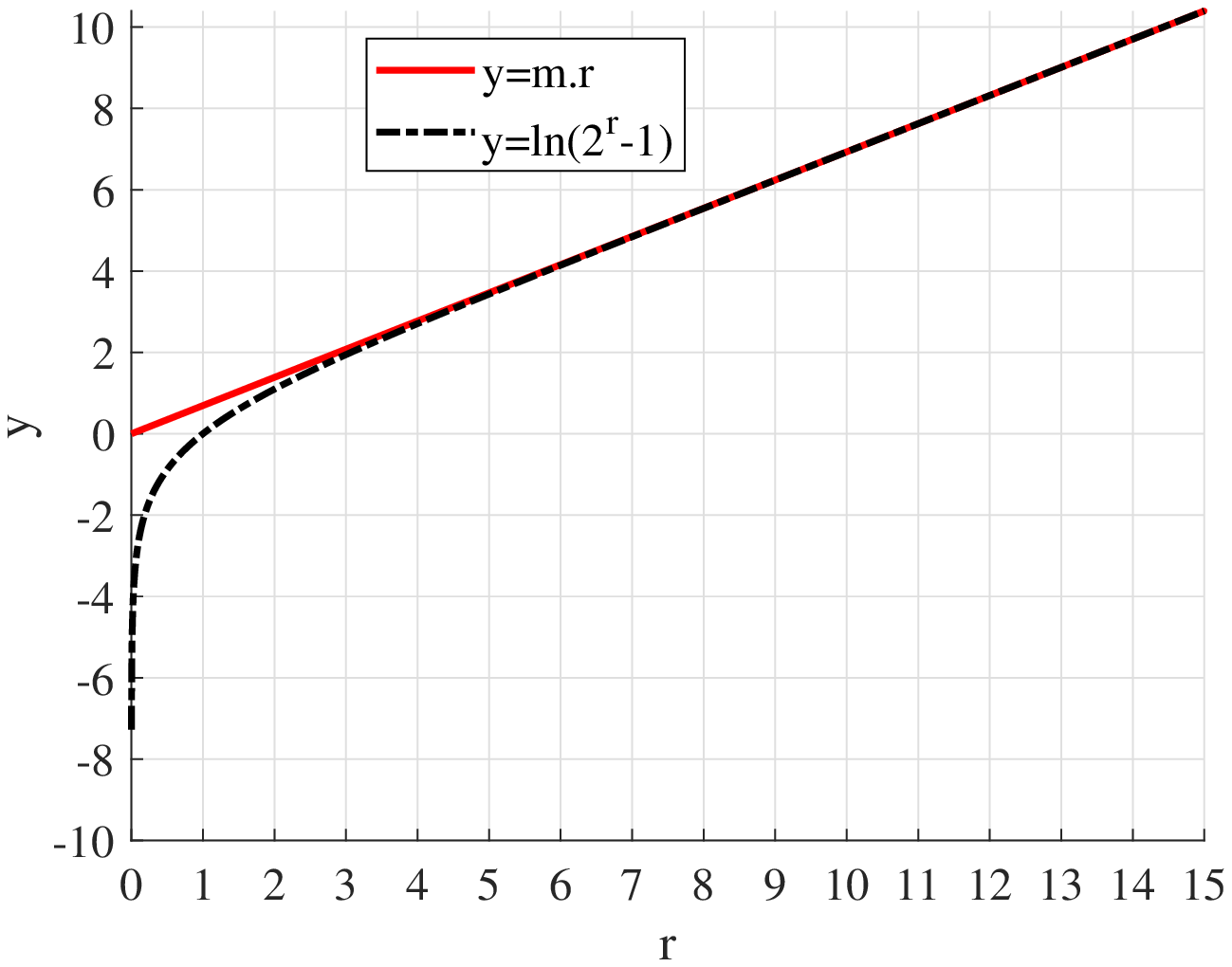}
		\label{Fig_approxrate}
	}
	\subfigure[The network topology and users placement for $|\mathcal{U}_m|=3$, and $|\mathcal{U}_f|=2$.]{
		\includegraphics[scale=0.52]{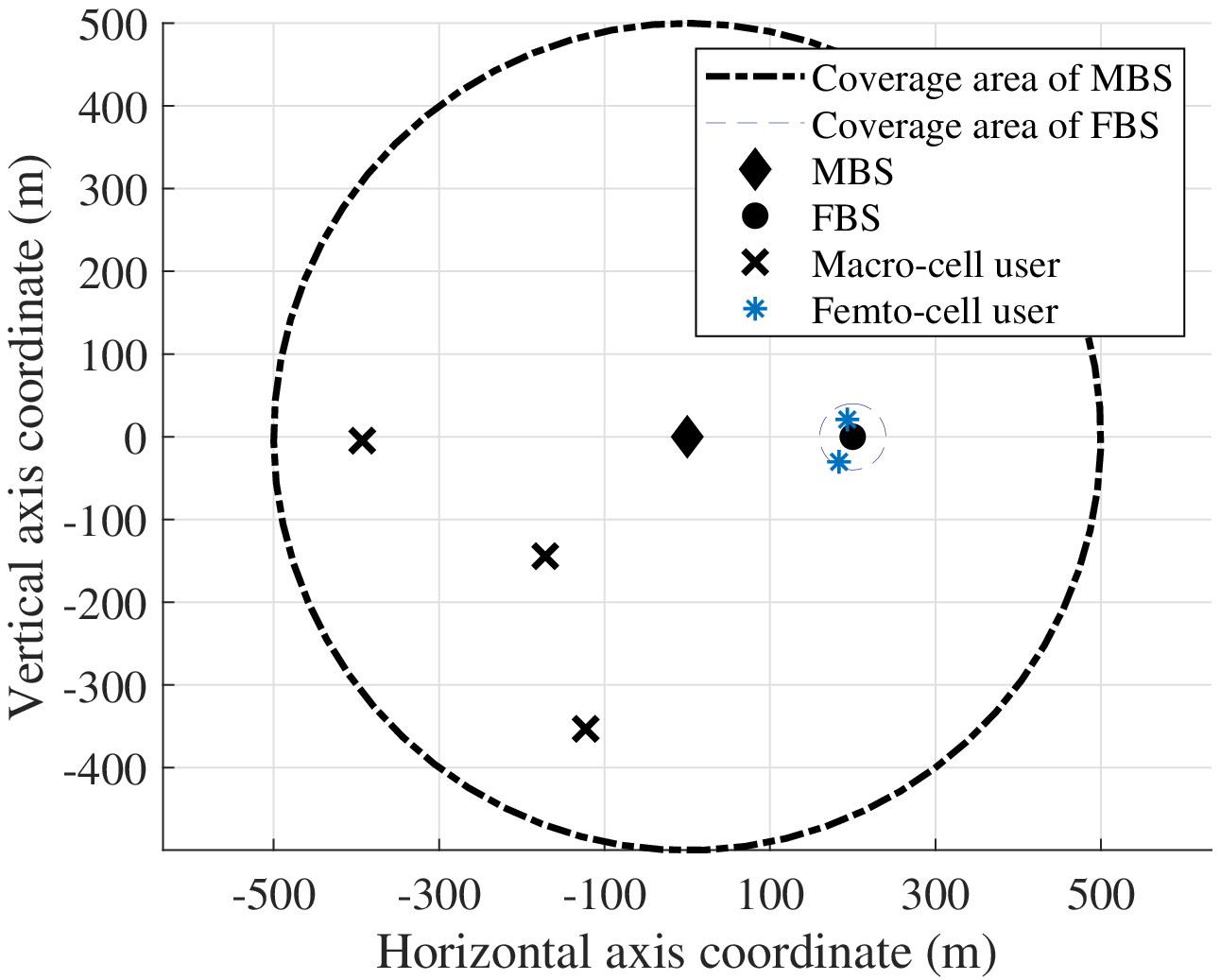}
		\label{Fig_topology_JRPAconverge}
	}
	\subfigure[Total spectral efficiency vs. iteration index for different initialization methods.]{
		\includegraphics[scale=0.52]{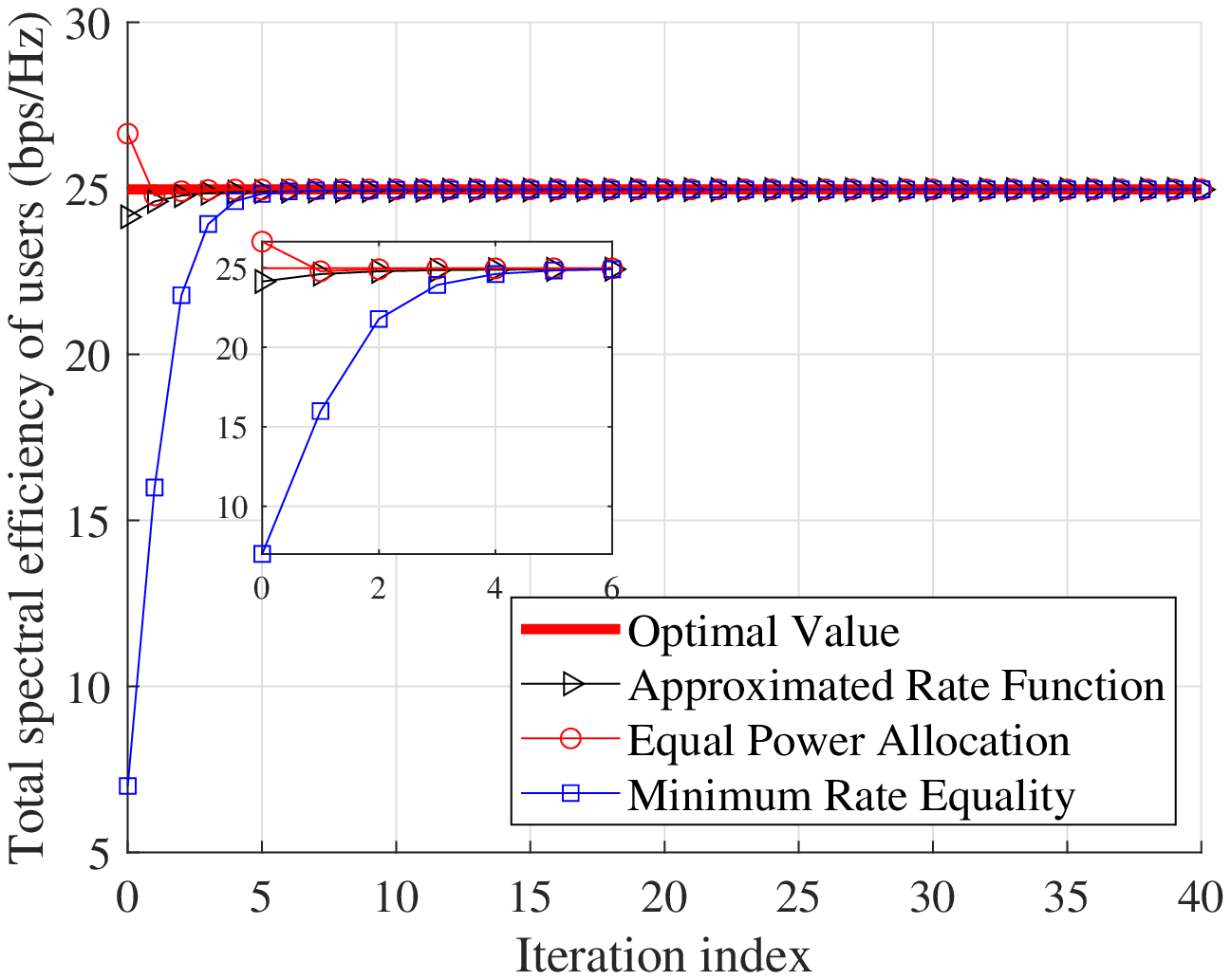}
		\label{Fig_convergenceJRPA_totrate}
	}
	\subfigure[User spectral efficiency vs. iteration index for the MRE initialization method.]{
		\includegraphics[scale=0.52]{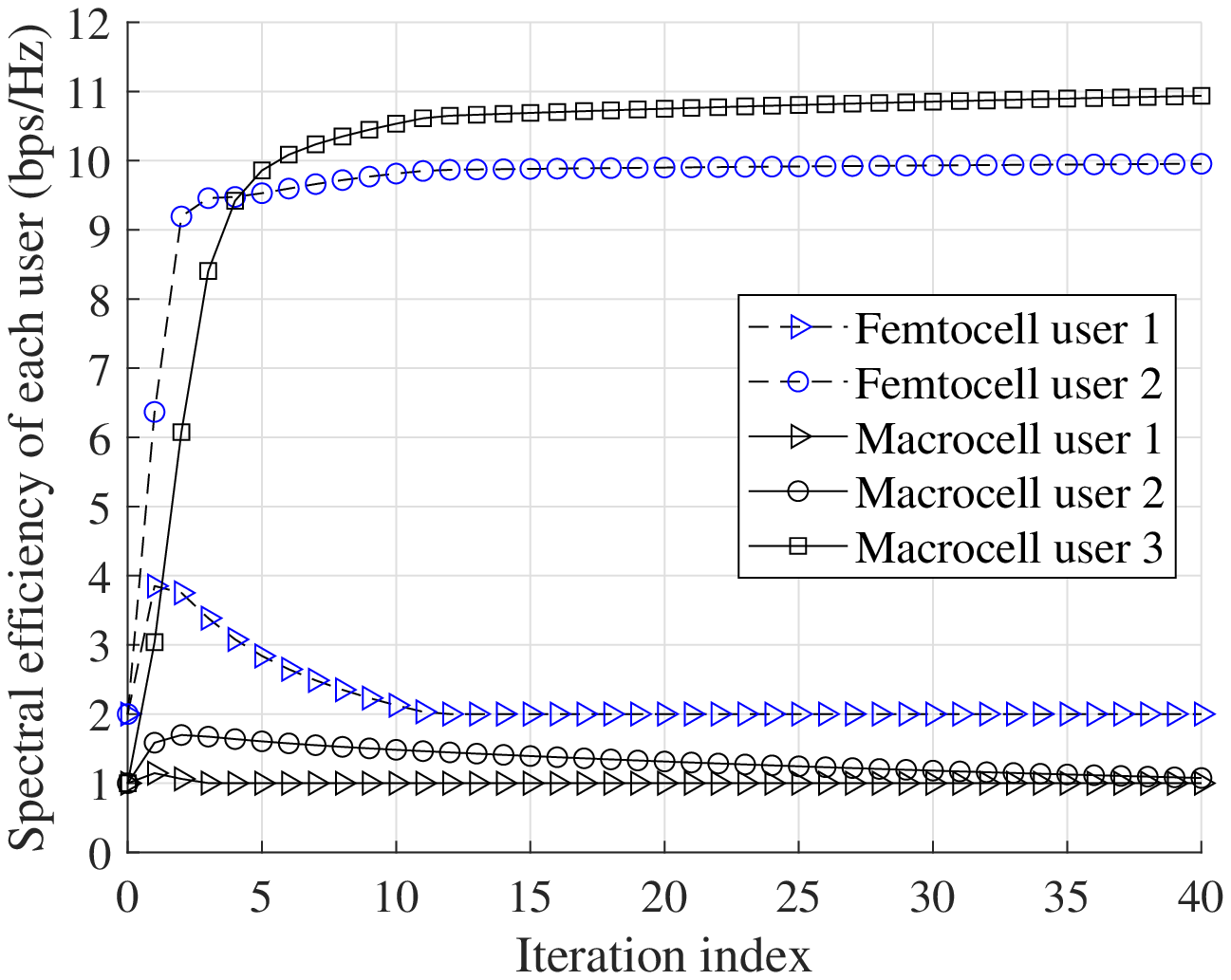}
		\label{Fig_convergenceJRPA_userrate_equalr}
	}
	\subfigure[User spectral efficiency vs. iteration index for the ARF initialization method.]{
		\includegraphics[scale=0.52]{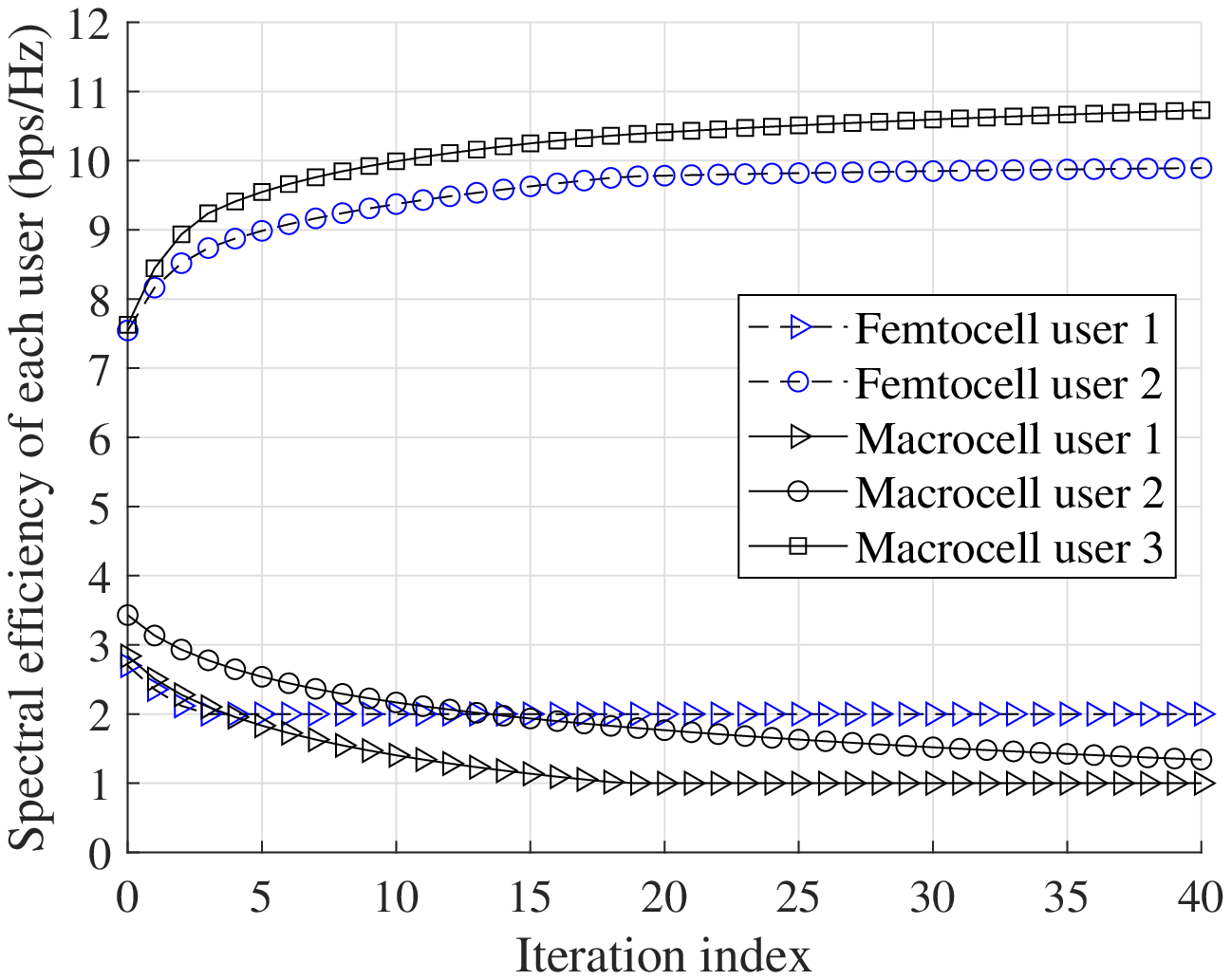}
		\label{Fig_convergenceJRPA_userrate_mr}
	}
	\subfigure[User spectral efficiency vs. iteration index for the EPA initialization method.]{
		\includegraphics[scale=0.52]{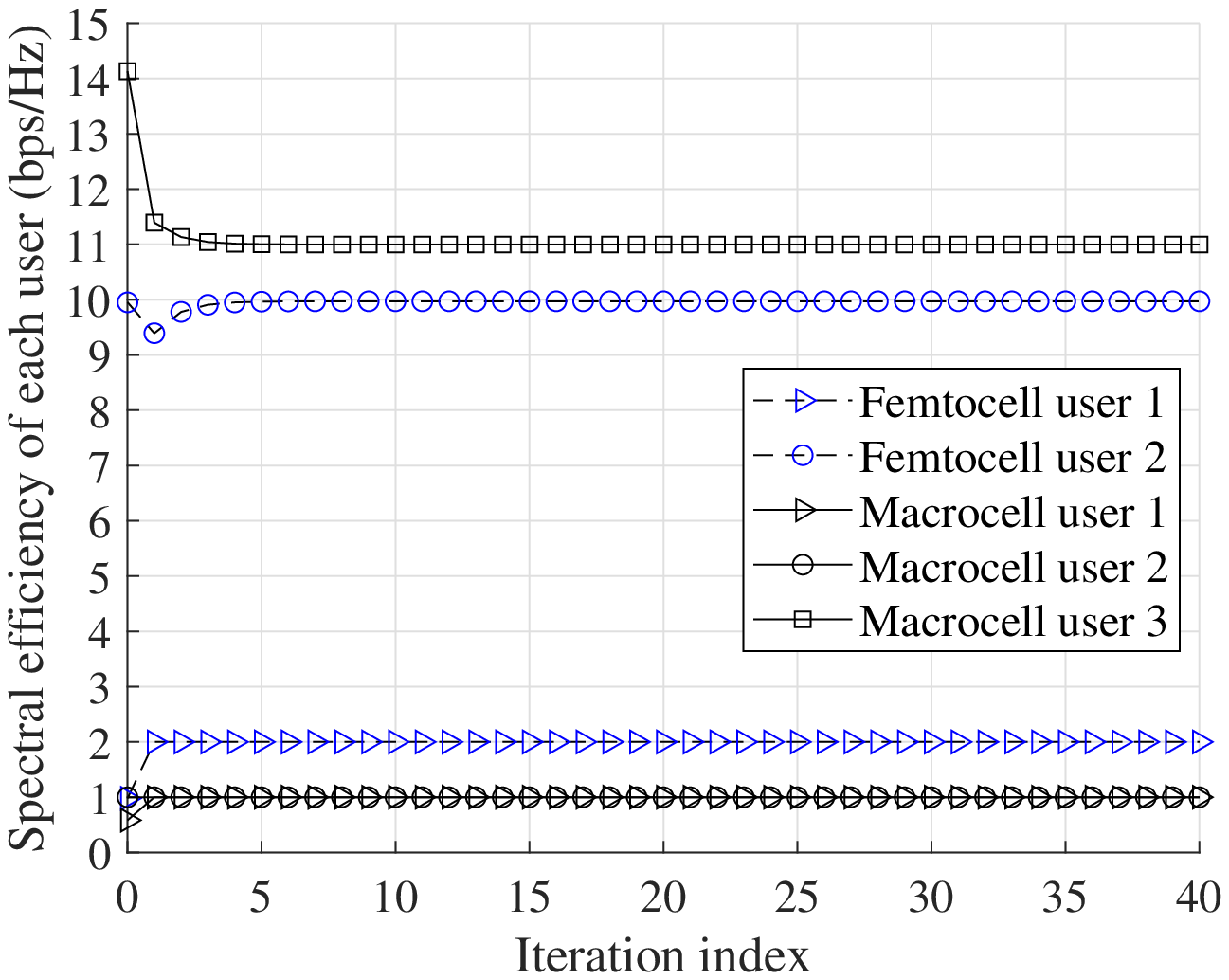}
		\label{Fig_convergenceJRPA_userrate_equalpow}
	}
	\caption
	{Convergence of Alg. \ref{Alg JRPA sequential} for different initialization methods for a scenario with $|\mathcal{U}_m|=3$, $|\mathcal{U}_f|=2$, $R^\text{min}_m=1$ bps/Hz, and $R^\text{min}_f=2$. The CNR-based decoding order is optimal: Macro-cell users: $3 \to 2 \to 1$; Femto-cell users: $2 \to 1$.}
	\label{Fig_convergenceJRPA}
\end{figure}
We assume that the CNR-based decoding order is applied.
Since the sequential programming converges to the locally optimal solution, the initial point may affect the performance of this method. 
In this study, we applied three initialization methods as
\begin{enumerate}
	\item Minimum rate equality (MRE): In this method, we obtain $\boldsymbol{r}^{(0)}$ by solving the total power minimization problem \eqref{subopt powermin problem}. It is proved that at the optimal (feasible) point, the spectral efficiency of each user achieves its minimum rate demand. Hence, we have $r^{(0)}_{b,i}=R^{\text{min}}_{b,i},~\forall b \in \mathcal{B},~i \in \mathcal{U}_b$.
	\item Approximated rate function (ARF): In this method, we substitute 
	the strictly concave term $g(r_{b,i})=\ln\left(2^{r_{b,i}}-1\right)$ with its approximated affine function $m r_{b,i}$ in \eqref{Constraint decoding 1}, where $m=\frac{\partial \ln\left(2^{R}-1\right)}{\partial R}$, where $R$ is significantly large. Then, we solve the convex approximated problem of \eqref{cent subopt problem 1} and obtain $\boldsymbol{r}^{(0)}$. For sufficiently large $R$, $g(r_{b,i})$ is upper-bounded by $m \times r_{b,i}$.
	\item Equal power allocation (EPA): In this method, we equally distribute $P^{\text{max}}_b$ to all the users in $\mathcal{U}_b$, and then obtain $\boldsymbol{r}^{(0)}$ according to \eqref{useri Mcell}. This method may lead to an infeasible $\boldsymbol{r}^{(0)}$.
\end{enumerate}
MRE provides a feasible $\boldsymbol{r}^{(0)}$. However, this method does not consider the heterogeneity of users spectral efficiency, leading to larger convergence speed and in some situations lower performance. MRE works well for the low-SINR scenarios with significantly high minimum rate demands. The ARF method provides a better feasible lower-bound for the total spectral efficiency of users at the initial point. Fig. \ref{Fig_approxrate} shows that for larger minimum rate demands, $\ln\left(2^{r}-1\right) \approx m r$. The performance gap of ARF and the globally optimal solution is allocating more powers to users operating in low spectral efficiency regions, which results in allocating less power to the stronger user deserving additional power. ARF also works well for the scenarios that the low additional minimum rate demands does not have significant impact on the users total spectral efficiency, i.e., high SINR regions. The EPA initialization method usually leads to infeasible $\boldsymbol{r}^{(0)}$ violating \eqref{Constraint QoS}, due to INI and ICI at users. More importantly, we observed that EPA also leads to high outage at the next iteration of the JRPA algorithm. In Fig. \ref{Fig_convergenceJRPA}, we selected the scenario that EPA (violating \eqref{Constraint QoS}) does not make the next iteration infeasible to show the convergence behavior of this initialization method.

The users placement are shown in Fig. \ref{Fig_topology_JRPAconverge}. 
As shown, the JRPA provides a sequence of improved solutions for any feasible initial point such that it converges to a stationary point. Interestingly, we observe that both the MRE and ARF methods converge to a unique point, which shows the low insensitivity of JRPA to these feasible initial points. In this scenario, EPA in iteration $0$ results in infeasible $\boldsymbol{r}^{(0)}$. And, JRPA finds a feasible solution $\boldsymbol{r}^{(1)}$ based on infeasible $\boldsymbol{r}^{(0)}$, which is indeed the updated initial feasible point. 
According to Figs. \ref{Fig_convergenceJRPA_userrate_equalr}-\ref{Fig_convergenceJRPA_userrate_equalpow}, we observe that at the converged point, only the NOMA cluster-head users get additional power (leading to higher spectral efficiency than their minimum rate demand). The fast convergence speed of individual rates in Figs. \ref{Fig_convergenceJRPA_userrate_equalr}-\ref{Fig_convergenceJRPA_userrate_equalpow} shows that our proposed JRPA has a fact convergence speed shown in Fig. \ref{Fig_convergenceJRPA_totrate}. In our simulations, we applied the ARF initialization method.

As is mentioned in Subsection \ref{subsection subopt order}, it is difficult to find the globally optimal JRPA for any fixed suboptimal decoding order. However, for the case that the fixed decoding order is the same as the optimal decoding order, the performance gap between the optimal JSPA and suboptimal JRPA algorithms is due to suboptimal JRPA based on sequential programming. In Fig. \ref{Fig_convergenceJRPA}, we choosed the case that the CNR-based decoding order satisfy Theorem \ref{Theorem SIC M-cell positiveterm}, so is optimal. The total spectral efficiency of users (optimal value) at the globally optimal point is shown in Fig. \ref{Fig_convergenceJRPA_totrate}. As can be seen, the sequential programming generates a sequence of improved solutions such that after few iterations, it converges to a near-optimal solution.

\subsection{Performance of the Approximated Optimal Powers in Remark \ref{remark approx power}.}\label{subsection approx opt power}
Here, we investigate the performance of approximated closed-form of optimal powers in Remark \ref{remark approx power}. The advantage of this approximation is its insensitivity to the exact CSI. Fig. \ref{Fig_apprxoptgap} compares the average gap between the exact and approximated forms of optimal powers in femto and macro-cells, separately.
\begin{figure}[tp]
	\centering
	\subfigure[Average optimal and approximated powers gap vs. AWGN power at macro-cell users.]{
		\includegraphics[scale=0.52]{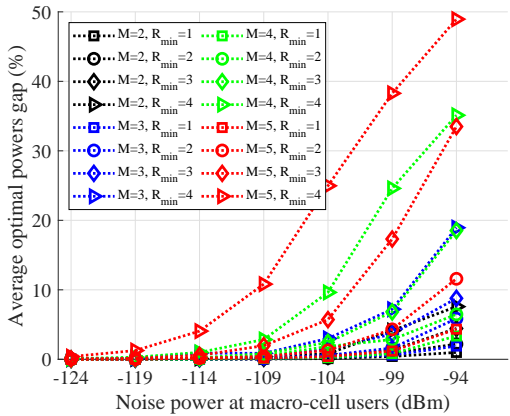}
		\label{Fig_Macro_PowApprox}
	}
	\subfigure[Average total spectral efficiency gap vs. AWGN power at macro-cell users.]{
		\includegraphics[scale=0.52]{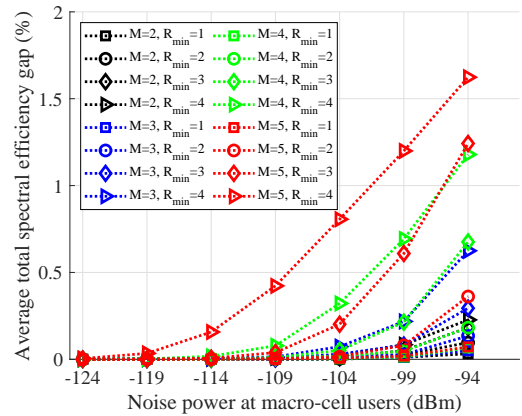}
		\label{Fig_Macro_RateApprox}
	}
	\subfigure[Average optimal and approximated powers gap vs. AWGN power at femto-cell users.]{
		\includegraphics[scale=0.52]{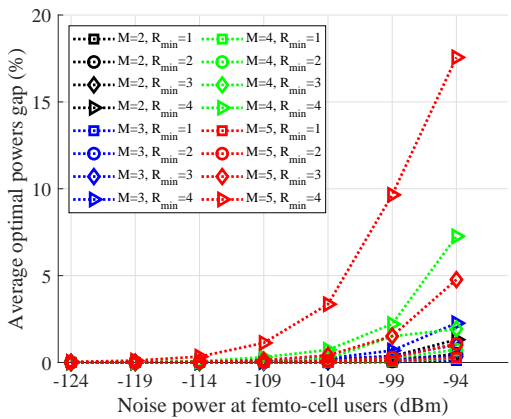}
		\label{Fig_Femto_PowApprox}
	}
	\subfigure[Average total spectral efficiency gap vs. AWGN power at femto-cell users.]{
		\includegraphics[scale=0.52]{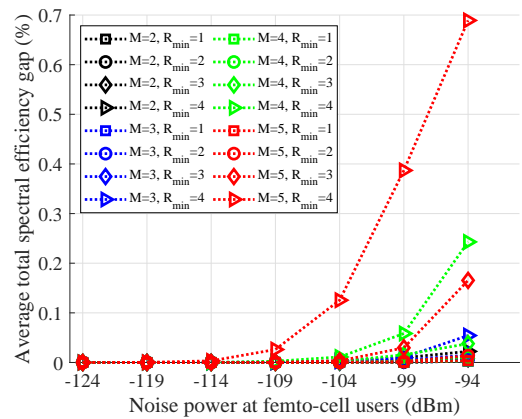}
		\label{Fig_Femto_RateApprox}
	}
	\caption
	{The performance gap between the approximated and exact closed-form expression of optimal powers in macro/femto-cell for different AWGN powers, number of macro/femto-cell users, and minimum rate demands.}
	\label{Fig_apprxoptgap}
\end{figure}
Since ICI is fully treated as AWGN, we evaluate this performance gap in different AWGN power levels of users which directly impacts the CINR of users.
In Fig. \ref{Fig_apprxoptgap}, $M$ is the number of users within the considered cell, and $R_\text{min}$ is their minimum rate demand. To reduce the randomness impact, we eliminated the Lognormal Shadowing from the path loss model (see Table \ref{Table parameters}).
As can be seen in Figs. \ref{Fig_Macro_PowApprox} and \ref{Fig_Femto_PowApprox}, the average gap between the optimal and approximated powers is increasing in the order of NOMA clusters, minimum rate demands, and specifically AWGN power. Interestingly, for lower AWGN powers, e.g., less than $-119$ dBm, this performance gap tends to zero. As a result, this approximation works well for middle and high SINR scenarios. On the other hand, we observe that the macro/femto-cell total spectral efficiency gaps between these two closed-form formulations are less than $1.5\%$ for $M\leq 4$, $R_\text{min}\leq 4$ bps/Hz, and AWGN power less than $-94$ dBm. Hence, the results show a high insensitivity level of optimal powers at macro/femto-cell users to the CSI. The impact of the approximated closed-form optimal powers on the ergodic rate regions and/or imperfect CSI can be considered as a future work.

\section{Concluding Remarks and Future Research Directions}\label{sec conclusion}
In this paper, we addressed the problem of optimal joint SIC ordering and power allocation in multi-cell NOMA systems to achieve the maximum users sum-rate. For the given total power consumption of BSs, we obtained the optimal powers in closed form. Then, we proposed a globally optimal JSPA algorithm with a significantly reduced computational complexity. For any fixed decoding order, we addressed the problem of joint rate and power allocation to maximize users sum-rate. We showed that for specific channel conditions, the CNR-based decoding order is optimal for a user pair independent from ICI. We also devised two decentralized resource allocation frameworks. Extensive numerical results are provided to evaluate the outage probability/sum-rate gains achieved by JSPA compared to FRPA and JRPA with CNR-based decoding order. The numerical results show that FRPA leads to a very high outage when the number of multiplexed users increases. JRPA has a near-optimal performance, however, there is still a certain performance gap due to suboptimal CNR-based decoding order. Moreover, the fully distributed framework has a poor performance, because of high ICI from MBS to FBS-users. The semi-centralized framework has a near-optimal performance with significant reduced complexity, such that it scales very well with larger number of FBSs and users. The fast semi-centralized algorithm is a good solution for practical implementations.

A list of possible extensions of this work is provided in the following.
\begin{enumerate}
	\item \textbf{Multi-Cluster NOMA}: The analysis of this work can be directly applied to NOMA with multiple clusters, where each user occupies one subchannel, and each subchannel has specific power budget, e.g., the system models in \cite{7557079,8114362,8352643}. For the case that the power budget of clusters is not predefined \cite{7982784}, an additional inter-cluster power allocation is required.
	\item \textbf{Hybrid NOMA}: In Hybrid NOMA, also known as multicarrier NOMA, each user can occupy more than one subcarrier. For the case that different symbols are sent over different subcarriers, the analysis of this paper can be applied, however more efforts are required to tackle the nonconvex per-user minimum rate constraints over all the assigned subcarriers. Another challenge is subcarrier allocation. The optimal joint power and subcarrier allocation is proved to be strongly NP-hard \cite{7587811,8422362,9044770}.
	\item \textbf{Decentralized Recourse Allocation Frameworks:} Our semi-centralized framework is specified for a two-tier HetNet including a single MBS and arbitrary number of FBSs and users. A more complicated scheme is when we have more than one MBS each having arbitrary number of FBSs and users. One future work would be how to generalize the analysis in this paper to a $N$-tier wireless network with efficient power control between the tiers, as well as power control among BSs within each tier.
	\item \textbf{Energy Efficiency Maximization Problem:} Another important direction of the closed-form of optimal powers obtained in this paper is energy efficiency maximization problem in both the single-cluster and multi-cluster NOMA.
	\item \textbf{NOMA with Imperfect CSI:} We showed that the approximated closed-form of optimal powers which are completely independent from channel gains are nearly optimal. However, the decoding order among users depends on the channel gains. One interesting topic would be how to find a robust decoding order for imperfect CSI while applying the approximated closed-form of optimal powers derived in this paper.
	\item \textbf{Applications of Theorem 1. SIC Sufficient Condition in Multi-Cell NOMA:} Since dynamically changing decoding orders based on CINR of users through power control among the cells is still challenging for sum-rate maximization problem, there are three schemes which can be analyzed and numerically compared in multi-cell multicarrier NOMA:
	\begin{enumerate}
		\item Considering CNR-based decoding order with FRPA scheme.
		\item Considering CNR-based decoding order and imposing Theorem 1 in user grouping (similar to applying Lemma 1 in user grouping in \cite{8353846}).
		\item Considering CNR-based decoding order and applying JRPA algorithm.
	\end{enumerate}
	In case (c), we consider the complete rate region of users instead of limiting user grouping by imposing Theorem 1 in case (b).
\end{enumerate}
It is worth noting that the analysis in this work as well as all the existing works on single-antenna NOMA cannot be generalized to the multi-antenna systems which are nondegraded in general \cite{1683918}.
The optimal successive decoding order criteria in nondegraded multi-antenna BCs depends on beamforming. Moreover, when SC-SIC is applied to the multi-antenna BCs, the users may operate in lower rate than their channel capacity for decoding their own signals \cite{1683918}. The analysis shows that NOMA looses significant multiplexing gain in multi-antenna systems \cite{9451194}.

\appendices
\section{Additional Notes on the Optimality of CINR-Based Decoding Order.}\label{appendix corol optorder}
Let $\tilde{h}_{b,k}=\frac{h_{b,k}}{I_{b,k} + \sigma^2_{b,k}}$. The rate function in \eqref{useri Mcell} can be reformulated as $$R_{b,i}=\min\limits_{k \in \{i\} \cup \Phi_{b,i}}\left\{\log_2\left( 1+ \frac{p_{b,i} \tilde{h}_{b,k}}{\sum\limits_{j \in \Phi_{b,i}} p_{b,j} \tilde{h}_{b,k} + 1} \right)\right\},$$ which is the same as the achievable rate of single-cell NOMA based on the equivalent noise. Since the SISO Gaussian BCs are degraded \cite{10.5555/1146355,NIFbook,1683918}, NOMA with CINR-based decoding order is capacity-achieving in each cell of multi-cell NOMA, so the decoding order $\tilde{h}_{b,k} > \tilde{h}_{b,i} \Rightarrow k \to i$ is optimal. 

For the sake of completeness, we provide a proof for the optimality of the CINR-based decoding order in multi-cell NOMA with single-cell processing. Similar to single-cell NOMA, in each cell $b$, we have $\frac{\partial \log_2\left( 1+\frac{p_{b,k} \tilde{h}_{b,k}}{\sum\limits_{j \in \Phi_{b,i}} p_{b,j} \tilde{h}_{b,k} + 1}\right)} {\partial \tilde{h}_{b,k}} > 0$. Assume that $\boldsymbol{\boldsymbol{\alpha}_{-b}}$ is fixed. In the following, we analytically show that at any given $\boldsymbol{p}_b$, the decoding order based on $\tilde{h}_{b,k}(\boldsymbol{\alpha}_{-b})>\tilde{h}_{b,i}(\boldsymbol{\alpha}_{-b}) \Rightarrow k \to i$ achieves the maximum total spectral efficiency of users. Assume that cell $b$ has $M$ users. Moreover, the users index are updated based on $k > i$ if $\tilde{h}_{b,k}(\boldsymbol{\alpha}_{-b})>\tilde{h}_{b,i}(\boldsymbol{\alpha}_{-b})$. We prove that the decoding order $M \to M-1 \to \dots \to 1$ outperforms any other possible decoding orders in terms of total spectral efficiency of users, so is optimal. To prove this, for two adjacent users $i$ and $i+1$ (with $\tilde{h}_{b,i+1}>\tilde{h}_{b,i}$) consider different decoding orders as A) $i+1 \to i$, and subsequently $M \to M-1 \to \dots \to 1$; B) $i \to i+1$, and subsequently $M \to M-1 \to \dots \to i+2 \to i \to i+1 \to i-1 \to \dots \to 1$. 
According to \eqref{useri Mcell}, the achievable spectral efficiency of users $i$ and $i+1$ in case A (after SIC) can be obtained by
$$
R^{A}_{b,i}(\boldsymbol{p}_b)=\log_2\left( 1+ \frac{p_{b,i} \tilde{h}_{b,i}}{\sum\limits_{j=i+1}^{M} p_{b,j} \tilde{h}_{b,i} + 1} \right),~~~~~~~R^{A}_{b,i+1}(\boldsymbol{p}_b)=\log_2\left( 1+ \frac{p_{b,i+1} \tilde{h}_{b,i+1}}{\sum\limits_{j=i+2}^{M} p_{b,j} \tilde{h}_{b,i+1} + 1} \right).
$$
The achievable spectral efficiency of users $i$ and $i+1$ in case B (after SIC) is given by
$$
R^{B}_{b,i}(\boldsymbol{p}_b)=\log_2\left( 1+ \frac{p_{b,i} \tilde{h}_{b,i}}{\sum\limits_{j=i+2}^{M} p_{b,j} \tilde{h}_{b,i} + 1} \right),~~~~~~~R^{B}_{b,i+1}(\boldsymbol{p}_b)=\log_2\left( 1+ \frac{p_{b,i+1} \tilde{h}_{b,i}}{(p_{b,i} + \sum\limits_{j=i+2}^{M} p_{b,j}) \tilde{h}_{b,i} + 1} \right).
$$
In both Cases A and B, the signal of users $i$ and $i+1$ is treated as INI at users $1,\dots,i-1$. Moreover, the signal of users $i$ and $i+1$ is scheduled to be decoded and canceled by all the users $i+2,\dots,M$. Hence, the set $\Phi_{b,k}$ of each user $k \in \{1,\dots,i-1,i+2,\dots,M\}$ is the same in both Cases A and B, resulting in the same spectral efficiency formulated in \eqref{useri Mcell}. Since the signal of all the users in $\mathcal{U}_b$ is fully treated as noise (called ICI), the set $\Phi_{b',i}$ of each user $i \in \mathcal{U}_{b'},~b' \in \mathcal{B}\setminus\{b\}$ is the same in Cases A and B resulting in the same spectral efficiency. Accordingly, changing the decoding order of two adjacent users only changes the achievable rate of these two decoding orders for the given $\boldsymbol{p}_b$. As a result, the total spectral efficiency gap between different decoding orders A and B for the given $\boldsymbol{p}_b$ can be formulated by
\begin{align*}
	&R^{A-B}_\text{gap} (\boldsymbol{p}_b)= \sum\limits_{b \in \mathcal{B}} \sum\limits_{i \in \mathcal{U}_b} R^{A}_{b,i} (\boldsymbol{p}_b) - \sum\limits_{b \in \mathcal{B}} \sum\limits_{i \in \mathcal{U}_b} R^{B}_{b,i} (\boldsymbol{p}_b) 
	\\
	=& \left(R^{A}_{b,i} (\boldsymbol{p}_b) + R^{A}_{b,i+1} (\boldsymbol{p}_b)\right) - \left(R^{B}_{b,i} (\boldsymbol{p}_b) + R^{B}_{b,i+1} (\boldsymbol{p}_b)\right)
	\\
	=& \log_2\left( \frac
	{ \left( 1+\sum\limits_{j=i}^{M} p_{b,j} \tilde{h}_{b,i} \right) \left( 1+\sum\limits_{j=i+1}^{M} p_{b,j} \tilde{h}_{b,i+1} \right) }
	{ \left( 1+\sum\limits_{j=i+1}^{M} p_{b,j} \tilde{h}_{b,i} \right) \left( 1+\sum\limits_{j=i+2}^{M} p_{b,j} \tilde{h}_{b,i+1} \right) } \right) 
	+ 
	\\
	&~~~~~~~~~~~~~~~
	\log_2\left( \frac
	{ \left( 1 + \sum\limits_{j=i+2}^{M} p_{b,j} \tilde{h}_{b,i} \right) \left(1+(p_{b,i} + \sum\limits_{j=i+2}^{M} p_{b,j}) \tilde{h}_{b,i}\right) }
	{ \left( 1 + (p_{b,i} + \sum\limits_{j=i+2}^{M} p_{b,j}) \tilde{h}_{b,i} \right) \left( 1+\sum\limits_{j=i}^{M} p_{b,j} \tilde{h}_{b,i} \right) } \right)
	\\
	=&\log_2\left( \frac
	{ \left(1+\sum\limits_{j=i+1}^{M} p_{b,j} \tilde{h}_{b,i+1}\right) \left(1+\sum\limits_{j=i+2}^{M} p_{b,j} \tilde{h}_{b,i}\right) }
	{ \left(1+\sum\limits_{j=i+1}^{M} p_{b,j} \tilde{h}_{b,i}\right) \left(1+\sum\limits_{j=i+2}^{M} p_{b,j} \tilde{h}_{b,i+1}\right) } \right)
	\\
	=&\log_2\left( \frac
	{ 1+\left(\sum\limits_{j=i+1}^{M} p_{b,j} \right) \tilde{h}_{b,i+1} + \left(\sum\limits_{j=i+2}^{M} p_{b,j} \right) \tilde{h}_{b,i} + \left(\sum\limits_{j=i+2}^{M} p_{b,j} \right) \left(\sum\limits_{j=i+1}^{M} p_{b,j} \right) \tilde{h}_{b,i} \tilde{h}_{b,i+1} }
	{ 1+\left(\sum\limits_{j=i+1}^{M} p_{b,j} \right) \tilde{h}_{b,i} + \left(\sum\limits_{j=i+2}^{M} p_{b,j} \right) \tilde{h}_{b,i+1} + \left(\sum\limits_{j=i+2}^{M} p_{b,j} \right) \left(\sum\limits_{j=i+1}^{M} p_{b,j} \right) \tilde{h}_{b,i} \tilde{h}_{b,i+1} } \right).
\end{align*}
The difference of the numerator and denominator of the latter fraction is $p_{b,i+1} \left( \tilde{h}_{b,i+1} - \tilde{h}_{b,i} \right)$, which is always positive since $\tilde{h}_{b,i+1}>\tilde{h}_{b,i}$, which results in
$R^{A-B}_\text{gap} (\boldsymbol{p}_b)>0$. Therefore, for any feasible $\boldsymbol{p}_b$, the decoding order $i+1 \to i$ for each two adjacent users $i$ and $i+1$ in cell $b$ is optimal if and only if $\tilde{h}_{b,i+1} >\tilde{h}_{b,i}$. Imposing this optimality condition to each two adjacent users in cell $b$ results in the decoding order based on $\tilde{h}_{b,k}(\boldsymbol{\alpha}_{-b})>\tilde{h}_{b,i}(\boldsymbol{\alpha}_{-b}) \Rightarrow k \to i$. As a result, $\lambda^*_{b,i,k}=1$ if and only if $\tilde{h}_{b,i}(\boldsymbol{\alpha}_{-b})\leq \tilde{h}_{b,k}(\boldsymbol{\alpha}_{-b})$, and the proof is completed.

\section{Proof of Proposition \ref{Propos optpower}.}\label{appendix optpower}
According to Corollary \ref{corol optorder 3}, the achievable spectral efficiency of each user $i \in \mathcal{U}_b$ for the fixed $\boldsymbol{\alpha}_{-b}$ and optimal decoding order $M \to M-1 \to \dots \to 1$
can be formulated by $$\tilde{R}_{b,i}(\boldsymbol{p}_b) = \log_2\left( 1+ \frac{p_{b,i} \tilde{h}_{b,i} (\boldsymbol{\alpha}_{-b})}{\sum\limits_{j=i+1}^{M} p_{b,j} \tilde{h}_{b,i}(\boldsymbol{\alpha}_{-b}) + 1} \right).$$ 
Note that $\tilde{h}_{b,k} (\boldsymbol{\alpha}_{-b}) > \tilde{h}_{b,i} (\boldsymbol{\alpha}_{-b})$ for each $i,k \in \mathcal{U}_b,~k>i$. Moreover, $\tilde{R}_{b,i}$ is independent from $\boldsymbol{p}_{-b}$ for the given $\boldsymbol{\alpha}$, meaning that \eqref{centg problem} can be equivalently divided into $B$ single-cell NOMA sub-problems. In cell $b$, we find $\boldsymbol{p}^*_b$ by solving the following sub-problem as
\begin{subequations}\label{cell problem 2}
	\begin{align}\label{obf cell problem 2}
		\max_{ \boldsymbol{p}_b \geq 0}\hspace{.0 cm}	
		~~ & \sum\limits_{i=1}^{M} \tilde{R}_{b,i} (\boldsymbol{p}_b)
		\\
		\label{Constraint max powercell}
		\text{s.t.}~~& \sum\limits_{i \in \mathcal{U}_b} p_{b,i} = \alpha_b P^{\text{max}}_b,
		\\
		\label{Constraint QoScell 2}
		& \tilde{R}_{b,i}(\boldsymbol{p}_b) \geq R^{\text{min}}_{b,i},~\forall i \in \mathcal{U}_b.
	\end{align}
\end{subequations}
Similar to single-cell NOMA, it can be easily shown that the Hessian of \eqref{obf cell problem 2} is negative definite in $\boldsymbol{p}_b$ for $\tilde{h}_{b,k} (\boldsymbol{\alpha}_{-b}) > \tilde{h}_{b,i} (\boldsymbol{\alpha}_{-b})$ for each $i,k \in \mathcal{U}_b,~k>i$ \cite{8352643}, so the objective function \eqref{obf cell problem 2} is strictly concave in $\boldsymbol{p}_b$. \eqref{Constraint QoScell 2} can be rewritten as the following linear constraint $2^{R^{\text{min}}_{b,i}} \left(1+\sum\limits_{j=i+1}^{M} p_{b,j} \tilde{h}_{b,i}\right) \leq 1+\sum\limits_{j=i}^{M} p_{b,j} \tilde{h}_{b,i}$. Hence, the feasible region of \eqref{cell problem 2} is affine, so is convex. Accordingly, the problem \eqref{cell problem 2} is strictly convex in $\boldsymbol{p}_b$. The Slater's condition holds in \eqref{cell problem 2} since it is convex and there exists $\boldsymbol{p}_b \geq 0$ satisfying \eqref{Constraint QoScell 2} with strict inequalities. Therefore, the strong duality holds in \eqref{cell problem 2}. As a result, the KKT conditions are satisfied and the optimal solution $\boldsymbol{p}^*$ can be obtained by using the Lagrange dual method \cite{Boydconvex}.
The Lagrange function (upper-bound) of \eqref{cell problem 2} is given by
\begin{multline*}
	L(\boldsymbol{p}_b,\boldsymbol{\mu},\boldsymbol{\delta},\nu) = \sum\limits_{i=1}^{M} \log_2\left( 1+\frac{p_{b,i} \tilde{h}_{b,i}}{1+\sum\limits_{j=i+1}^{M} p_{b,j} \tilde{h}_{b,i}} \right) +
	\\
	\sum\limits_{i=1}^{M} \mu_i \left(  \log_2\left( 1+\frac{p_{b,i} \tilde{h}_{b,i}}{1+\sum\limits_{j=i+1}^{M} p_{b,j} \tilde{h}_{b,i}} \right) - R^{\text{min}}_{b,i} \right)
	+ \sum\limits_{i=1}^{M} \delta_i p_{b,i} +\nu \left( \alpha_b P^{\text{max}}_b - \sum\limits_{i=1}^{M} p_{b,i} \right),
\end{multline*}
where $\boldsymbol{\mu}=[\mu_1,\dots,\mu_M]$, $\nu$, and $\boldsymbol{\delta}=[\delta_1,\dots,\delta_M]$ are the Lagrangian multipliers corresponding to the constraints \eqref{Constraint QoScell 2}, \eqref{Constraint max powercell}, and $p_{b,i}\geq 0,~i=1,\dots,M$, respectively. The Lagrange dual problem is given by
\begin{subequations}
	\begin{align*}
		\min_{\boldsymbol{\mu},\boldsymbol{\delta},\nu}\hspace{.0 cm}	
		~~ & \sup\limits_{\boldsymbol{p}} \left\{ L(\boldsymbol{p},\boldsymbol{\mu},\boldsymbol{\delta},\nu) \right\}
		\\
		\text{s.t.}~~& \mu_i \geq 0,~\forall i=1,\dots,M,
		\\
		& \delta_i \geq 0,~\forall i=1,\dots,M.
	\end{align*}
\end{subequations}
The KKT conditions are listed below:
\begin{enumerate}
	\item Feasibility of the primal problem \eqref{cell problem 2}:
	$$
	\hspace{-0.4cm}\textbf{C-1.1:}~\log_2\left( 1+\frac{p^*_{b,i} \tilde{h}_{b,i}}{1+\sum\limits_{j=i+1}^{M} p^*_{b,j} \tilde{h}_{b,i}} \right) \geq R^{\text{min}}_{b,i},~\forall i,~~\textbf{C-1.2:}~p^*_{b,i} \geq 0,~\forall i,~~\textbf{C-1.3:}~\sum\limits_{i=1}^{M} p^*_{b,i} = \alpha_b P^{\text{max}}_b.
	$$
	\item Feasibility of the dual problem:
	$$
	\textbf{C-2.1:}~\mu^*_i \geq 0,~\forall i=1,\dots,M,~~~~~~\textbf{C-2.2:}~\delta^*_i \geq 0,~\forall i=1,\dots,M,~~~~~~\textbf{C-2.3:}~\nu^* \in \mathbb{R}.
	$$
	\item The complementary slackness conditions:
	$$
	\textbf{C-3.1:}~\mu^*_i \left(  \log_2\left( 1+\frac{p^*_{b,i} \tilde{h}_{b,i}}{1+\sum\limits_{j=i+1}^{M} p^*_{b,j} \tilde{h}_{b,i}} \right) - R^{\text{min}}_{b,i} \right)=0,\forall i=1,\dots,M,$$
	$$
	\textbf{C-3.2:}~\delta^*_i p^*_{b,i} = 0,~\forall i=1,\dots,M.
	$$
	\item The condition $\nabla_{\boldsymbol{p}^*_b} L(\boldsymbol{p}^*_b,\boldsymbol{\mu}^*,\boldsymbol{\delta}^*,\nu^*) =0$, which implies that
	\begin{multline*}
		\textbf{C-4:}~\frac{\partial L}{\partial p^*_{b,i}} = \sum\limits_{k=1}^{i} \frac{1+\mu^*_i}{\ln 2}.\frac{\tilde{h}_{b,k}}{1+\sum\limits_{j=k}^{M} p^*_{b,j} \tilde{h}_{b,k}} - \sum\limits_{k=1}^{i-1} \frac{1+\mu^*_i}{\ln 2}.\frac{\tilde{h}_{b,k}}{1+\sum\limits_{j=k+1}^{M} p^*_{b,j} \tilde{h}_{b,k}} + \delta^*_i - \nu^* = 0, \\
		\forall i=1,\dots,M.
	\end{multline*}
	This equation can be reformulated by
	\begin{multline*}
		\frac{\partial L}{\partial p^*_{b,i}} =
		\frac{1+\mu^*_i}{\ln 2}.\frac{h_{b,1}}{1+\sum\limits_{j=1}^{M} p^*_{b,j} \tilde{h}_{b,1}} +
		\sum\limits_{k=1}^{i-1} \frac{1+\mu^*_i}{\ln 2}.\left( \frac{\tilde{h}_{b,k+1}} 
		{ 1+\sum\limits_{j=k+1}^{M} p^*_{b,j} \tilde{h}_{b,k+1}} - \frac{\tilde{h}_{b,k}}{1+\sum\limits_{j=k+1}^{M} p^*_{b,j} \tilde{h}_{b,k}} \right)+
		\\
		\delta^*_i - \nu^* = 0, ~\forall i=1,\dots,M.
	\end{multline*}
	To ease of convenience, we indicate $A (\boldsymbol{p}^*_b)=\frac{1}{\ln 2}.\frac{\tilde{h}_{b,1}}{1+\sum\limits_{j=1}^{M} p^*_{b,j} \tilde{h}_{b,1}}$ and $B_{k} (\boldsymbol{p}^*_b) = \frac{1}{\ln 2}.\frac{\tilde{h}_{b,k+1}} { 1+\sum\limits_{j=k+1}^{M} p^*_{b,j} \tilde{h}_{b,k+1}} - \frac{\tilde{h}_{b,k}}{1+\sum\limits_{j=k+1}^{M} p^*_{b,j} \tilde{h}_{b,k}},~\forall k=1,\dots,M-1$.
	Then, the last KKT condition can be reformulated as
	\begin{equation}\label{KKT power Muser}
		\textbf{C-4.1:}~\frac{\partial L}{\partial p^*_{b,i}} =
		\left(1+\mu^*_i\right) A(\boldsymbol{p}^*_b) + \left(1+\mu^*_i\right) \sum\limits_{k=1}^{i-1} B_{k}(\boldsymbol{p}^*_b)	+ \delta^*_i - \nu^* = 0, ~\forall i=1,\dots,M.
	\end{equation}
\end{enumerate}
The primal dual $\delta^*_i,~\forall i=1,\dots,M,$ acts as a slack variable in \textbf{C-4.1} (due to the KKT condition \textbf{C-2.2}), so it can be eliminated by reformulating the KKT conditions (\textbf{C-4.1},\textbf{C-2.2}) and \textbf{C-3.2}, respectively as
\begin{equation}\label{transf cond2 KKTM}
	\nu^* \geq \left(1+\mu^*_i\right) \left(A(\boldsymbol{p}^*_b) + \sum\limits_{k=1}^{i-1} B_{k}(\boldsymbol{p}^*_b)\right),~\forall i=1,\dots,M,
\end{equation}
and
\begin{equation}\label{transf cond3 KKTM}
	p^*_{b,i} \left( \nu^* - \left(1+\mu^*_i\right) \left(A(\boldsymbol{p}^*_b) + \sum\limits_{k=1}^{i-1} B_{k}(\boldsymbol{p}^*_b)\right) \right) = 0,~\forall i=1,\dots,M.
\end{equation}
Obviously, $A(\boldsymbol{p}^*),~i=1,\dots,M,$ is positive in $\boldsymbol{p}^*_b$. According to Corollary \ref{corol optorder 3}, $B_{k}(\boldsymbol{p}^*_b),~i=1,\dots,M,$ is also positive in $\boldsymbol{p}^*_b$, since $\tilde{h}_{b,k+1} > \tilde{h}_{b,k}$. To simplify the derivations, in the following, we assume that $R^{\text{min}}_{b,i}>0,~\forall i=1,\dots,M$. Then, we show that the derivations are valid for the case that $R^{\text{min}}_{b,i}=0$ for some $i \in \mathcal{U}_b$.
The assumption $R^{\text{min}}_{b,i}>0,~\forall i=1,\dots,M$ implies that $p^*_{b,i}>0,~\forall i=1,\dots,M$ in $\textbf{C-1.2}$. According to \eqref{transf cond3 KKTM}, we have
\begin{equation}\label{transf cond4 KKTM}
	\nu^* = \left(1+\mu^*_i\right) \left(A(\boldsymbol{p}^*_b) + \sum\limits_{k=1}^{i-1} B_{k}(\boldsymbol{p}^*_b)\right),~\forall i=1,\dots,M.
\end{equation}
Consider two adjacent users $i$ and $i+1$. According to \eqref{transf cond4 KKTM}, we have
$$
\left(1+\mu^*_i\right) \left(A(\boldsymbol{p}^*_b) + \sum\limits_{k=1}^{i-1} B_{k}(\boldsymbol{p}^*_b)\right) = \left(1+\mu^*_{i+1}\right) \left(A(\boldsymbol{p}^*_b) + \sum\limits_{k=1}^{i} B_{k}(\boldsymbol{p}^*_b)\right),
$$
Since $A(\boldsymbol{p}^*_b)>0$ and $B_{k}(\boldsymbol{p}^*_b)>0,\forall k$, we have $A(\boldsymbol{p}^*_b) + \sum\limits_{k=1}^{i} B_{k}(\boldsymbol{p}^*_b) > A(\boldsymbol{p}^*_b) + \sum\limits_{k=1}^{i-1} B_{k}(\boldsymbol{p}^*_b)$. Accordingly, the latter equality holds if $\mu^*_{i+1} < \mu^*_i$. Therefore, there exists $\nu^*$ satisfying \eqref{transf cond4 KKTM} if $\mu^*_{i+1} < \mu^*_i$ for each $i \in \mathcal{U}_b$. Accordingly, the following strict inequalities hold
$$
\mu^*_M < \mu^*_{M-1} < \dots < \mu^*_1.
$$ 
Based on Condition \textbf{C-2.1}, we have $\mu^*_M \geq 0$ which implies that $\mu^*_i>0,~\forall i=1,\dots,M-1$. According to Condition \textbf{C-3.1}, the optimal power for users $1,\dots,M-1$ can be obtained by
\begin{equation}\label{opt poweri Scell}
	\log_2\left( 1+\frac{p^*_{b,i} \tilde{h}_{b,i}}{1+\sum\limits_{j=i+1}^{M} p^*_{b,j} \tilde{h}_{b,i}} \right)= R^{\text{min}}_{b,i},~\forall i=1,\dots,M-1.
\end{equation}
According to the power condition \textbf{C-1.3}, the optimal power of user $M$ can be obtained by
\begin{equation}\label{opt powerM Scell}
	p^*_{b,M}=\alpha_b P^{\text{max}}_b-\sum\limits_{i=1}^{M-1} p^*_{b,i}.
\end{equation}
Note that $\mu^*_M > 0$ implies that $\log_2\left( 1 + p^*_{b,M} \tilde{h}_{b,M}  \right)= R^{\text{min}}_{b,i}$ due to Condition \textbf{C-3.1} which may violate Condition \textbf{C-1.3}. Therefore, at the optimal point $p^*_{b,M}$ obtained by \eqref{opt powerM Scell}, we have $\mu^*_M = 0$.
Additionally, $\nu^*$ can take any value since at the optimal point, the KKT condition \textbf{C-1.3} holds.
From \eqref{opt poweri Scell}, it can be concluded that at the optimal point $\boldsymbol{p}^*_b$, the allocated power to all the users with lower decoding order, i.e., users $i=1,\dots,M-1$, is to only maintain their minimum spectral efficiency demand $R^{\text{min}}_{b,i}$. Moreover, \eqref{opt powerM Scell} proves that only the NOMA cluster-head user $M$ deserves additional power.
According to \eqref{opt poweri Scell}, the optimal power for each user $i<M$ can be obtained by
\begin{equation}\label{opt pow t}
	p^*_{b,i}=\frac{T_{b,i} \left( 1 + \left(\alpha_b P^{\text{max}}_b-\sum\limits_{j=1}^{i-1} p^*_{b,j}\right) \tilde{h}_{b,i} \right)}{1+T_{b,i} \tilde{h}_{b,i}},~\forall i=1,\dots,M-1,
\end{equation}
where $T_{b,i}=\frac{2^{R^{\text{min}}_{b,i}} -1}{\tilde{h}_{b,i}},~\forall i=1,\dots,M-1$.
For the case that $R^{\text{min}}_{b,i} \to 0$ for each user $i=1,\dots,M-1$, then $T_{b,i} \to 0$. Therefore, $p^*_{b,i} \to 0$ meaning that when the spectral efficiency demand of the weaker user is zero, no power will be allocated to that user. According to \eqref{opt pow t}, $p^*_{b,i}$ depends on optimal powers $p^*_{b,j},~\forall j=1,\dots,i-1$. Hence, the optimal powers can be directly obtained by calculating $p^*_{b,1} \Rightarrow p^*_{b,2} \Rightarrow \dots \Rightarrow p^*_{b,M-1}$ by \eqref{opt pow t}, and finally $p^*_{b,M}$ according to \eqref{opt powerM Scell}.
To find a closed-form expression for $p^*_{b,i}$, we rewrite \eqref{opt pow t} as
\begin{equation*}
	p^*_{b,i}=\beta_{b,i} \left(P_b^{\text{max}} - \sum\limits_{j=1}^{i-1} p^*_{b,j} + \frac{1}{\tilde{h}_{b,i}} \right),~\forall i=1,\dots,M-1,
\end{equation*}
where $\beta_{b,i}=\frac{2^{R^{\text{min}}_{b,i}} -1}{2^{R^{\text{min}}_{b,i}}},~\forall i=1,\dots,M-1$. Then, we have
\begin{align*}
	p^*_{b,i}&=\beta_{b,i} \left(\alpha_b P_b^{\text{max}} - p^*_{b,i-1} - \sum\limits_{j=1}^{i-2} p^*_{b,j} + \frac{1}{\tilde{h}_{b,i}}  \right)
	\\
	&
	\\
	&=\beta_{b,i} \bigg(\alpha_b P_b^{\text{max}} - \beta_{b,i-1} \left(\alpha_b P_b^{\text{max}} - \sum\limits_{j=1}^{i-2} p^*_{b,j} + \frac{1}{\tilde{h}_{b,i-1}}  \right) - \sum\limits_{j=1}^{i-2} p^*_{b,j} + \frac{1}{\tilde{h}_{b,i}}  \bigg)
	\\
	&=\beta_{b,i} \bigg( \left(1-\beta_{b,i-1}\right) \alpha_b P_b^{\text{max}} - \left(1-\beta_{b,i-1}\right) 
	\sum\limits_{j=1}^{i-1} p^*_{b,j} + \frac{1}{\tilde{h}_{b,i}} - \frac{\beta_{b,i-1}}{\tilde{h}_{b,i-1}} \bigg)
	\\
	&\vdots
	\\
	&=\beta_{b,i} \bigg( \left(1-\beta_{b,i-1}\right) \left(1-\beta_{b,i-2}\right)\dots \left(1-\beta_{b,1}\right) \alpha_b P_b^{\text{max}} + \frac{1}{\tilde{h}_{b,i}} - \frac{\beta_{b,i-1}}{\tilde{h}_{b,i-1}} - \frac{\left(1-\beta_{b,i-1}\right) \beta_{b,i-2}}{\tilde{h}_{b,i-2}} \dots
	\\
	&~~~~~~~~-\frac{\left(1-\beta_{b,i-1}\right) \left(1-\beta_{b,i-2}\right) \dots \left(1-\beta_{b,2}\right) \beta_{b,1}}{\tilde{h}_{b,1}} \bigg).
\end{align*}
According to the above, we have
\begin{equation*}
	p^*_{b,i}=\beta_{b,i} \left( \prod\limits_{j=1}^{i-1} \left(1-\beta_{b,j}\right) \alpha_b P^{\text{max}}_b
	+\frac{1}{\tilde{h}_{b,i}}-
	\sum\limits_{j=1}^{i-1} \frac{\beta_{b,j} \prod\limits_{k=j+1}^{i-1} \left(1-\beta_{b,k}\right)} {\tilde{h}_{b,j}} \right)
	,\forall i=1,\dots,M-1.
\end{equation*}
According to \eqref{opt powerM Scell}, the optimal power of the NOMA cluster-head user $M$ can be obtained by
\begin{equation*}
	p^*_{b,M}=\alpha_b P^{\text{max}}_b - \sum\limits_{i=1}^{M-1} \beta_{b,i} \left( \prod\limits_{j=1}^{i-1} \left(1-\beta_{b,j}\right) \alpha_b P^{\text{max}}_b +\frac{1}{\tilde{h}_{b,i}}-
	\sum\limits_{j=1}^{i-1} \frac{\beta_{b,j} \prod\limits_{k=j+1}^{i-1} \left(1-\beta_{b,k}\right)} {\tilde{h}_{b,j}} \right).
\end{equation*}

\section{Closed-Form Expression of Optimal Powers for Total Power Minimization Problem}\label{appendix optpower powmin}
Here, we first obtain the closed-form expression of optimal powers for a $M$-user singe-cell NOMA system under the CNR-based decoding order. Then, we extend the results to the case that ICI is fixed in cell $b$ serving $M$ users and find the closed-form expressions of powers in $\boldsymbol{p}^*_b$ under the CINR-based decoding order. 

In the power minimization problem of a $M$-user single-cell NOMA system with $\tilde{h}_1 < \tilde{h}_2 < \dots < \tilde{h}_M$ and thus the optimal (CNR-based) decoding order $M \to M-1 \to \dots \to 1$, the achievable spectral efficiency of user $i$ can be obtained by $R_{i}(\boldsymbol{p})= \log_2\left( 1+ \frac{p_{i} \tilde{h}_{i}} {\sum\limits_{j=i+1}^{M} p_{j} \tilde{h}_{i} + 1} \right)$. Here, $\tilde{h}_{i}=\frac{h_{i}}{\sigma^2_{i}}$ is the normalized channel gain of user $i$ by its noise power $\sigma^2_{i}$.
The total power minimization problem under the CNR-based decoding order can be formulated by
\begin{subequations}\label{problem Scell Powmin}
	\begin{align}\label{obf Scell Powmin}
	\min_{ \boldsymbol{p} \geq 0 }\hspace{.0 cm}	
	~~ & \sum\limits_{i=1}^{M} p_i
	\\
	\label{Constraint max power scell}
	\text{s.t.}~~& \sum\limits_{i=1}^{M} p_i \leq P^{\text{max}},
	\\
	\label{Constraint Qos scell}
	& \log_2\left( 1+ \frac{p_{i} \tilde{h}_{i}} {\sum\limits_{j=i+1}^{M} p_{j} \tilde{h}_{i} + 1} \right) \geq R^{\text{min}}_i,~\forall i=1,\dots,M.
	\end{align}
\end{subequations}
The minimum rate constraint \eqref{Constraint Qos scell} can be rewritten as $p_i h_{i} \geq \left(2^{R^{\text{min}}_i}-1\right) \left(\sum\limits_{j=i+1}^{M} p_j \tilde{h}_i + 1\right),~\forall i=1,\dots,M$, which is affine in $\boldsymbol{p}$. Hence, problem \eqref{problem Scell Powmin} is convex in $\boldsymbol{p}$ with an affine feasible set. 
It can be shown that in the power minimization problem, the allocated power to each user is only to maintain its minimal rate demand. In the following, we prove this proposition by analyzing the KKT conditions.
The Slater's condition holds in \eqref{problem Scell Powmin} since it is convex and there exists $\boldsymbol{p} \geq 0$ satisfying \eqref{Constraint max power scell} and \eqref{Constraint Qos scell} with strict inequalities. Therefore, the strong duality in \eqref{problem Scell Powmin} holds. Hence, the KKT conditions are satisfied and the optimal solution $\boldsymbol{p}^*$ can be obtained by using the Lagrange dual method \cite{Boydconvex}.
The Lagrange function (lower-bound) of \eqref{problem Scell Powmin} is given by
\begin{multline*}
L(\boldsymbol{p},\boldsymbol{\mu},\boldsymbol{\delta},\nu) = \sum\limits_{i=1}^{M} p_i 
+ \sum\limits_{i=1}^{M} \mu_i \left( R^{\text{min}}_i - \log_2\left( 1+ \frac{p_i \tilde{h}_i}{1+\sum\limits_{j=i+1}^{M} p_j \tilde{h}_i} \right) \right)
+ \sum\limits_{i=1}^{M} \delta_i (-p_i) 
\\
+\nu \left( \sum\limits_{i=1}^{M} p_i - P^{\text{max}} \right),
\end{multline*}
where $\boldsymbol{\mu}=[\mu_1,\dots,\mu_M]$,  $\boldsymbol{\delta}=[\delta_1,\dots,\delta_M]$, and $\nu$ are the Lagrangian multipliers corresponding to the constraints \eqref{Constraint Qos scell}, \eqref{Constraint max power scell}, and $p_i\geq 0,~i=1,\dots,M$, respectively. The Lagrange dual problem is given by
\begin{subequations}
	\begin{align*}
	\max_{\boldsymbol{\mu},\boldsymbol{\delta},\nu}\hspace{.0 cm}	
	~~ & \inf\limits_{\boldsymbol{p}} \left\{ L(\boldsymbol{p},\boldsymbol{\mu},\boldsymbol{\delta},\nu) \right\}
	\\
	\text{s.t.}~~& \mu_i \geq 0,~\forall i=1,\dots,M,
	\\
	& \delta_i \geq 0,~\forall i=1,\dots,M.
	\end{align*}
\end{subequations}
The KKT conditions are listed below.
\begin{enumerate}
	\item Feasibility of the primal problem \eqref{problem Scell Powmin}:
	$$
	\textbf{C-1.1:}~\log_2\left( 1+ \frac{p^*_i \tilde{h}_i}{1+\sum\limits_{j=i+1}^{M} p^*_j \tilde{h}_i} \right) \geq R^{\text{min}}_i,~\forall i,~~~~~\textbf{C-1.2:}~p^*_i \geq 0,~\forall i,~~~~~\textbf{C-1.3:}~\sum\limits_{i=1}^{M} p^*_i \leq P^{\text{max}}.
	$$
	\item Feasibility of the dual problem:
	$$
	\textbf{C-2.1:}~\mu^*_i \geq 0,~\forall i=1,\dots,M,~~~~~~\textbf{C-2.2:}~\delta^*_i \geq 0,~\forall i=1,\dots,M,~~~~~~\textbf{C-2.3:}~\nu^* \geq 0.
	$$
	\item The complementary slackness conditions:
	$$
	\textbf{C-3.1:}~\mu^*_i \left( R^{\text{min}}_i - \log_2\left( 1+ \frac{p^*_i \tilde{h}_i}{1+\sum\limits_{j=i+1}^{M} p^*_j \tilde{h}_i} \right)  \right)=0,\forall i=1,\dots,M,
	$$
	$$
	\textbf{C-3.2:}~\delta^*_i p^*_i = 0,~\forall i=1,\dots,M,~~~~~~~~~
	\textbf{C-3.3:}~\nu^*\left( \sum\limits_{i=1}^{M} p^*_i - P^{\text{max}} \right) = 0.
	$$
	\item The condition $\nabla_{\boldsymbol{p}^*} L(\boldsymbol{p}^*,\boldsymbol{\mu}^*,\boldsymbol{\delta}^*,\nu^*) =0$, which implies that
	$$
	\textbf{C-4:}~\frac{\partial L}{\partial p^*_i} = 1 - 
	\sum\limits_{j=1}^{i-1} \mu^*_j \left(2^{R^{\text{min}}_j}-1\right) \tilde{h}_j - \mu^*_i \tilde{h}_i -	\delta^*_i + \nu^* = 0, ~\forall i=1,\dots,M.
	$$
	Let $B_j=\left(2^{R^{\text{min}}_j}-1\right) \tilde{h}_j,~j=1,\dots,M$. The latter equation is rewritten as
	\begin{equation*}
	\textbf{C-4.1:}~\frac{\partial L}{\partial p^*_i} = 1 - 
	\sum\limits_{j=1}^{i-1} \mu^*_j B_j - \mu^*_i \tilde{h}_i -	\delta^*_i + \nu^* = 0, ~\forall i=1,\dots,M.
	\end{equation*}
\end{enumerate}
The primal dual $\delta^*_i,~\forall i=1,\dots,M,$ acts as a slack variable in \textbf{C-4.1} (due to the KKT condition \textbf{C-2.2}), so it can be eliminated by reformulating the KKT conditions (\textbf{C-4.1},\textbf{C-2.2}) and \textbf{C-3.2}, respectively as
\begin{equation}\label{transf minpow1 KKT}
\nu^* \geq  \sum\limits_{j=1}^{i-1} \mu^*_j B_j + \mu^*_i \tilde{h}_i - 1 ,~\forall i=1,\dots,M,
\end{equation}
and
\begin{equation}\label{transf minpow2 KKT}
p^*_i \left( \nu^* - \left(\sum\limits_{j=1}^{i-1} \mu^*_j B_j + \mu^*_i \tilde{h}_i - 1\right) \right) = 0,~\forall i=1,\dots,M.
\end{equation}
We first assume that $R^{\text{min}}_i>0,~\forall i=1,\dots,M$. Then, we show that the derivations are valid for $R^{\text{min}}_i=0$ for some $i$. The assumption $R^{\text{min}}_i>0,~\forall i=1,\dots,M$ implies that $p^*_i>0,~\forall i=1,\dots,M$ in $\textbf{C-1.2}$. According to \eqref{transf minpow2 KKT}, we have
\begin{equation}\label{result minpow2 KKT}
\nu^* = \sum\limits_{j=1}^{i-1} \mu^*_j B_j + \mu^*_i \tilde{h}_i - 1,~\forall i=1,\dots,M.
\end{equation}
Consider two adjacent users $i$ and $i+1$. According to \eqref{result minpow2 KKT}, we have
$$
\sum\limits_{j=1}^{i-1} \mu^*_j B_j + \mu^*_i \tilde{h}_i = \left(\sum\limits_{j=1}^{i-1} \mu^*_j B_j + \mu^*_i B_i\right) + \mu^*_{i+1} h_{i+1},
$$
which can be simplified to
$$
\mu^*_i \tilde{h}_i = \mu^*_i B_i + \mu^*_{i+1} h_{i+1}~~~~\Rightarrow~~~~ \mu^*_i \left(\tilde{h}_i - B_i\right) = \mu^*_{i+1} h_{i+1}.
$$
In the following, we prove that $\mu^*_{i+1} < \mu^*_i$. Let $\mu^*_{i+1} \geq \mu^*_i$. It implies that $h_{i+1} \leq \tilde{h}_i - B_i$, which is equivalent to $h_{i+1} + B_i \leq \tilde{h}_i$. Since $B_i>0$, it results in $h_{i+1} \leq \tilde{h}_i$ which violates our assumption $h_{i+1} > \tilde{h}_i$. Accordingly, \eqref{result minpow2 KKT} holds if $\mu^*_{i+1} < \mu^*_i$ for each two adjacent users $i$ and $i+1$. 
Hence, there exists $\nu^*$ satisfying \eqref{result minpow2 KKT} if 
$$\mu^*_M < \mu^*_{M-1} < \dots < \mu^*_1.$$ 
According to Condition \textbf{C-2.1}, we have $\mu^*_M \geq 0$ which implies that $\mu^*_i>0,~\forall i=1,\dots,M-1$. The optimal Lagrangian multiplier $\mu^*_M$ of the NOMA cluster-head user is also positive. This is due to the fact that $r_M(p^*_M)=\log_2\left( 1+ p^*_M h_M \right)$ in \eqref{Constraint Qos scell} is monotonically increasing in $p^*_M$, and also independent from the other optimal powers. Hence, at the optimal point which corresponds to the minimal $p^*_M$, the spectral efficiency $r_M(p^*_M)$ reaches to its lower-bound $R^{\text{min}}_M$. Hence, we have $\log_2\left( 1+ p^*_M h_M \right)=R^{\text{min}}_M$. According to the KKT condition \textbf{C-3.1}, $\mu^*_M>0$. As a result, we have
$$0< \mu^*_M < \mu^*_{M-1} < \dots < \mu^*_1.$$
According to Condition \textbf{C-3.1}, the optimal power for each user $i=1,\dots,M$ can be obtained by
\begin{equation}\label{opt poweri powmin Scell}
\log_2\left( 1+ \frac{p^*_i \tilde{h}_i}{1+\sum\limits_{j=i+1}^{M} p^*_j \tilde{h}_i}  \right)= R^{\text{min}}_i,~\forall i=1,\dots,M.
\end{equation}
It is noteworthy that the duality gap between the primal and dual problems is zero when the Slater's condition holds. This condition implies that there exists $\boldsymbol{p}$ such that the KKT condition \textbf{C-1.3} with strict inequality holds, meaning that the feasible region of \eqref{problem Scell Powmin} with strict inequality power constraint $\sum\limits_{i=1}^{M} p_i < P^{\text{max}}$ is nonempty.
Since $\sum\limits_{i=1}^{M} p_i = P^{\text{max}}$ corresponds to the maximum value of the objective function \eqref{obf Scell Powmin},
satisfying the Slater's condition ensures us $\sum\limits_{i=1}^{M} p^*_i < P^{\text{max}}$. According to Condition \textbf{C-3.3}, we have\footnote{The discussions about optimal $\nu^*=0$ is only additional notes on the impact of the power constraint. We proved that  for any non-empty feasible set satisfying the Slater's condition, the power constraint will not be active.} $\nu^*=0$.
According to the above, it can be concluded that at the optimal point $\boldsymbol{p}^*$, the allocated power to each user $i$ is only to maintain its minimum spectral efficiency demand $R^{\text{min}}_i$.
According to \eqref{opt poweri Scell}, the optimal power (in Watts) for each user $i<M$ can be obtained by
\begin{equation}\label{opt totminpow scell}
p^*_i=T_i \left( 1 + \sum\limits_{j=i+1}^{M} p^*_j \tilde{h}_i \right),~\forall i=1,\dots,M,
\end{equation}
where $T_i=\frac{2^{R^{\text{min}}_i} -1}{\tilde{h}_i},~\forall i=1,\dots,M$.
For the case that $R^{\text{min}}_i \to 0$, then $T_i \to 0$. Therefore, $p^*_i \to 0$ meaning that no power will be allocated to user $i$. Similar to \eqref{opt pow t}, it can be easily shown that the optimal powers can be obtained directly by \eqref{opt totminpow scell}.
To obtain a closed-form expression for $p^*_i$, we rewrite \eqref{opt totminpow scell} as
\begin{equation*}
p^*_i=\beta_i \left(\frac{1}{\tilde{h}_i} + \sum\limits_{j=i+1}^{M} p^*_j \right),~\forall i=1,\dots,M,
\end{equation*}
where $\beta_i=2^{R^{\text{min}}_i} -1,~\forall i=1,\dots,M$. The optimal power $p^*_{i}$ can be reformulated as
\begin{align*}
p^*_{i}&=\beta_i \left(\frac{1}{\tilde{h}_i} + \sum\limits_{j=i+1}^{M} p^*_j \right)
\\
&=\beta_i \left(\frac{1}{\tilde{h}_i} + p^*_{i+1} + \sum\limits_{j=i+2}^{M} p^*_j \right)
\\
&=\beta_i \left(\frac{1}{\tilde{h}_i} + \beta_{i+1} \left(\frac{1}{h_{i+1}} + \sum\limits_{j=i+2}^{M} p^*_j \right) + \sum\limits_{j=i+2}^{M} p^*_j \right)
\\
&=\beta_i \left((1+\beta_{i+1}) \sum\limits_{j=i+2}^{M} p^*_j + \frac{1}{\tilde{h}_i} + \frac{\beta_{i+1}}{h_{i+1}} \right)
\\
&=\beta_i \left((1+\beta_{i+1}) \left(p^*_{i+2} + \sum\limits_{j=i+3}^{M} p^*_j\right) + \frac{1}{\tilde{h}_i} + \frac{\beta_{i+1}}{h_{i+1}} \right)
\\
&=\beta_i \left((1+\beta_{i+1}) \left( \beta_{i+2} \left(\frac{1}{h_{i+2}} + \sum\limits_{j=i+3}^{M} p^*_j \right) + \sum\limits_{j=i+3}^{M} p^*_j\right) + \frac{1}{\tilde{h}_i} + \frac{\beta_{i+1}}{h_{i+1}} \right)
\\
&=\beta_i \left((1+\beta_{i+1}) (1+\beta_{i+2}) \sum\limits_{j=i+3}^{M} p^*_j + \frac{1}{\tilde{h}_i} + \frac{\beta_{i+1}}{h_{i+1}} + \frac{\beta_{i+2}(1+\beta_{i+1})}{h_{i+2}} \right)
\\
&\vdots
\\
&=\beta_i \bigg( (1+\beta_{i+1}) (1+\beta_{i+2})\dots (1+\beta_{M}) + \frac{1}{\tilde{h}_i} + \frac{\beta_{i+1}}{h_{i+1}} + \frac{\beta_{i+2}(1+\beta_{i+1})}{h_{i+2}} + \dots 
\\
&~~~+ \frac{\beta_{M}(1+\beta_{M-1})\dots (1+\beta_{i+1})}{h_{M}} \bigg).
\end{align*}
According to the above, we have
\begin{equation}\label{optimal minpow Scell}
p^*_i=\beta_{i} \left(\prod\limits_{j=i+1}^{M} \left(1+\beta_{j}\right) +\frac{1}{\tilde{h}_i}+\sum\limits_{j=i+1}^{M} \frac{ \beta_{j} \prod\limits_{k=i+1}^{j-1} \left(1+\beta_{k}\right)}{\tilde{h}_j}\right),~\forall i=1,\dots,M.
\end{equation}

In multi-cell NOMA, let cell $b$ has $M$ users with $\tilde{h}_{b,1} < \tilde{h}_{b,2} < \dots < \tilde{h}_{b,M}$, where $\tilde{h}_{b,i}=\frac{h_{b,i}}{I_{b,i} + \sigma^2_{b,i}}$. According to Corollary \ref{corol optorder 3}, the achievable spectral efficiency of each user $i \in \mathcal{U}_b$ under the optimal decoding order $M \to M-1 \to \dots \to 1$ is $\tilde{R}_{b,i}(\boldsymbol{p}) = \log_2\left( 1+ \frac{p_{b,i} \tilde{h}_{b,i}}{\sum\limits_{j=i+1}^{M} p_{b,j} \tilde{h}_{b,i} + 1} \right)$. In multi-cell NOMA, for the case that ICI is fixed, the power minimization problem of cell $b$ under the optimal (CINR-based) decoding order corresponds to the power minimization problem of single-cell NOMA. According to \eqref{optimal minpow Scell}, for the case that $\tilde{h}_{b,1} < \tilde{h}_{b,2} < \dots < \tilde{h}_{b,M}$, the optimal power of each user $i \in \mathcal{U}_b$ can be obtained in closed form as
\begin{equation*}
p^*_{b,i}=\beta_{b,i} \left(\prod\limits_{j=i+1}^{M} \left(1+\beta_{b,j}\right) +\frac{1}{\tilde{h}_{b,i}}+
\sum\limits_{j=i+1}^{M} \frac{ \beta_{b,j} \prod\limits_{k=i+1}^{j-1} \left(1+\beta_{b,k}\right)}{\tilde{h}_{b,j}}\right),~\forall i=1,\dots,M.
\end{equation*}

\section{Joint Power Allocation and Rate Adoption Algorithm}\label{appendix jointpowerrate}
By taking $\ln$ from the both sides of \eqref{Constraint decoding}, we have
$\ln\left(2^{r_{b,i}}-1\right) + \ln\bigg( \sum\limits_{j \in \Phi_{b,i}} p_{b,j} h_{b,k} + I_{b,k} + \sigma^2_{b,k} \bigg) \leq \ln\left( p_{b,i} h_{b,k} \right),~\forall b \in \mathcal{B},~i,k \in \mathcal{U}_b,~k \in \{i\} \cup \Phi_{b,i}$.
Now, let $p_{b,i}=\text{e}^{\tilde{p}_{b,i}}$, and subsequently $I_{b,i} (\boldsymbol{\tilde{p}}_{-b})=\sum\limits_{\hfill j \in \mathcal{B} \hfill\atop  j \neq b} \left(\sum\limits_{l\in \mathcal{U}_j} \text{e}^{\tilde{p}_{j,l}}\right) h_{j,b,i}$. Accordingly, problem \eqref{centg problem} can be rewritten as
\begin{subequations}\label{cent subopt problem 1}
	\begin{align}\label{obf cent subopt problem 1}
		\max_{ \boldsymbol{\tilde{p}},~\boldsymbol{r} \geq 0 }\hspace{.0 cm}	
		~~ & \sum\limits_{b \in \mathcal{B}} \sum\limits_{i \in \mathcal{U}_b} r_{b,i}
		\\
		\text{s.t.}~~~& \eqref{Constraint QoS RA}, \nonumber
		\\
		\label{max power equiv}
		& \sum\limits_{i \in \mathcal{U}_b} \text{e}^{\tilde{p}_{b,i}} \leq P^{\text{max}}_b,~\forall b \in \mathcal{B},
		\\
		\label{Constraint decoding 1}
		& \ln\left(2^{r_{b,i}}-1\right) + \ln\left( \sum\limits_{j \in \Phi_{b,i}} \text{e}^{\tilde{p}_{b,j}} h_{b,k} + \sum\limits_{\hfill j \in \mathcal{B} \hfill\atop  j \neq b} \left(\sum\limits_{l\in \mathcal{U}_j} \text{e}^{\tilde{p}_{j,l}}\right) h_{j,b,i} + \sigma^2_{b,k} \right) \leq \tilde{p}_{b,i} \ln\left(h_{b,k}\right), \nonumber
		\\&~~~~~~~~~~~~~~~~~~~~~~~~~~~~~~~~~~~~~~~~~~~~~~~~~~~
		\forall b \in \mathcal{B},~i,k \in \mathcal{U}_b,~k \in \{i\} \cup \Phi_{b,i}.
	\end{align}
\end{subequations}
The objective function \eqref{cent subopt problem 1} is affine, so is concave on $\boldsymbol{r}$. 
Constraint \eqref{Constraint QoS RA} is affine, so is convex. Constraint \eqref{max power equiv} is also convex since log-sum-exp is convex \cite{Boydconvex}. However, \eqref{Constraint decoding 1} is nonconvex. In the left hand side of \eqref{Constraint decoding 1}, it can be easily shown that the first term $\ln\left(2^{r_{b,i}}-1\right)$ is strictly concave on $r_{b,i}$ which makes \eqref{cent subopt problem 1} nonconvex and strongly NP-hard. 
Now, we apply the iterative sequential programming method. At each iteration $t$, we approximate the term $g\left(r^{(t)}_{b,i}\right)=\ln\left(2^{r^{(t)}_{b,i}}-1\right)$ to its first-order Taylor series around $r^{(t-1)}_{b,i}$ obtained from prior iteration $(t-1)$ as follows:
\begin{equation}\label{approx rate ln}
	\hat{g}\left(r^{(t)}_{b,i}\right) = g\left(r^{(t-1)}_{b,i}\right) + g'\left(r^{(t-1)}_{b,i}\right) \left(r^{(t)}_{b,i} - r^{(t-1)}_{b,i}\right),
\end{equation}
where $g'\left(r\right) = \frac{2^{r}}{2^{r} - 1}$.
By substituting $g\left(r^{(t)}_{b,i}\right)$ with its affine approximated form $\hat{g}\left(r^{(t)}_{b,i}\right)$ in \eqref{approx rate ln}, problem \eqref{cent subopt problem 1} at iteration $t$ will be approximated to the following convex form as
\begin{subequations}\label{cent subopt problem 2}
	\begin{align}\label{obf cent subopt problem 2}
		\max_{ \boldsymbol{\tilde{p}}^{(t)},~\boldsymbol{r}^{(t)} \geq 0 }\hspace{.0 cm}	
		~~ & \sum\limits_{b \in \mathcal{B}} \sum\limits_{i \in \mathcal{U}_b} r^{(t)}_{b,i}
		\\
		\text{s.t.}~& \eqref{Constraint QoS RA},~\eqref{max power equiv} \nonumber
		\\
		\label{Constraint decoding 2}
		& \hat{g}\left(r^{(t)}_{b,i}\right) + \ln\left( \sum\limits_{j \in \Phi_{b,i}} \text{e}^{\tilde{p}^{(t)}_{b,j}} h_{b,k} + \sum\limits_{\hfill j \in \mathcal{B} \hfill\atop  j \neq b} \left(\sum\limits_{l\in \mathcal{U}_j} \text{e}^{\tilde{p}^{(t)}_{j,l}}\right) h_{j,b,i} + \sigma^2_{b,k} \right) \leq \tilde{p}^{(t)}_{b,i} \ln\left(h_{b,k}\right),~\forall b \in \mathcal{B}, \nonumber
		\\&~~~~~~~~~~~~~~~~~~~~~~~~~~~~~~~~~~~~~~~~~~~~~~~~~~~~~~~~~~
		i,k \in \mathcal{U}_b,~k \in \{i\} \cup \Phi_{b,i}.
	\end{align}
\end{subequations}
It can be shown that \eqref{cent subopt problem 2} satisfies the KKT conditions \cite{7862919}, so it can be solved by using the Lagrange dual method, or IPMs \cite{Boydconvex}. In the sequential programming, we first initialize $\boldsymbol{r}^{(0)}$. At each iteration $t$, we solve \eqref{cent subopt problem 2} and find $\left({\boldsymbol{r^*}^{(t)}}, {\boldsymbol{\tilde{p}^*}^{(t)}}\right)$ according to the updated $\hat{g}\left(r^{(t)}_{b,i}\right)$ based on ${\boldsymbol{r^*}^{(t-1)}}$. We continue the iterations until the convergence is achieved.

The solution of \eqref{cent subopt problem 2} remains in the feasible region of the main problem \eqref{cent subopt problem 1}. This is due to the fact that at each iteration $t$, we have $\hat{g}\left(r^{(t)}_{b,i}\right) \geq g\left(r^{(t)}_{b,i}\right)$. Let $\left({\boldsymbol{\hat{r}}^{(t)}},{\boldsymbol{\hat{\tilde{p}}}^{(t)}}\right)$ be the feasible solution of \eqref{cent subopt problem 2}. It implies that \eqref{Constraint decoding 2} is satisfied. Thus, we have
\begin{multline*}
	\hat{\tilde{p}}^{(t)}_{b,i} \ln\left(h_{b,k}\right) - 
	\ln\left( \sum\limits_{j \in \Phi_{b,i}} \text{e}^{\hat{\tilde{p}}^{(t)}_{b,j}} h_{b,k} + \sum\limits_{\hfill j \in \mathcal{B} \hfill\atop  j \neq b} \left(\sum\limits_{l\in \mathcal{U}_j} \text{e}^{\hat{\tilde{p}}^{(t)}_{j,l}}\right) h_{j,b,i} + \sigma^2_{b,k} \right)
	\geq \hat{g}\left(\hat{r}^{(t)}_{b,i}\right),~\forall b \in \mathcal{B},
	\\
	i,k \in \mathcal{U}_b,~k \in \{i\} \cup \Phi_{b,i}.
\end{multline*}
Since $\hat{g}\left(r^{(t)}_{b,i}\right) \geq g\left(r^{(t)}_{b,i}\right)$, we can guarantee that \eqref{Constraint decoding 1} is satisfied, 
meaning that \eqref{cent subopt problem 2} remains in the feasible region of \eqref{cent subopt problem 1}. It can be easily shown that the sequential programming generates a sequence of improved feasible solutions, such that it converges to a stationary point which is a local maxima of \eqref{cent subopt problem 1} \cite{7862919}.
The performance and convergence of this algorithm is numerically evaluated in Subsection \ref{subsection JRPA converg}.

\section{Proof of Theorem \ref{Theorem SIC M-cell positiveterm}.}\label{appendix theorem positiveterm}
For each user pair $i,k \in \mathcal{U}_b$, if \eqref{SIC nec cond} holds at any power level $\boldsymbol{p}_{-b}$, the decoding order $k \to i$ is optimal. Note that \eqref{SIC nec cond} for the user pairs in cell $b$ is completely independent from $\boldsymbol{p}_{b}$ (see Corollary \ref{corollary optorder}).
Let us rewrite \eqref{SIC nec cond} for the user pair $i,k \in \mathcal{U}_b$ with $\frac{h_{b,k}}{\sigma_{b,k}} > \frac{h_{b,i}}{\sigma_{b,i}}$ as
\begin{equation}\label{SIC Mcell ref}
	\underbrace{\frac{h_{b,k}}{\sigma_{b,k}} - \frac{h_{b,i}}{\sigma_{b,i}}}_\text{Non-negative} \geq \sum\limits_{\hfill j \in \mathcal{B} \hfill\atop  j \neq b} \frac{ \left(\sum\limits_{i \in \mathcal{U}_j} p_{j,i}\right) }{\sigma_{b,i} \sigma_{b,k}}  \underbrace{\left( h_{j,b,k} h_{b,i} - h_{j,b,i} h_{b,k} \right)}_{H_{j,b,i,k}}.
\end{equation}
Assume that the left-hand side (LHS) is non-negative (which corresponds to applying the CNR-based decoding order). If each term $H_{j,b,i,k}$ in the right-hand side (RHS) is non-positive, \eqref{SIC Mcell ref} holds for any $\boldsymbol{p}$ \cite{8353846}. The condition $H_{j,b,i,k} \leq 0$ is equivalent to $\frac{h_{b,k}}{h_{b,i}} \geq \frac{h_{j,b,k}}{h_{j,b,i}}$.
The latter sufficient condition implies that for $\frac{h_{b,k}}{\sigma_{b,k}} \geq \frac{h_{b,i}}{\sigma_{b,i}}$, if $H_{j,b,i,k} \leq 0,~\forall j \in \mathcal{B}\setminus\{b\}$, the decoding order $k \to i$ is optimal. This condition is first proposed in Lemma 1 in \cite{8353846} to guarantee that \eqref{SIC Mcell ref} holds independent from $\boldsymbol{p}$. In this paper, we state that the latter independency guarantees that the decoding order $k \to i$ is optimal independent from $\boldsymbol{p}$.

Now, we consider a more general case, where $H_{j,b,i,k}>0$ for at least one neighboring BS $j$. In this case, Lemma 1 in \cite{8353846} cannot guarantee any optimality for the decoding order $k \to i$. However, we show that there exists a sufficient condition to guarantee the optimality of the decoding order $k \to i$. In this line, we find an upper-bound constant for the term in the RHS of \eqref{SIC Mcell ref}. If the LHS of \eqref{SIC Mcell ref} is greater than that of the upper-bound term for the RHS of \eqref{SIC Mcell ref}, it is guaranteed that \eqref{SIC Mcell ref} holds independent from $\boldsymbol{p}$, so the decoding order $k \to i$ is optimal. We first define $\mathcal{Q}_{b,i,k}=\left\{j \in \mathcal{B}\setminus\{b\} \middle\vert H_{j,b,i,k}>0\right\}$ for the user pair $i,k \in \mathcal{U}_b$ with $\frac{h_{b,k}}{\sigma_{b,k}} > \frac{h_{b,i}}{\sigma_{b,i}}$. Obviously, $H_{j,b,i,k} \leq 0$ for each neighboring BS $j' \in \mathcal{B}\setminus\{\mathcal{Q}_{b,i,k}\cup \{b\}\}$. In this way, the following inequality is always guaranteed
\begin{equation}\label{SIC positiveterm ineq}
	\sum\limits_{j \in \mathcal{Q}_{b,i,k}} \frac{P^\text{max}_j} {\sigma_{b,i} \sigma_{b,k}} H_{j,b,i,k} \geq
	\sum\limits_{\hfill j \in \mathcal{B} \hfill\atop  j \neq b} \frac{\left(\sum\limits_{i \in \mathcal{U}_j} p_{j,i}\right)} {\sigma_{b,i} \sigma_{b,k}} H_{j,b,i,k}.
\end{equation}
In other words, the upper-bound of the RHS of \eqref{SIC Mcell ref} can be achieved if we assume that each BS $j \in \mathcal{Q}_{b,i,k}$ (or equivalently each BS $j$ with $H_{j,b,i,k} > 0$) operates at its maximum power budget $P^\text{max}_j$, and each BS $j \notin \mathcal{Q}_{b,i,k}$ (or equivalently each BS $j$ with $H_{j,b,i,k} \leq 0$) is deactivated. According to \eqref{SIC Mcell ref} and \eqref{SIC positiveterm ineq}, the following inequality always holds
\begin{equation}\label{SIC positiveterm ineq 2}
	\frac{h_{b,k}}{\sigma_{b,k}} - \frac{h_{b,i}}{\sigma_{b,i}} \geq \sum\limits_{j \in \mathcal{Q}_{b,i,k}} \frac{P^\text{max}_j} {\sigma_{b,i} \sigma_{b,k}} H_{j,b,i,k} \geq
	\sum\limits_{\hfill j \in \mathcal{B} \hfill\atop  j \neq b} \frac{\left(\sum\limits_{i \in \mathcal{U}_j} p_{j,i}\right)} {\sigma_{b,i} \sigma_{b,k}} H_{j,b,i,k}.
\end{equation}
Thus, if $\frac{h_{b,k}}{\sigma_{b,k}} - \frac{h_{b,i}}{\sigma_{b,i}} \geq \sum\limits_{j \in \mathcal{Q}_{b,i,k}} \frac{P^\text{max}_j} {\sigma_{b,i} \sigma_{b,k}} H_{j,b,i,k}$, we guarantee that \eqref{SIC Mcell ref} holds independent from $\boldsymbol{p}$, so the decoding order $k \to i$ is optimal, and the proof is completed.

\section{Proof of Fact \ref{Propos negative SIC}.}\label{appendix negative SIC}
For convenience, let all the channel gains be normalized by noise, and the CNR-based decoding order is applied. According to \eqref{SIC Mcell ref}, for each user pair $i,k \in \mathcal{U}_b,~k \in \Phi_{b,i}$ or equivalently $h_{b,k} \geq h_{b,i}$, the constraint \eqref{SIC nec ind pb} can be rewritten as:
\begin{equation}\label{SIC normalized}
	h_{b,k} - h_{b,i} \geq \sum\limits_{\hfill j \in \mathcal{B} \hfill\atop  j \neq b} \left(\sum\limits_{i \in \mathcal{U}_j} p_{j,i}\right) H_{j,b,i,k},
\end{equation}
where $H_{j,b,i,k}=\left( h_{j,b,k} h_{b,i} - h_{j,b,i} h_{b,k} \right)$.
Constraint \eqref{SIC normalized} can be equivalently transformed to $(B-1)$ maximum power constraints including all the BSs in $\mathcal{B} \setminus \{b\}$. 
The power consumption of each BS $j \in \mathcal{B} \setminus \{b\}$ is upper-bounded by
\begin{equation*}\label{SIC max powerj}
	\sum\limits_{i \in \mathcal{U}_b} p_{j,i} \leq 
	\underbrace{\frac{(h_{b,k} - h_{b,i}) - \sum\limits_{\hfill l \in \mathcal{B}\setminus\{j,b\}} \left(\sum\limits_{i \in \mathcal{U}_l} p_{l,i}\right) H_{l,b,i,k}}{H_{j,b,i,k}}}_{\Psi_{j,b,i,k} (\boldsymbol{p}_{-(j,b)})}.
\end{equation*}
In general, the maximum power constraint \eqref{Constraint max power} and SIC constraint \eqref{SIC nec ind pb}  can be equivalently combined as
\begin{equation}\label{SIC combine maxpow}
	\sum\limits_{i \in \mathcal{U}_b} p_{b,i} \leq \min \left\{P^\text{max}_b, \min\limits_{\hfill j \in \mathcal{B}\setminus\{b\},\hfill \atop i,k \in \mathcal{U}_j, k \in \Phi_{j,i}} \Psi_{b,j,i,k} (\boldsymbol{p}_{-(b,j)}) \right\},~\forall b \in \mathcal{B}.
\end{equation}
The negative side impact of the SIC necessary constraint \eqref{SIC nec ind pb} can observed in \eqref{SIC combine maxpow}. For the case that $\min\limits_{\hfill j \in \mathcal{B}\setminus\{b\},\hfill \atop i,k \in \mathcal{U}_j, k \in \Phi_{j,i}} \Psi_{b,j,i,k} (\boldsymbol{p}_{-(b,j)}) < P^\text{max}_b$, \eqref{SIC nec ind pb} imposes additional limitations on the total power consumption of BS $b$ which restricts the feasible region of \eqref{Constraint max power}.

\section{Proof of Fact \ref{Propos singlecell maxpow}.}\label{appendix Propos singlecell maxpow}
In the sum-rate maximization problem of a $M$-user single-cell NOMA system with normalized channel gains $h_1 < h_2 < \dots < h_M$, and thus optimal (CNR-based) decoding order $M \to M-1 \to \dots \to 1$, the achievable spectral efficiency of user $i$ can be obtained by $R_{i}(\boldsymbol{p})= \log_2\left( 1+ \frac{p_{i} h_{i}} {\sum\limits_{j=i+1}^{M} p_{j} h_{i} + 1} \right)$. The sum-rate function $\sum\limits_{i=1}^{M} \log_2\left( 1+ \frac{p_{i} h_{i}} {\sum\limits_{j=i+1}^{M} p_{j} h_{i} + 1} \right)$ is monotonically increasing in $p_1$, i.e., the allocated power to user $1$ with the lowest decoding order. In fact, the signal of user $1$ with lowest decoding order is decoded and canceled by all the other users $i=2,\dots,M$. Due to the power constraint $\sum\limits_{i=1}^{M} p_i \leq P^{\text{max}}$, and the fact that $R_{1}(\boldsymbol{p})$ is monotonically increasing in $p_1$, at the optimal point $\boldsymbol{p}^*$, we have $p^*_1=P^{\text{max}} - \sum\limits_{i=2}^{M} p^*_i$. In other words, we ensure that at the optimal point, we have $\sum\limits_{i=1}^{M} p^*_i = P^{\text{max}}$, and the proof is completed.

\bibliographystyle{IEEEtran}
\bibliography{IEEEabrv,Bibliography}

\begin{thebibliography}{10}
\providecommand{\url}[1]{#1}
\csname url@samestyle\endcsname
\providecommand{\newblock}{\relax}
\providecommand{\bibinfo}[2]{#2}
\providecommand{\BIBentrySTDinterwordspacing}{\spaceskip=0pt\relax}
\providecommand{\BIBentryALTinterwordstretchfactor}{4}
\providecommand{\BIBentryALTinterwordspacing}{\spaceskip=\fontdimen2\font plus
\BIBentryALTinterwordstretchfactor\fontdimen3\font minus
  \fontdimen4\font\relax}
\providecommand{\BIBforeignlanguage}[2]{{%
\expandafter\ifx\csname l@#1\endcsname\relax
\typeout{** WARNING: IEEEtran.bst: No hyphenation pattern has been}%
\typeout{** loaded for the language `#1'. Using the pattern for}%
\typeout{** the default language instead.}%
\else
\language=\csname l@#1\endcsname
\fi
#2}}
\providecommand{\BIBdecl}{\relax}
\BIBdecl

\bibitem{10.5555/1146355}
T.~M. Cover and J.~A. Thomas, \emph{Elements of Information Theory (Wiley
  Series in Telecommunications and Signal Processing)}.\hskip 1em plus 0.5em
  minus 0.4em\relax USA: Wiley-Interscience, 2006.

\bibitem{NIFbook}
A.~E. Gamal and Y.-H. Kim, \emph{Network Information Theory}.\hskip 1em plus
  0.5em minus 0.4em\relax Cambridge University Press, 2011.

\bibitem{1683918}
H.~{Weingarten}, Y.~{Steinberg}, and S.~S. {Shamai}, ``The capacity region of
  the {Gaussian} multiple-input multiple-output broadcast channel,'' \emph{IEEE
  Transactions on Information Theory}, vol.~52, no.~9, pp. 3936--3964, 2006.

\bibitem{6692652}
Y.~{Saito}, Y.~{Kishiyama}, A.~{Benjebbour}, T.~{Nakamura}, A.~{Li}, and
  K.~{Higuchi}, ``Non-orthogonal multiple access {(NOMA)} for cellular future
  radio access,'' in \emph{Proc. IEEE 77th Vehicular Technology Conference (VTC
  Spring)}, 2013, pp. 1--5.

\bibitem{7973146}
Z.~{Ding}, X.~{Lei}, G.~K. {Karagiannidis}, R.~{Schober}, J.~{Yuan}, and V.~K.
  {Bhargava}, ``A survey on non-orthogonal multiple access for {5G} networks:
  Research challenges and future trends,'' \emph{IEEE Journal on Selected Areas
  in Communications}, vol.~35, no.~10, pp. 2181--2195, 2017.

\bibitem{7676258}
S.~M.~R. {Islam}, N.~{Avazov}, O.~A. {Dobre}, and K.~{Kwak}, ``Power-domain
  non-orthogonal multiple access {(NOMA)} in {5G} systems: Potentials and
  challenges,'' \emph{IEEE Communications Surveys \& Tutorials}, vol.~19,
  no.~2, pp. 721--742, 2017.

\bibitem{8010756}
W.~{Shin}, M.~{Vaezi}, B.~{Lee}, D.~J. {Love}, J.~{Lee}, and H.~V. {Poor},
  ``Non-orthogonal multiple access in multi-cell networks: Theory, performance,
  and practical challenges,'' \emph{IEEE Communications Magazine}, vol.~55,
  no.~10, pp. 176--183, 2017.

\bibitem{7272042}
P.~{Xu}, Z.~{Ding}, X.~{Dai}, and H.~V. {Poor}, ``A new evaluation criterion
  for non-orthogonal multiple access in {5G} software defined networks,''
  \emph{IEEE Access}, vol.~3, pp. 1633--1639, 2015.

\bibitem{9154358}
O.~{Maraqa}, A.~S. {Rajasekaran}, S.~{Al-Ahmadi}, H.~{Yanikomeroglu}, and S.~M.
  {Sait}, ``A survey of rate-optimal power domain {NOMA} with enabling
  technologies of future wireless networks,'' \emph{IEEE Communications Surveys
  \& Tutorials}, vol.~22, no.~4, pp. 2192--2235, 2020.

\bibitem{8823873}
M.~{Vaezi}, R.~{Schober}, Z.~{Ding}, and H.~V. {Poor}, ``Non-orthogonal
  multiple access: Common myths and critical questions,'' \emph{IEEE Wireless
  Communications}, vol.~26, no.~5, pp. 174--180, 2019.

\bibitem{7557079}
M.~S. {Ali}, H.~{Tabassum}, and E.~{Hossain}, ``Dynamic user clustering and
  power allocation for uplink and downlink non-orthogonal multiple access
  {(NOMA)} systems,'' \emph{IEEE Access}, vol.~4, pp. 6325--6343, 2016.

\bibitem{8352643}
M.~S. Ali, E.~Hossain, A.~Al-Dweik, and D.~I. Kim, ``Downlink power allocation
  for {CoMP-NOMA} in multi-cell networks,'' \emph{IEEE Transactions on
  Communications}, vol.~66, no.~9, pp. 3982--3998, Sept. 2018.

\bibitem{7982784}
J.~{Zhu}, J.~{Wang}, Y.~{Huang}, S.~{He}, X.~{You}, and L.~{Yang}, ``On optimal
  power allocation for downlink non-orthogonal multiple access systems,''
  \emph{IEEE Journal on Selected Areas in Communications}, vol.~35, no.~12, pp.
  2744--2757, 2017.

\bibitem{7964738}
Y.~{Fu}, Y.~{Chen}, and C.~W. {Sung}, ``Distributed power control for the
  downlink of multi-cell {NOMA} systems,'' \emph{IEEE Transactions on Wireless
  Communications}, vol.~16, no.~9, pp. 6207--6220, 2017.

\bibitem{8848606}
L.~{Lei}, L.~{You}, Y.~{Yang}, D.~{Yuan}, S.~{Chatzinotas}, and B.~{Ottersten},
  ``Load coupling and energy optimization in multi-cell and multi-carrier
  {NOMA} networks,'' \emph{IEEE Transactions on Vehicular Technology}, vol.~68,
  no.~11, pp. 11\,323--11\,337, 2019.

\bibitem{8114362}
Z.~{Yang}, C.~{Pan}, W.~{Xu}, Y.~{Pan}, M.~{Chen}, and M.~{Elkashlan}, ``Power
  control for multi-cell networks with non-orthogonal multiple access,''
  \emph{IEEE Transactions on Wireless Communications}, vol.~17, no.~2, pp.
  927--942, 2018.

\bibitem{8478339}
D.~{Ni}, L.~{Hao}, Q.~T. {Tran}, and X.~{Qian}, ``Transmit power minimization
  for downlink multi-cell multi-carrier {NOMA} networks,'' \emph{IEEE
  Communications Letters}, vol.~22, no.~12, pp. 2459--2462, 2018.

\bibitem{414651}
R.~D. {Yates}, ``A framework for uplink power control in cellular radio
  systems,'' \emph{IEEE Journal on Selected Areas in Communications}, vol.~13,
  no.~7, pp. 1341--1347, 1995.

\bibitem{7812683}
Y.~{Sun}, D.~W.~K. {Ng}, Z.~{Ding}, and R.~{Schober}, ``Optimal joint power and
  subcarrier allocation for full-duplex multicarrier non-orthogonal multiple
  access systems,'' \emph{IEEE Transactions on Communications}, vol.~65, no.~3,
  pp. 1077--1091, 2017.

\bibitem{7954630}
J.~{Zhao}, Y.~{Liu}, K.~K. {Chai}, A.~{Nallanathan}, Y.~{Chen}, and Z.~{Han},
  ``Spectrum allocation and power control for non-orthogonal multiple access in
  {HetNets},'' \emph{IEEE Transactions on Wireless Communications}, vol.~16,
  no.~9, pp. 5825--5837, 2017.

\bibitem{7862919}
A.~{Zappone}, E.~{Björnson}, L.~{Sanguinetti}, and E.~{Jorswieck}, ``Globally
  optimal energy-efficient power control and receiver design in wireless
  networks,'' \emph{IEEE Transactions on Signal Processing}, vol.~65, no.~11,
  pp. 2844--2859, 2017.

\bibitem{8809359}
K.~{Wang}, Y.~{Liu}, Z.~{Ding}, A.~{Nallanathan}, and M.~{Peng}, ``User
  association and power allocation for multi-cell non-orthogonal multiple
  access networks,'' \emph{IEEE Transactions on Wireless Communications},
  vol.~18, no.~11, pp. 5284--5298, 2019.

\bibitem{sourcecode}
S.~Rezvani and E.~A. Jorswieck, ``Optimal {SIC} ordering and power allocation
  in downlink multi-cell {NOMA} systems,''
  \url{https://gitlab.com/sepehrrezvani/Optimal-JSPA-MultiCell-NOMA.git}, 2021.

\bibitem{8861078}
W.~U. {Khan}, F.~{Jameel}, T.~{Ristaniemi}, S.~{Khan}, G.~A.~S. {Sidhu}, and
  J.~{Liu}, ``Joint spectral and energy efficiency optimization for downlink
  {NOMA} networks,'' \emph{IEEE Transactions on Cognitive Communications and
  Networking}, vol.~6, no.~2, pp. 645--656, 2020.

\bibitem{4453890}
Z.-Q. Luo and S.~Zhang, ``Dynamic spectrum management: Complexity and
  duality,'' \emph{IEEE Journal of Selected Topics in Signal Processing},
  vol.~2, no.~1, pp. 57--73, 2008.

\bibitem{Boydconvex}
S.~Boyd and L.~Vandenberghe, \emph{Convex Optimization}.\hskip 1em plus 0.5em
  minus 0.4em\relax Cambridge University Press, 2009.

\bibitem{8353846}
L.~{You}, D.~{Yuan}, L.~{Lei}, S.~{Sun}, S.~{Chatzinotas}, and B.~{Ottersten},
  ``Resource optimization with load coupling in multi-cell {NOMA},'' \emph{IEEE
  Transactions on Wireless Communications}, vol.~17, no.~7, pp. 4735--4749,
  2018.

\bibitem{8325426}
C.~{Liu} and D.~{Liang}, ``Heterogeneous networks with power-domain {NOMA}:
  Coverage, throughput, and power allocation analysis,'' \emph{IEEE
  Transactions on Wireless Communications}, vol.~17, no.~5, pp. 3524--3539,
  2018.

\bibitem{7587811}
L.~{Lei}, D.~{Yuan}, C.~K. {Ho}, and S.~{Sun}, ``Power and channel allocation
  for non-orthogonal multiple access in {5G} systems: Tractability and
  computation,'' \emph{IEEE Transactions on Wireless Communications}, vol.~15,
  no.~12, pp. 8580--8594, 2016.

\bibitem{8422362}
L.~{Salaün}, C.~S. {Chen}, and M.~{Coupechoux}, ``Optimal joint subcarrier and
  power allocation in {NOMA} is strongly {NP-Hard},'' in \emph{2018 IEEE
  International Conference on Communications (ICC)}, 2018, pp. 1--7.

\bibitem{9044770}
L.~Salaün, M.~Coupechoux, and C.~S. Chen, ``Joint subcarrier and power
  allocation in {NOMA}: Optimal and approximate algorithms,'' \emph{IEEE
  Transactions on Signal Processing}, vol.~68, pp. 2215--2230, 2020.

\bibitem{9451194}
B.~Clerckx, Y.~Mao, R.~Schober, E.~A. Jorswieck, D.~J. Love, J.~Yuan, L.~Hanzo,
  G.~Y. Li, E.~G. Larsson, and G.~Caire, ``Is {NOMA} efficient in multi-antenna
  networks? {A} critical look at next generation multiple access techniques,''
  \emph{IEEE Open Journal of the Communications Society}, vol.~2, pp.
  1310--1343, 2021.

\end{thebibliography}

\end{document}